\def \VersionLong {}
\def\VersionAuthor{}
\ifdefined \VersionLong
	\newcommand{\LongVersion}[1]{#1}
	\newcommand{\ShortVersion}[1]{}
\else
	\newcommand{\LongVersion}[1]{\ifdefined\VersionWithComments{\color{black!40}#1}\fi}
	\newcommand{\ShortVersion}[1]{\ifdefined\VersionWithComments{\color{red!40!black}#1}\else#1\fi}
\fi

\ifdefined \VersionAuthor
	\newcommand{\AuthorVersion}[1]{#1}
	\newcommand{\FinalVersion}[1]{}
\else
	\newcommand{\AuthorVersion}[1]{}
	\newcommand{\FinalVersion}[1]{#1}
\fi

\ifdefined \VersionAuthor
	\PassOptionsToPackage{svgnames,table}{xcolor}
	\documentclass[]{article}
\else

	\PassOptionsToPackage{svgnames,table}{xcolor}
	\documentclass[acmsmall]{acmart}
	\AtBeginDocument{%
	\providecommand\BibTeX{{%
		\normalfont B\kern-0.5em{\scshape i\kern-0.25em b}\kern-0.8em\TeX}}}

	\setcopyright{acmcopyright}
	\acmJournal{TOSEM}
	\acmYear{2022} \acmVolume{1} \acmNumber{1} \acmArticle{1} \acmMonth{1} \acmPrice{15.00}\acmDOI{10.1145/3517194}
\fi

\usepackage[utf8]{inputenc}
\usepackage[english]{babel}
\usepackage{csquotes}

\usepackage[ruled,vlined,linesnumbered]{algorithm2e}
	\SetKwInOut{Input}{input}
	\SetKwInOut{Output}{output}
\SetKw{KwPush}{push}
\SetKw{KwTo}{to}
\SetKw{KwPop}{pop}
\SetKw{KwFrom}{from}
\SetKw{KwBreak}{break}
\SetKw{KwCompute}{compute}

\usepackage{subcaption}
\usepackage{paralist} %
\usepackage{xspace}

\newenvironment{ienumeration}
	{\ifdefined\VersionAuthor\begin{enumerate}\else\begin{inparaenum}[\itshape i\upshape)]\fi}
	{\ifdefined\VersionAuthor\end{enumerate}\else\end{inparaenum}\fi}

\usepackage{amsmath} %
\AuthorVersion{
	\usepackage{amssymb} %
}
\usepackage{multirow}

\usepackage{graphicx}	%
\graphicspath{{figs/}}

\FinalVersion{
	\usepackage{natbib}
}

\ifdefined\VersionAuthor
	\usepackage[svgnames,table]{xcolor}

	\definecolor{USPNcobalt}{HTML}{293358}
	\definecolor{USPNocre}{HTML}{8b7d6d}
	\definecolor{USPNblanc}{HTML}{ffffff}
	\definecolor{USPNceruleen}{HTML}{354878}
	\definecolor{USPNsable}{HTML}{ad947e}

	\usepackage[
			pdfauthor={Waga, Andre, Hasuo},%
			pdftitle={Parametric Timed Pattern Matching},
			breaklinks  = true,
			colorlinks  = true,
		\ifdefined \VersionWithComments
		\fi
			citecolor   = USPNsable,
			linkcolor   = USPNocre,
			urlcolor    = USPNceruleen,
		]{hyperref}

\fi
	
\usepackage[capitalise,english,nameinlink]{cleveref} %
\crefname{line}{\text{line}}{\text{lines}} %
\crefname{assumption}{\text{Assumption}}{\text{Assumptions}} %

\ifdefined\VersionAuthor
	\usepackage[backend=biber,backref=true,style=alphabetic,url=false,doi=true,defernumbers=true,sorting=anyt,maxnames=99]{biblatex} %
	\DeclareSourcemap{
		\maps[datatype=bibtex]{
			\map[overwrite]{
				\step[fieldsource=toappear, match={true},
					fieldset=note, fieldvalue={To appear.~{}}, append]
			}
		}
	}

	\renewbibmacro*{doi+eprint+url}{%
		\iftoggle{bbx:doi}
			{\color{black!40}\footnotesize\printfield{doi}}
			{}%
		\newunit\newblock
		\iftoggle{bbx:eprint}
			{\usebibmacro{eprint}}
			{}%
		\newunit\newblock
		\iftoggle{bbx:url}
			{\usebibmacro{url+urldate}}
			{}%
	}

\fi
\usepackage{tikz}
\usetikzlibrary{arrows,automata,positioning,calc}
\tikzstyle{every node}=[initial text=]
\tikzstyle{location}=[rectangle, rounded corners, minimum size=12pt, draw=black, fill=blue!10, inner sep=2pt]
\tikzstyle{final}=[double]
\tikzstyle{accepting}=[final]
\tikzstyle{PTPMOPT}=[,dashed,color=red,semithick]

\newcommand{\styleact}[1]{\ensuremath{\textcolor{coloract}{\mathrm{#1}}}}
\newcommand{\styleclock}[1]{\ensuremath{\textcolor{colorclock}{#1}}}

\newcommand{\styleparam}[1]{\ensuremath{\textcolor{colorparam}{#1}}}

\definecolor{coloract}{named}{black}
\definecolor{colorclock}{named}{black}
\definecolor{colorconst}{named}{black}
\definecolor{colordisc}{named}{black}
\definecolor{colorloc}{rgb}{0.4, 0.4, 0.65}
\definecolor{colorparam}{named}{black}

\newif\iftikzgnuplot
\tikzgnuplottrue 
\usepackage{gnuplot-lua-tikz}
\usepackage{pgfplotstable}
\pgfplotsset{compat=1.12}
\usepackage{booktabs}

\newcommand{\init}{_0}

\newcommand{\A}{\ensuremath{\mathcal{A}}}
\newcommand{\Actions}{\Sigma}
\newcommand{\action}{\ensuremath{a}}
\newcommand{\actionStart}{\ensuremath{\styleact{start}}}
\newcommand{\actionEnd}{\ensuremath{\styleact{end}}}
\newcommand{\ActionsIndices}{\zeta}
\newcommand{\assign}{\leftarrow}
\newcommand{\BranchingCard}{B}

\newcommand{\Constraint}{C}

\newcommand{\Clock}{\mathbb{X}} %
\newcommand{\ClockCard}{H} %
\newcommand{\clock}{x} %
\newcommand{\clockabs}{\ensuremath{x_\mathit{abs}}} %
\newcommand{\clockval}{\nu} %
\newcommand{\ClocksZero}{\vec{0}}
\newcommand{\compOp}{\bowtie}
\newcommand{\compOpLeq}{\triangleleft}
\newcommand{\CTrue}{\mathbf{true}}

\newcommand{\edge}{e}
\newcommand{\Edges}{E}
\newcommand{\longuefleche}[1]{\stackrel{#1}{\longrightarrow}}
\newcommand{\longueflecheRel}[1]{\stackrel{#1}{\mapsto}}

\newcommand{\flecheRel}{{\rightarrow}}

\newcommand{\guard}{g}

\newcommand{\K}{K}

\newcommand{\KFalse}{\bot}
\newcommand{\Lg}{\mathcal{L}}
\newcommand{\loc}{\ell} %
\newcommand{\locinit}{\loc\init}
\newcommand{\Loc}{L} %
\newcommand{\LocFinal}{F}
\newcommand{\lterm}{\mathit{lt}}
\newcommand{\optvalue}{\ensuremath{o}}
\newcommand{\Param}{\mathbb{P}} %
\newcommand{\param}{p} %
\newcommand{\ParamCard}{M} %
\newcommand{\PVal}{({\setQplus})^{\Param}}
\newcommand{\pval}{v} %
\newcommand{\pvals}{V} %
\newcommand{\pvalsi}[1]{V_{#1}} %

\newcommand{\R}{{\mathbb{R}}}
\newcommand{\Rgeqzero}{\R_{\geq 0}}
\newcommand{\Rgzero}{{\R_{>0}}}
\newcommand{\setN}{{\mathbb N}}
\newcommand{\setQ}{{\mathbb Q}}
\newcommand{\setQplus}{\setQ_{+}} %
\newcommand{\setZ}{{\mathbb Z}}
\newcommand{\sinit}{s\init} %
\newcommand{\somelocs}{T} %
\newcommand{\concstate}{\ensuremath{s}} %
\newcommand{\States}{S} %
\newcommand{\TW}{\ensuremath{\mathcal{T}}}
\newcommand{\word}{\textcolor{colorok}{w}}
\newcommand{\words}{\textcolor{colorok}{W}}

\newcommand{\wloc}{w}

\newcommand{\resets}{R}
\newcommand{\project}[2]{\ensuremath{#1{\downarrow_{#2}}}}
\newcommand{\reset}[2]{\ensuremath{[#1]_{#2}}}
\newcommand{\valuate}[2]{\ensuremath{#2(#1)}}

\newcommand{\KMPSkipFunc}{\Delta_{\mathrm{KMP}}}
\newcommand{\QSSkipFunc}{\Delta_{\mathrm{QS}}}
\newcommand{\pat}{\mathit{pat}}
\newcommand{\str}{\word}
\newcommand{\Zp}{\setN_{>0}}%
\newcommand{\Rp}{{\mathbb{R}_{>0}}}
\newcommand{\Rnn}{{\mathbb{R}_{\ge 0}}}
\newcommand{\disjointUnion}{\sqcup}
\newcommand{\powerset}[1]{\mathcal{P}(#1)}

\usepackage{pifont}%
\newcommand{\cmark}{\ding{51}}%

\newcommand{\pmonaa}{${\mathtt{ParamMONAA}}$\xspace}
\newcommand{\monaa}{${\mathtt{MONAA}}$\xspace}
\newcommand{\CurrConf}{\mathit{CurrConf}}
\newcommand{\PrevConf}{\mathit{PrevConf}}

\newcommand{\eval}{\mathrm{eval}}

\newcommand{\rhoEmpty}{\rho_{\emptyset}}
\newcommand{\untimed}[1]{{\mathrm{Untimed}({#1})}}

\newcommand{\defProblem}[3]
{%
\noindent\fcolorbox{black}{blue!15}{
	\begin{minipage}{.95\columnwidth}
		\textbf{#1 problem:}\\
		\textsc{Input}: #2\\
		\textsc{Problem}: #3
	\end{minipage}
}
}

\newcommand{\cellHeaderColor}[1]{\cellcolor{blue!20}}
\newcommand{\cellHeader}[1]{\cellcolor{blue!20}\textbf{#1}}
\newcommand{\multiCellHeader}[2]{\multicolumn{#1}{c|}{\cellcolor{blue!20}\textbf{#2}}}
\newcommand{\startMultiCellHeader}[2]{\multicolumn{#1}{|c|}{\cellcolor{blue!20}\textbf{#2}}}

\AuthorVersion{
	\usepackage{amsthm}
}

\AuthorVersion{
	\theoremstyle{plain}
	\newtheorem{lemma}{Lemma}
	\newtheorem{theorem}{Theorem}
}
\newtheorem{assumption}{Assumption}

\AuthorVersion{
	\theoremstyle{definition}
	\newtheorem{definition}{Definition}
	\newtheorem{example}{Example}
}
\theoremstyle{remark}
\newtheorem{remark}{Remark}

\newcommand{\stylealgo}[1]{\ensuremath{\textsf{#1}}}
\newcommand{\EFsynth}{\stylealgo{EFsynth}}
\newcommand{\EFsynthOpt}{\stylealgo{OptParamSynth}}
\newcommand{\EFsynthMin}{\stylealgo{MinParamSynth}}
\newcommand{\EFsynthMax}{\stylealgo{MaxParamSynth}}
\newcommand{\PTPM}{\stylealgo{PTPM}}
\newcommand{\PTPMmax}{\ensuremath{\stylealgo{\PTPM}_{\textsf{max}}}}
\newcommand{\PTPMmin}{\ensuremath{\stylealgo{\PTPM}_{\textsf{min}}}}
\newcommand{\PTPMopt}{\ensuremath{\stylealgo{\PTPM}_{\textsf{opt}}}}
\newcommand{\TransPattern}{\stylealgo{MakeSymbolic}}
\newcommand{\TransWord}{\stylealgo{TW2PTA}}

\newcommand{\breach}{\textsc{Breach}}
\newcommand{\imitator}{\textsf{IMITATOR}}

\ifdefined \VersionWithComments
 	\definecolor{colorok}{RGB}{80,80,150}
\else
	\definecolor{colorok}{RGB}{0,0,0}
\fi

\newcommand{\eg}{\textcolor{colorok}{e.\,g.,}\xspace}
\newcommand{\ie}{\textcolor{colorok}{i.\,e.,}\xspace}
\newcommand{\st}{\textcolor{colorok}{s.t.}\xspace}
\newcommand{\viz}{\textcolor{colorok}{viz.,}\xspace}
\newcommand{\wrt}{\textcolor{colorok}{w.r.t.}\xspace}

\makeatletter
\AtBeginDocument{%
  \@ifpackageloaded{hyperref}
  {\def\@doi#1{\href{https://doi.org/#1}
      {\ttfamily https://doi.org/#1}\egroup}}
  {\def\@doi#1{\ttfamily https://doi.org/#1\egroup}}
  \def\doi{\bgroup\catcode`\_=12\relax\@doi}}
\makeatother

\ifdefined\VersionAuthor
	\newcommand{\orcidID}[1]{\href{https://orcid.org/#1}{\includegraphics[height=1em]{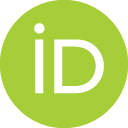}}}
\fi

\usepackage{version}	%

\newcommand{\ourAbstract}{%
	Given a log and a specification, timed pattern matching aims at exhibiting for which start and end dates a specification holds on that log.
	For example, ``a given action is always followed by another action before a given deadline''.
	This problem has strong connections with \emph{monitoring} real-time systems.
	We address here timed pattern matching in the presence of an \emph{uncertain} specification, \ie{} that may contain timing parameters (\eg{} the deadline can be uncertain or unknown).
	We want to know for which start and end dates, and for what values of the timing parameters, a property holds.
	For instance, we look for the minimum or maximum deadline (together with the corresponding start and end dates) for which the property holds.
	We propose two frameworks for \emph{parametric} timed pattern matching.
	The first one is based on parametric timed model checking.
	In contrast to most parametric timed problems, the solution is effectively computable.
	The second one is a \emph{dedicated} method; not only we largely improve the efficiency compared to the first method, but we further propose optimizations with skipping.
	Our experiment results suggest that our algorithms, especially the second one, are efficient and practically relevant.
}

\begin{document}

\title{Parametric Timed Pattern Matching\AuthorVersion{\footnote{%
	This is the author version of the manuscript of the same name published in ACM Transactions on Software Engineering and Methodology (Volume 32, Issue~1, 2023).
	The final version is available at \href{https://doi.org/10.1145/3517194}{\nolinkurl{10.1145/3517194}}.
}}}
\ifdefined\VersionAuthor
	\author{}
	\date{}
\else

	\author{Masaki Waga}
	\email{mwaga@fos.kuis.kyoto-u.ac.jp}
	\orcid{0000-0001-9360-7490}
	\affiliation{%
	\institution{Kyoto University}
	\streetaddress{Yoshida-honmachi, Sakyo-ku}
	\city{Kyoto}
	\postcode{606-8501}
	\country{Japan}
	}
	\author{\'Etienne Andr\'e}
	\orcid{0000-0001-8473-9555}
	\affiliation{%
	\institution{Université de Lorraine, CNRS, Inria, LORIA}
	\streetaddress{54\,506}
	\city{Vandœuvre-lès-Nancy}
	\country{France}
	}

	\author{Ichiro Hasuo}
	\email{i.hasuo@acm.org}
	\orcid{0000-0002-8300-4650}
	\affiliation{%
	\institution{National Institute of Informatics}
	\streetaddress{2-1-2 Hitotsubashi}
	\city{Tokyo}
	\postcode{101-8430}
	\country{Japan}}
	\additionalaffiliation{the Graduate University for Advanced Studies}

	\renewcommand{\shortauthors}{Masaki Waga, \'Etienne Andr\'e, and Ichiro Hasuo}

\begin{abstract}
	\ourAbstract{}
\end{abstract}



\fi

\maketitle

\ifdefined\VersionAuthor
	\noindent{}\textbf{Masaki Waga\orcidID{0000-0001-9360-7490}}
	\\
	{\em\small{}Kyoto University, Japan}

	\smallskip

	\noindent{}\textbf{Étienne André\orcidID{0000-0001-8473-9555}}
	\\
	{\em\small{}Université de Lorraine, CNRS, Inria, LORIA, Nancy, France}

	\smallskip

	\noindent{}\textbf{Ichiro Hasuo\orcidID{0000-0002-8300-4650}}
	\\
	{\em\small{}National Institute of Informatics, Japan}
	\\
	{\em\small{}The Graduate University for Advanced Studies, Japan}

	\begin{abstract}
		\ourAbstract{}
	\end{abstract}

\fi

\section{Introduction}\label{section:introduction}
Monitoring real-time systems consists in deciding whether a log satisfies a specification.
A problem of interest is to determine \emph{for which segment} of the log the specification is satisfied or violated.
This problem can be related to string matching and pattern matching.
The \emph{timed pattern matching problem} was formulated in~\cite{UFAM14}, with subsequent works varying the setting and improving the technique (\eg{} \cite{UFAM16,WAH16,AMNU17,WHS17}).
The problem takes as input a log and a specification, and decides where in the log the specification is satisfied or violated.
In~\cite{WAH16,WHS17}, we introduced a solution to the timed pattern matching problem where the log is given in the form of a timed word (a sequence of events with their associated timestamps), and the specification in the form of a timed automaton (TA), an extension of finite-state automata with clocks~\cite{AD94}.

\begin{example}\label{example:motivation}
\LongVersion{%
	As a motivating example, consider the example in \cref{figure:example}.
}%
Consider the automaton in \cref{figure:example:PTA}, and fix $\styleparam{\param_1} = 1$ and $\styleparam{\param_2} = 1$---which gives a timed automaton~\cite{AD94}.
Here \$ is a special terminal character.
For this timed automaton (say~$\A$) and the target timed word~$\word$ in \cref{figure:example:word}, the output of the timed pattern matching problem is the set of matching intervals $\{(t,t') \mid \word|_{(t,t')} \in \Lg(\A) \} = \{(t,t') \mid t \in [3.7,3.9), t' \in (6.0, \infty)\}$.
We note that in our semantics, the automaton is forced to take a transition for each event, and thus, any intervals starting before 3.7 does not match.
\end{example}

\begin{figure*}[t]
	\begin{subfigure}[b]{\textwidth}
	\centering
		\scalebox{1}{
		\begin{tikzpicture}[shorten >=1pt,node distance=2.5cm,on grid,auto]
		\node[location,initial] (s_0) {$\loc_0$};
		\node[location] (s_1) [right of=s_0] {$\loc_1$};
		\node[location] (s_2) [right of=s_1] {$\loc_2$};
		\node[location] (s_3) [right of=s_2]{$\loc_3$};
		\node[location,accepting] (s_4) [right of=s_3]{$\loc_4$};

		\path[->] 
                 (s_0) edge [above] node[above,align=center] {$\styleclock{x} > \styleparam{\param_1}$\\
							$\styleact{a}$}
                                    node[below] {$\styleclock{x} := 0$} (s_1)
		(s_1) edge [above] node[above,align=center] {$\styleclock{x} < \styleparam{\styleparam{\param_2}}$\\
                                         $\styleact{a}$}
                                   node[below] {$\styleclock{x} := 0$} (s_2)
		(s_2) edge [above] node[align=center] {$\styleclock{x} < \styleparam{\styleparam{\param_2}}$\\$\styleact{a}$} (s_3)
		(s_3) edge [above] node[align=center] {$\CTrue$\\$\styleact{\$}$} (s_4);
		\end{tikzpicture}}
	\caption{A parametric timed automaton}
	\label{figure:example:PTA}
	\end{subfigure}

	\begin{subfigure}[b]{0.99\textwidth}
	\centering
	\scalebox{1}{
	\begin{tikzpicture}[scale=1.2,xscale=1.5]
	\draw [thick, -stealth](-0.5,0)--(6.5,0) node [anchor=north]{$t$};
	\draw (0,0.1) -- (0,-0.1) node [anchor=north]{$0$};

	\draw (0.5,0.1) node[anchor=south]{$\styleact{a}$} -- (0.5,-0.1) node[anchor=north]{$0.5$};
	\draw (0.9,0.1) node[anchor=south]{$\styleact{a}$} -- (0.9,-0.1) node[anchor=north]{$0.9$};
	\draw (1.3,0.1) node[anchor=south]{$\styleact{b}$} -- (1.3,-0.1) node[anchor=north]{$1.3$};
	\draw (1.7,0.1) node[anchor=south]{$\styleact{b}$} -- (1.7,-0.1) node[anchor=north]{$1.7$};
	\draw (2.8,0.1) node[anchor=south]{$\styleact{a}$} -- (2.8,-0.1) node[anchor=north]{$2.8$};
	\draw (3.7,0.1) node[anchor=south]{$\styleact{a}$} -- (3.7,-0.1) node[anchor=north]{$3.7$};
	\draw (5.3,0.1) node[anchor=south]{$\styleact{a}$} -- (5.3,-0.1) node[anchor=north]{$5.3$};
	\draw (4.9,0.1) node[anchor=south]{$\styleact{a}$} -- (4.9,-0.1) node[anchor=north]{$4.9$};
	\draw (6.0,0.1) node[anchor=south]{$\styleact{a}$} -- (6.0,-0.1) node[anchor=north]{$6.0$};
	\end{tikzpicture}}
	\caption{A timed word}
	\label{figure:example:word}
	\end{subfigure}

	\caption{An example of parametric timed pattern matching \cite{WHS17}}
	\label{figure:example}
\end{figure*}

While the log is by definition concrete, it may happen that the specification is subject to uncertainty.
For example, we may want to detect \emph{cyclic patterns} with a period~$d$, without knowing the value of~$d$ with full certainty.
Therefore, the more abstract problem of \emph{parametric timed pattern matching} becomes of interest:
\textbf{given a (concrete) timed log and an incomplete specification where some of the timing constants may be known with limited precision or completely unknown, what are the time intervals and the valuations of the parameters for which the specification holds?}

Coming back to \cref{figure:example}, the question becomes to exhibit values for $\styleparam{t}, \styleparam{t'}, \styleparam{\param_1}, \styleparam{\param_2}$ for which the specification holds on the log, \ie{}
	$\{(t,t',\pval)\mid \word|_{(t,t')}\in \Lg(\valuate{\A}{\pval})\}$, where $\pval$ denotes a valuation of $\param_1, \param_2$ and $\valuate{\A}{\pval}$ denotes the replacement of $\styleparam{\param_1}, \styleparam{\param_2}$ in~$\A$ with their respective valuation in~$\pval$.
It is appreciated to conduct parametric timed pattern matching \emph{online}, \ie{} returning a partial result before obtaining the entire log.
An online algorithm for parametric timed pattern matching can synthesize parameter valuations \emph{in parallel} with the system execution under monitoring.
This is useful, for example because the synthesized parameter valuation can tell how severely the monitored system violates the timing constraints, and we can immediately react based on~it.

\begin{example}%
 \label{example:severity}
 \begin{figure}[tbp]
  \centering
  \begin{tikzpicture}[shorten >=1pt,node distance=2.5cm,on grid,auto]
   \node[location,initial] (s_0) {$\loc_0$};
   \node[location] (s_1) [right of=s_0] {$\loc_1$};
   \node[location] (s_2) [right of=s_1] {$\loc_2$};
   \node[location,accepting] (s_3) [right of=s_2]{$\loc_3$};

   \path[->] 
   (s_0) edge [above] node[above,align=center] {$\styleact{request}$}
   node[below] {$\styleclock{x} := 0$} (s_1)
   (s_1) edge [above] node[above,align=center] 
   {$\styleclock{x} > \styleparam{\param} \geq 10$\\
   $\styleact{approve}$} (s_2)
   (s_2) edge [above] node[align=center] {$\CTrue$\\$\styleact{\$}$} (s_3);
  \end{tikzpicture}
  \caption{A parametric timed automaton to monitor a sequence of requests and approvals}%
  \label{figure:request_approval_PTA}
 \end{figure}
 Consider monitoring a sequence of messages consisting of \styleact{request} and \styleact{approve}.
 We can detect too late approvals by the PTA in \cref{figure:request_approval_PTA}, which matches an approval at least 10 time units after the request.
 The severity of the violation, \ie{} the delay between \styleact{request} and \styleact{approve}, is captured by the parameter $\styleparam{\param}$.
\end{example}

\paragraph{Contribution}
In this work, we introduce two approaches addressing the parametric timed pattern matching problem.
Both use as underlying formalisms timed words and parametric timed automata (PTAs)~\cite{AHV93}, two well-known formalisms in the real-time systems community.
We show that the problem is decidable (which mainly comes from the fact that logs are finite).

Then, our algorithmic contribution is twofold.
First, we propose a practical solution based on parametric timed model checking.
We implement our method using \imitator{}~\cite{Andre21} and we perform a set of experiments on a set of automotive benchmarks.

Second, we propose a new \emph{dedicated} technique for performing efficient parametric timed pattern matching.
We propose optimizations based on \emph{skipping}, in the line of~\cite{WHS17}.
We implement our framework in a prototypical tool \pmonaa, we perform a set of experiments on a set of automotive benchmarks, and show that we increase the efficiency compared to our first approach.

Both algorithms are suitable for online monitoring, as they do not need the whole run to be executed, and experiments show that they are fast enough to be applied at runtime.

\paragraph{About this manuscript}\label{newtext:about}
This manuscript is an extension of~\cite{AHW18,WA19}, with~\cite{AHW18} describing the first method of this manuscript (based on parametric timed model checking) while~\cite{WA19} describes the second dedicated method.
We merged all concepts, and significantly improved the content, by notably
        modifying the problem formulation of~\cite{WA19} for the consistency,
	adding missing proofs,
	formalizing \TransPattern{} and \TransWord{} (in \cref{definition:TransPattern,definition:TransWord} respectively),
	adding a new theoretical result on the correctness of our first approach (\cref{ss:correctness}),
        almost rewriting the explanation of the skip values in Section 4.2 to make it more intuitive,
        adding new examples to illustrate the concept of the skip values in \cref{section:skipping},
	and
        adding \cref{section:comparison} to compare these two approaches.

\paragraph{Related work}%
Several algorithms have been proposed for online monitoring of real-time temporal logic specifications.
Online monitoring consists in monitoring on-the-fly at runtime, while offline monitoring is performed after the execution is completed, with less hard constraints on the monitoring algorithm performance.
An online monitoring algorithm for ptMTL (a past time fragment of MTL~\cite{Koymans90}) was proposed in~\cite{RFB14} and an algorithm for MTL[U,S] (a variant of MTL with both forward and backward temporal modalities) was proposed in~\cite{HOW14}.
In addition, a case study on an autonomous research vehicle monitoring~\cite{KCDK15} shows such procedures can be performed in an actual vehicle.

The approaches most related to ours are~\cite{UFAM14,UFAM16,Ulus17}.
In that series of works, logs are encoded by \emph{signals}, \ie{} values that vary over time.
This can be seen as a \emph{state-based} view, while our timed words are \emph{event-based}.
The formalism used for specification in~\cite{UFAM14,UFAM16} is timed regular expressions (TREs).
An offline monitoring algorithm is presented in~\cite{UFAM14} and an online one is \LongVersion{presented }in~\cite{UFAM16}.
These algorithms are implemented in the tool \emph{Montre}~\cite{Ulus17}.
In~\cite{BFNMA18}, the setting is signals matched against a temporal pattern; the construction is automata-based as in~\cite{WAH16,WHS17}.

Some algorithms have also been proposed for parameter identification of a temporal logic specification with uncertainty over a log.
In the discrete time setting, an algorithm for an extension of LTL is proposed in~\cite{FR08}; and in the real-time setting, algorithms for parametric signal temporal logic (PSTL) are proposed in~\cite{ADMN11,JTSSS17,BFM18,BBMJ19}.
Although these works are related to our approach, previous approaches do not focus on segments of a log but on a whole log.
In contrast, we exhibit intervals together with their associated parameter valuations, in a fully symbolic fashion.
We believe our matching-based setting is advantageous in many usage scenarios \eg{} from hours of a log of a car, extracting timing constraints of a certain actions to cause slipping.
Also, our setting allows the patterns with complex timing constraints (see the pattern in \cref{figure:patterns:blowup} for example).
Another related research direction is \emph{specification mining}.
In~\cite{NCJF18,SNF17,NF19}, the authors propose algorithms to mine a TRE formula from a TRE \emph{template} and a set of logs.
Beyond the fact that we used a different formalism (an extension of timed automata vs.\ TRE),
the main difference with ours is if the timing constraints are mined by a \emph{statistical} approach or a \emph{symbolic} approach:
in~\cite{NF19}, timing constraints are mined by clustering while we use symbolic analysis of polyhedra.\label{paragraph:specification_mining}

In~\cite{BFMU17}, an offline algorithm for the robust pattern matching problem is considered over signal regular expressions, consisting in computing the \emph{quantitative} (robust) semantics of a signal relative to an expression.
For piecewise-constant and piecewise-linear signals, the problem can be effectively solved using a finite union of zones (linear constraints with a special form~\cite{BY03}).
In~\cite{Waga19}, an online algorithm for the robust pattern matching problem is presented over timed symbolic weighted automata, where the semantics is generalized with semiring. %
The algorithm in~\cite{Waga19} effectively solves the problem for piecewise-constant signals.

In~\cite{WAH19}, we proposed a symbolic monitoring algorithm against specifications parametric in time and data, \ie{} we considered logs extended with data as the log formalism, and parametric timed data automata (an \emph{ad-hoc} extension of timed automata) as the specification formalism.

In~\cite{WAH21}, we proposed the \emph{model-bounded monitoring} scheme to use prior knowledge about the monitored system in the interpolation of the observation obtained by discrete sampling. One of the procedures in~\cite{WAH21} is by reduction to the reachability analysis of a linear hybrid automaton~\cite{HPR94}.
The idea of the reduction is essentially the same as ours in \cref{section:PTMC}.

Further works attempted to quantify the distance between a specification and a signal temporal logic (STL) specification (\eg{} \cite{DFM13,DMP17,JBGNN18}).
The main difference with our work is that these works compute a distance \wrt{} to a whole log, while we aim at exhibiting where in the log is the property satisfied; our notion of parameters can also be seen as a relative time distance.
However, our work is closer to the robust satisfaction of guards rather than signal values; in that sense, our contribution is more related to the time robustness in~\cite{DM10} or the distance in~\cite{ABD18}.

Finally, while our work is related to parameter synthesis, in the sense that we identify parameter valuations in the property such that it holds (or not), the term ``parameter synthesis'' is also used in monitoring with a slightly different meaning: given a \emph{model} with parameters, the goal is to find parameters that maximize the robustness of the specification, \ie{} satisfying behaviors for a range of parameters for which the model robustly satisfies the property.
A notable tool achieving this is \breach{}~\cite{Donze10}.

\begin{table*}[t]
	\centering
	
	\caption{Matching problems}

	\scalebox{.8}{%
	\begin{tabular}{c||c|c|c}
		&log, target&specification, pattern &output\\\hline\hline
		string matching&
		a word $\str\in \Sigma^{*}$
		&
		a word $\pat\in\Sigma^{*}$
		&
		$\{(i,j)\in(\Zp)^{2}\mid \str(i,j)=\pat\}$
		\\\hline
		pattern matching (PM) &
		a word $\str\in \Sigma^{*}$
		&
		an NFA $\A$
		&
		$\{(i,j)\in(\Zp)^{2}\mid \str(i,j)\in \Lg(\A)\}$
		\\\hline
		timed PM&
		a timed word $\str\in(\Sigma \times \Rgzero)^{*}$
		&
		a TA $\A$
		&
		$\{(t,t') \in (\Rgzero)^{2}\mid \word|_{(t,t')} \in \Lg(\A)\}$
		\\\hline
		parametric timed PM&
		a timed word $\str\in(\Sigma \times \Rgzero)^{*}$
		&
		a PTA $\A$
		&
		$\{(t,t', \pval) \mid \word|_{(t,t')} \in \Lg(\valuate{\A}{\pval})\}$
		\\\hline
	\end{tabular}
	}
	
	\label{table:matching}
\end{table*}

A summary of various matching problems is recalled in \cref{table:matching}.

\paragraph{Outline}
We introduce the necessary definitions and state our main objective in \cref{section:preliminaries}.
We present and evaluate our method based on parametric timed model checking in \cref{section:PTMC}.
We then present and evaluate a \emph{dedicated} method, enhanced with automata-based skipping, in \cref{section:adhoc}.
In \cref{section:comparison}, we discuss the comparison between the approaches in \cref{section:PTMC,section:adhoc}.
We conclude in \cref{section:conclusion}.

\section{Preliminaries and objective}\label{section:preliminaries}

Our target strings are \emph{timed words}~\cite{AD94}, that are time-stamped words over an alphabet $\Actions$.
Our patterns are given by parametric timed automata~\cite{AHV93}.

\LongVersion{
\subsection{Timed words and timed segments}\label{subsection:timed_word}
}

For an alphabet $\Actions$, a \emph{timed word} is a
sequence $\word$ of pairs $(\action_i,\tau_i) \in (\Actions \times \Rgzero)$
satisfying $\tau_i \leq \tau_{i + 1}$ for any $i \in \{1,2,\dots,|\word|-1\}$.
We let $\tau_0=0$.
For an alphabet $\Actions$, we denote the set of the timed words on $\Actions$ by $\TW(\Actions)$.
For an alphabet $\Actions$ and $n \in\Zp$, we denote the set of the timed words of length $n$ on $\Actions$ by $\TW^n(\Actions)$.
Given a timed word $\word$, we often denote it by $(\overline{a},\overline{\tau})$, where $\overline{a}$ is the sequence $(\action_1, \action_2, \cdots)$ and $\overline{\tau}$ is the sequence $(\tau_1, \tau_2, \cdots)$.
Let $\word = (\overline{a},\overline{\tau})$ be a timed word.
We denote the subsequence $(\action_i, \tau_i),(\action_{i+1},
\tau_{i+1}),\cdots,(\action_j,\tau_j)$ by $\word (i,j)$.
For $t \in \R$ such that $- \tau_1 < t$, the \emph{$t$-shift} of $\word$ is
$(\overline{a}, \overline{\tau}) + t = (\overline{a}, \overline{\tau} +
t)$ where
$\overline{\tau} + t = \tau_1 + t,\tau_2 + t,\cdots, \tau_{|\tau|} + t$.
For timed words $\word = (\overline{a},\overline{\tau})$ and
$\word' = (\overline{a'},\overline{\tau'})$,
their \emph{absorbing concatenation} is $\word \circ \word' = (\overline{a} \circ \overline{a'}, \overline{\tau} \circ \overline{\tau'})$ where $\overline{a} \circ \overline{a'}$ and  $\overline{\tau} \circ \overline{\tau'}$ are usual concatenations, and their \emph{non-absorbing concatenation} is $\word \cdot \word' = \word \circ (\word' + \tau_{|\word|})$.
The concatenations of subsets of $\TW(\Actions)$, \ie{} sets of timed words, are also defined similarly. %
For a set $\words\subseteq\TW(\Actions)$ of timed words,%
 its untimed projection $\untimed{\words}\in\Actions^*$ is $\{\overline{a}\mid (\overline{a},\overline{\tau})\in\words\}$.

For a timed word $\word = (\overline{a}, \overline{\tau})$ on $\Actions$
and 
$t,t' \in \Rgeqzero$ satisfying $t < t'$, a \emph{timed word segment}
$\word|_{(t,t')}$ is defined by the timed word $(\word (i,j) - t) \circ (\$,t' - t)$
on the augmented alphabet $\Actions \disjointUnion \{\$\}$,
where $i,j$ are chosen so that
$\tau_{i-1} < t \leq \tau_i$ and
$\tau_{j} \leq t' < \tau_{j+1}$.
Here the fresh symbol \(\$\) is called the \emph{terminal character}.
\subsection{Clocks, parameters and guards}\label{ss:guards}

We assume a set~$\Clock = \{ \clock_1, \dots, \clock_\ClockCard \} $ of \emph{clocks} (where $\ClockCard \in \setN$ denotes the clocks set cardinality), \ie{} real-valued variables that evolve at the same rate.
A clock valuation is\LongVersion{ a function}
$\clockval : \Clock \rightarrow \Rgeqzero$.
We write $\ClocksZero$ for the clock valuation assigning $0$ to all clocks.
Given $d \in \Rgeqzero$, $\clockval + d$ \ShortVersion{is}\LongVersion{denotes the valuation} \st{} $(\clockval + d)(\clock) = \clockval(\clock) + d$, for all $\clock \in \Clock$.
Given $\resets \subseteq \Clock$, we define the \emph{reset} of a valuation~$\clockval$, denoted by $\reset{\clockval}{\resets}$, as follows: $\reset{\clockval}{\resets}(\clock) = 0$ if $\clock \in \resets$, and $\reset{\clockval}{\resets}(\clock)=\clockval(\clock)$ otherwise.

We assume a set~$\Param = \{ \styleparam{\param_1}, \dots, \styleparam{\param_\ParamCard} \} $ of \emph{parameters}\LongVersion{, \ie{} unknown constants}.
A parameter {\em valuation} $\pval$ is\LongVersion{ a function}
$\pval : \Param \rightarrow \setQplus$.
We assume ${\compOp} \in \{<, \leq, =, \geq, >\}$.
A guard~$\guard$ is a constraint over $\Clock \cup \Param$ defined by a conjunction of inequalities of the form $\clock \compOp d$, 
or $\clock \compOp \param$ with $d \in \setN$ and $\param \in \Param$.
Given~$\guard$, we write~$\clockval\models\pval(\guard)$ if %
the expression obtained by replacing each~$\clock$ with~$\clockval(\clock)$ and each~$\param$ with~$\pval(\param)$ in~$\guard$ evaluates to true.

A linear term over $\Clock \cup \Param$ is of the form $\sum_{1 \leq i \leq \ClockCard} \alpha_i \clock_i + \sum_{1 \leq j \leq \ParamCard} \beta_j \param_j + d$, with
	$\clock_i \in \Clock$,
	$\param_j \in \Param$,
	and
	$\alpha_i, \beta_j, d \in \setZ$.
A \emph{constraint}~$\Constraint$ (\ie{} a convex polyhedron) over $\Clock \cup \Param$ is a conjunction of inequalities of the form $\lterm \compOp 0$, where $\lterm$ is a linear term.
Given a set~$\Param$ of parameters, we denote by $\project{\Constraint}{\Param}$ the projection of~$\Constraint$ onto~$\Param$, \ie{} obtained by eliminating the variables not in $\Param$ (\eg{} using Fourier-Motzkin\LongVersion{~\cite{Schrijver86}}).
$\KFalse$ denotes the constraint over~$\Param$ representing the empty set of parameter valuations.

\subsection{Parametric timed automata}

Parametric timed automata (PTAs) extend timed automata with parameters within guards %
 in place of integer constants~\cite{AHV93}.
 
\subsubsection{Syntax}

\begin{definition}[PTA]\label{def:uPTA}
	A parametric timed automaton (PTA) $\A$ is a tuple \mbox{$\A = (\Actions, \Loc, \locinit, \LocFinal, \Clock, \Param, %
\Edges)$}\footnote{Following~\cite{WHS17}, we do not allow invariants in the locations. Allowing invariants would be straightforward.}, where:
	\begin{enumerate}
		\item $\Actions$ is a finite set of actions,
		\item $\Loc$ is a finite set of locations,
		\item $\locinit \in \Loc$ is the initial location,
		\item $\LocFinal \subseteq \Loc$ is the set of final locations,
		\item $\Clock$ is a finite set of clocks,
		\item $\Param$ is a finite set of parameters,
		\item $\Edges$ is a finite set of edges  $\edge = (\loc, \guard, \action, \resets, \loc')$
		where~$\loc, \loc' \in \Loc$ are the source and target locations, $\action \in \Actions$, $\resets\subseteq \Clock$ is a set of clocks to be reset, and $\guard$ is a guard.
	\end{enumerate}
\end{definition}

Given\LongVersion{ a parameter valuation}~$\pval$, we denote by $\valuate{\A}{\pval}$ the non-parametric structure where all occurrences of a parameter~$\param_i$ have been replaced by~$\pval(\param_i)$.
In this work, we refer as a \emph{timed automaton} to any structure $\valuate{\A}{\pval}$.\footnote{%
	Technically, a timed automaton requires non-negative integer constants, while we define non-negative \emph{rational} valuations.
	So a TA can be obtained by assuming a rescaling of the constants: by multiplying all constants in $\valuate{\A}{\pval}$ by the least common multiple of their denominators, we obtain an equivalent (integer-valued) TA, as defined in \cite{AD94}.
}

\subsubsection{Synchronous product}

The synchronous product (using strong broadcast, \ie{} synchronization on shared actions) of several PTAs gives a PTA.

\begin{definition}\label{def:product}[synchronized product of PTAs]
	Let $N \in \setN$.
	Given a set of PTAs $\A_i = (\Actions_i, \Loc_i, {(\locinit)}_i, \LocFinal_i, \Clock_i, \Param_i, %
 \Edges_i)$, $1 \leq i \leq N$,
	the \emph{synchronized product} of $\A_i$, $1 \leq i \leq N$,
	denoted by $\A_1 \parallel \A_2 \parallel \cdots \parallel \A_N$,
	is the tuple
		$(\Actions, \Loc, \locinit, \LocFinal, \Clock, \Param, %
 \Edges)$, where:
	\begin{enumerate}
		\item $\Actions = \bigcup_{i=1}^N\Actions_i$,
		\item $\Loc = \prod_{i=1}^N \Loc_i$,
		\item $\locinit = ({(\locinit)}_{1}, \dots, {(\locinit)}_{N})$,
		\item $\LocFinal = \LocFinal_1 \times \dots \times \LocFinal_{N}$,
		\item $\Clock = \bigcup_{1 \leq i \leq N} \Clock_i$,
		\item $\Param = \bigcup_{1 \leq i \leq N} \Param_i$,
	\end{enumerate}
	and $\Edges{}$ is defined as follows.
	For all $\action \in \Actions$,
	let $\ActionsIndices_\action$ be the subset of indices $i \in 1, \dots, N$
	such that $\action \in \Actions_i$.
	For all  $\action \in \Actions$,
	for all $(\loc_1, \dots, \loc_N) \in \Loc$,
	for all \mbox{$(\loc_1', \dots, \loc_N') \in \Loc$},
	$\big((\loc_1, \dots, \loc_N), \guard, \action, \resets, (\loc'_1, \dots, \loc'_N)\big) \in \Edges$
	if:
	\begin{itemize}
		\item for all $i \in \ActionsIndices_\action$, there exist $\guard_i, \resets_i$ such that $(\loc_i, \guard_i, \action, \resets_i, \loc_i') \in \Edges_i$, $\guard = \bigwedge_{i \in \ActionsIndices_\action} \guard_i$, $\resets = \bigcup_{i \in \ActionsIndices_\action}\resets_i$, and,
		\item for all $i \not\in \ActionsIndices_\action$, $\loc_i' = \loc_i$.
\end{itemize}
\end{definition}

Note that we define the set of accepting locations to be the \emph{product} of the individual accepting locations, \ie{} a location of the product automaton is accepting if all its component locations are accepting.

\subsubsection{Concrete semantics}

Let us now recall the concrete semantics of TAs.

\begin{definition}[Semantics of a TA]
	Given a PTA $\A = (\Actions, \Loc, \locinit, \LocFinal, \Clock, \Param, %
\Edges)$ 
        with clocks~$\Clock = \{ \clock_1, \dots, \clock_\ClockCard \}$,
	and a parameter valuation~\(\pval\),
	the semantics of $\valuate{\A}{\pval}$ is given by the timed transition system (TTS) $(\States, \sinit, \flecheRel)$, with
	\begin{itemize}
		\item $\States = \Loc \times \Rgeqzero^\ClockCard$ %
		\item $\sinit = (\locinit, \ClocksZero) $,
		\item  $\flecheRel$ consists of the discrete and (continuous) delay transition relations:
		\begin{ienumeration}
			\item discrete transitions: $(\loc, \clockval) \longueflecheRel{\edge} (\loc', \clockval')$, %
				if
				there exists $\edge = (\loc, \guard, \action, \resets, \loc') \in \Edges$, such that $\clockval' = \reset{\clockval}{\resets}$, and $\clockval \models \pval(\guard$).
			\item delay transitions: $(\loc, \clockval) \longueflecheRel{d} (\loc, \clockval + d)$, with $d \in \Rgeqzero$.
		\end{ienumeration}
	\end{itemize}
\end{definition}

    Moreover we write $(\loc, \clockval)\longuefleche{(d, \edge)} (\loc', \clockval')$ for a combination of a delay and discrete transition if
		$\exists  \clockval'' :  (\loc, \clockval) \longueflecheRel{d} (\loc, \clockval'') \longueflecheRel{\edge} (\loc', \clockval')$.

Given a TA~$\valuate{\A}{\pval}$ with concrete semantics $(\States, \sinit, \flecheRel)$, we refer to the states of~$\States$ as the \emph{concrete states} of~$\valuate{\A}{\pval}$.
A \emph{run} of~$\valuate{\A}{\pval}$ is an alternating sequence of concrete states of $\valuate{\A}{\pval}$ and pairs of edges and delays starting from the initial state $\sinit$ of the form
$\concstate_0, (d_0, \edge_0), \concstate_1, (d_1, \edge_1), \cdots$
with
$i = 0, 1, \dots$, $\edge_i \in \Edges$, $d_i \in \Rgeqzero$ and
	$\concstate_i \longuefleche{(d_i, \edge_i)} \concstate_{i+1}$.
Given such a run, the associated \emph{timed word} is $(\action_1, \tau_1), (\action_2, \tau_2), \cdots$, where $\action_i$ is the action of edge~$\edge_{i-1}$, and $\tau_i = \sum_{0 \leq j \leq i-1} d_j$, for $i = 1, 2 \cdots$.\footnote{%
	The ``$-1$'' in indices comes from the fact that, following usual conventions in the literature, states are numbered starting from~0 while words are numbered from~1.
}
Given\LongVersion{ a state}~$\concstate=(\loc, \clockval)$, we say that $\concstate$ is reachable in~$\valuate{\A}{\pval}$ if $\concstate$ appears in a run of $\valuate{\A}{\pval}$.
By extension, we say that $\loc$ is reachable; and by extension again, given a set~$\somelocs$ of locations, we say that $\somelocs$ is reachable if there exists $\loc \in \somelocs$ such that $\loc$ is reachable in~$\valuate{\A}{\pval}$.

A finite run is \emph{accepting} if its last state $(\loc, \clockval)$ is such that $\loc \in \LocFinal$.
The (timed) \emph{language} $\Lg(\valuate{\A}{\pval})$ is defined to be the set of timed words associated with all accepting runs of~$\valuate{\A}{\pval}$.
\subsection{Reachability synthesis}

We use here reachability synthesis for the following two purposes.
\begin{itemize}
 \item In \cref{section:PTMC}: to solve parametric timed pattern matching
 \item In \cref{section:adhoc}: to improve the dedicated parametric timed pattern matching algorithm (\cref{alg:online_no_skip}) with a skipping optimization
\end{itemize}
This procedure, called \EFsynth{}, takes as input a PTA~$\A$ and a set of target locations~$\somelocs$, and attempts to synthesize all parameter valuations~$\pval$ for which~$\somelocs$ is reachable in~$\valuate{\A}{\pval}$.
\EFsynth{} was formalized in \eg{} \cite{JLR15} and is a procedure that may not terminate, but that computes an exact result (sound and complete) if it terminates.
\EFsynth{} traverses the \emph{parametric zone graph} of~$\A$, which is a potentially infinite extension of the well-known zone graph of TAs (see, \eg{} \cite{ACEF09,JLR15}\LongVersion{ for a formal definition}).

\paragraph{Optimal parameter reachability synthesis}\label{newtext:EFsynthopt}
In addition, for the specific case of pattern matching with optimization (\cref{ss:optimization}), we will make use of \emph{optimal parameter reachability synthesis}~\cite{ABPP19}.
This procedure, called \EFsynthOpt{}, takes as input a PTA~$\A$, a set of target locations~$\somelocs$ and a parameter~$\param$ to minimize, and attempts to synthesize all parameter valuations~$\pval$ for which~$\somelocs$ is reachable in~$\valuate{\A}{\pval}$ \emph{and} for which the value of~$\param$ is optimal; \EFsynthOpt{} is a generic algorithm, that can be instantiated to \EFsynthMin{} (where $\param$ should be minimized) or \EFsynthMax{} (where $\param$ should be maximized).
\EFsynthMin{} was studied in~\cite{ABPP19}, and is essentially similar to~\EFsynth{}, with a condition to only keep the states leading to a minimal value of~$\param$.

\subsection{Parametric timed pattern matching}

Let us recall timed pattern matching~\cite{WAH16,WHS17,WHS18}.%

\smallskip

\defProblem
	{Timed pattern matching}
	{a timed word~$\str$ over an alphabet $\Actions$ and a TA~$\A$ over the augmented alphabet $\Actions \disjointUnion \{\$\}$}
	{compute all the intervals $(t,t')$ for which the segment  $\word|_{(t,t')}$ is 
	accepted by $\A$.
	That is, it requires
	the \emph{match set} $\mathcal{M}
	(\word,\A) = \{(t,t') \mid \word|_{(t,t')} \in \Lg(\A)\}$.}

\medskip

The match set $\mathcal{M} (\word,\A)$ is in general uncountable.
For example, if $\A$ accepts any timed word, the match set is $\{(t,t') \in \Rnn \times \Rnn \mid t < t' \}$, which is uncountable.
However it allows finite representation, as a finite union of special polyhedra called \emph{zones} (see~\cite{BY03,WAH16}).
Roughly speaking, zones are made of constraints of the form $\clock \compOpLeq c$ or $\clock - \clock' \compOpLeq c$, with $\clock, \clock' \in \Clock$, $c \in \setZ$ and $\compOpLeq \in \{ \leq, < \}$.

We now extend the timed pattern matching problem to parameters by allowing a specification expressed using PTAs.
The problem now requires not only the start and end dates for which the property holds, but also the associated parameter valuations.

\smallskip

\defProblem%
	{Parametric timed pattern matching}
	{a timed word~$\str$ over an alphabet $\Actions$ and a PTA~$\A$ over the augmented alphabet $\Actions \disjointUnion \{\$\}$}
	{compute all the triples $(t,t', \pval)$ for which the segment  $\word|_{(t,t')}$ is accepted by $\valuate{\A}{\pval}$.
	That is, it requires the \emph{match set} $\mathcal{M} (\word,\A) = \{(t,t', \pval) \mid \word|_{(t,t')} \in \Lg(\valuate{\A}{\pval})\}$.}

\medskip

Since parametric timed pattern matching is a generalization of (non-parametric) timed pattern matching, the match set $\mathcal{M} (\word,\A)$ is again in general uncountable; however,
we will see that it can still be represented as a finite union of polyhedra, but in more dimensions, \viz{} $|\Param| + 2$, \ie{} the number of parameters + 2 further dimensions for~$t$ and~$t'$.
In addition, the form of the obtained polyhedra is more general than zones, as parameters may ``accumulate'' to produce sums of parameters with coefficients (\eg{} $3 \times \param_1 < \param_2 + 2 \times \param_3$).

\begin{example}%
 \label{example:PTPM}
 \begin{figure}[t]
 \begin{subfigure}[b]{0.60\textwidth}
  \centering
  \footnotesize
   \begin{tikzpicture}[shorten >=1pt,node distance=1.5cm,on grid,auto] 
   \node[location,initial] (s_0)  {$\locinit$}; 
   \node[location,node distance=2.5cm] (s_1) [right=of s_0] {$\loc_1$}; 
   \node[location,node distance=2.5cm] (s_2) [right=of s_1] {$\loc_2$};
   \node[location,accepting,node distance=1.5cm] (s_3) [right=of s_2] {\cmark};
   \path[->] 
   (s_0) edge  [above] node[align=center] {$\styleact{a}$\\$\styleclock{x} > 1$} (s_1)
   (s_1) edge  [above] node[above,align=center] {$\styleact{a}$\\$\styleclock{x} < \styleparam{\param}$} node[below]{$\styleclock{x}:=0$}(s_2) 
   (s_2) edge  [above] node[align=center] {$\styleact{\$}$\\$\styleclock{x}<1$} (s_3);
   \end{tikzpicture}
   \caption{the input PTA $\A$}%
  \label{figure:running-example:PTA}
  \begin{tikzpicture}[scale=1.8,xscale=1.5]
   \draw [thick, -stealth](-0.5,0)--(2.25,0) node [anchor=north]{$t$};
   \draw (0,0.1) -- (0,-0.1) node [anchor=north]{$0$};

   \draw (0.35,0.1) node[anchor=south]{$\styleact{a}$} -- (0.35,-0.1) node[anchor=north]{$0.7$};
   \draw (1.0,0.1) node[anchor=south]{$\styleact{a}$} -- (1.0,-0.1) node[anchor=north]{$2.0$};
   \draw (2.05,0.1) node[anchor=south]{$\styleact{a}$} -- (2.05,-0.1) node[anchor=north]{$4.1$};
  \end{tikzpicture}
  \caption{the input timed word $\word$}
 \end{subfigure}
 \begin{subfigure}[b]{0.39\textwidth}
  \begin{displaymath}
   \begin{cases}
    0.7 < t < 1.0\\
    4.1 < p + t\\
    4.1 \leq t' < 5.1
   \end{cases}
  \end{displaymath}
  \caption{the convex polyhedron representing the match set $\mathcal{M}(\word, \A)$}
  \label{figure:running-example:match_set}
 \end{subfigure}
 \caption{Example of parametric timed pattern matching}
 \label{figure:running-example}
 \end{figure}
 \cref{figure:running-example} shows an example of parametric timed pattern matching.
 Given the PTA $\A$ and the timed word $\word$, the parametric timed pattern matching problem asks for the match set $\mathcal{M}(\word,\A)$ in \cref{figure:running-example:match_set}.

The PTA $\A$ shows that there must be at least 1 time unit between the beginning of the matching $t$ and the first action $\styleact{a}$ in the matching. Moreover, each matching consists of two actions $\styleact{a}$.
Therefore, any matching must begin after 0.7 and at least 1 time unit before the action at time 2.0, and we have $0.7 < t < 2.0 - 1.0 (= 1.0)$.
At location $\loc_1$, the value of the clock $x$ is the elapsed time from the beginning of the matching.
Therefore, at time $T$, the value of $\styleclock{x}$ is $T - t$.
Since we observe an action $\styleact{a}$ at 4.1 and we have $\styleclock{x} < \styleparam{p}$ at the edge from $\loc_1$ to $\loc_2$, we have $4.1 - t < p$.
Finally, $\A$ shows that the matching must end within 1 time unit after the second action $\styleact{a}$ in the matching. 
Therefore, we have $4.1 \leq t' < 5.1$.

 We observe that the convex polyhedron in the match set has three dimensions for the parameter $p$ and for $t$ and $t'$.
 We also observe that this convex polyhedron is not a zone because it contains the constraint $4.1 < p + t$.
\end{example}

In the rest of this paper, we make the following assumption on the PTA in parametric timed pattern matching to simplify the construction and the reasoning. This assumption is easy to remove in practice since we can simply remove the transitions that do not satisfy the second part of the assumption.

\begin{assumption}%
 \label{assumption:transitions}
 As in~\cite{WAH16,WHS17}, we assume that all transitions to the accepting locations are labeled with~$\$$, and actions that are not part of the timed word alphabet are removed, except for the special action~$\$$.
\end{assumption}

\section{Parametric timed pattern matching using model checking}\label{section:PTMC}
\subsection{General approach}\label{ss:general}

In addition to \cref{assumption:transitions}, 
we make the following assumption, that does not impact the correctness of our method, but simplifies the subsequent reasoning.

\begin{assumption}%
 \label{assumption:final}
 We assume that the pattern automaton contains a single accepting location.
\end{assumption}

\cref{assumption:final} is easy to remove in practice: if the pattern PTA contains more than one accepting location, they can be merged into a single accepting location.

We show using the following approach that parametric timed pattern matching can reduce to parametric reachability analysis.

\begin{enumerate}
	\item We turn the pattern into a symbolic pattern, by allowing it to start anytime.
		In addition, we use two parameters to measure the (symbolic) starting time and the (symbolic) ending time of the pattern.
	\item We turn the timed word into a (non-parametric) timed automaton that uses a single clock $\clockabs$ measuring the absolute time.
	\item We consider the synchronized product of the symbolic pattern PTA and the timed word (P)TA.
	\item We run the reachability synthesis algorithm \EFsynth{} to derive all possible parameter valuations for which the accepting location of the pattern PTA and of the timed word TA is reachable.
\end{enumerate}

\subsection{Our approach step by step}\label{subsection:reduction_efsynth}
\begin{figure*}[tb]
	\begin{subfigure}[b]{\textwidth}
	\centering
		\scalebox{.9}{
		\begin{tikzpicture}[shorten >=1pt,node distance=2cm,on grid,auto]
		\node[location,initial] (prepres_0) {$\loc_0''$};
		\node[location, above of=prepres_0] (pres_0) {$\loc_0'$};
		\node[location, below right of=prepres_0] (s_0) {$\loc_0$};
		\node[location] (s_1) [right of=s_0] {$\loc_1$};
		\node[location] (s_2) [right of=s_1] {$\loc_2$};
		\node[location] (s_3) [right of=s_2]{$\loc_3$};
		\node[location] (s_4) [right of=s_3]{$\loc_4$};
		\node[location,accepting] (s_5) [right of=s_4]{$\loc_5$};

		\path[->] 
		(prepres_0) edge node[below left,align=center] {$\styleclock{\clockabs} = \styleparam{t} = 0$ \\ \actionStart} (s_0)
		(prepres_0) edge[] node[align=center] {$\styleact{a},\styleact{b}$ \\ $\styleclock{x} := 0$} (pres_0)
		(pres_0) edge [loop left] node[left,align=center] {$\styleact{a},\styleact{b}$ \\ $\styleclock{x} := 0$} (pres_0)
		(pres_0) edge [bend left,above right] node[align=center] {$\styleclock{\clockabs} = \styleparam{t} \land \styleclock{x} > 0$ \\ \actionStart{} \\ $\styleclock{x} := 0$} (s_0)
		(s_0) edge [above] node {\begin{tabular}{c}
									$\styleclock{x} > \styleparam{\param_1}$\\
									$\styleact{a}$\\
									$\styleclock{x} := 0$
								\end{tabular}} (s_1)
		(s_1) edge [above] node {\begin{tabular}{c}
									$\styleclock{x} < \styleparam{\styleparam{\param_2}}$\\
									$\styleact{a}$\\
									$\styleclock{x} := 0$
								\end{tabular}} (s_2)
		(s_2) edge [above] node[align=center] {$\styleclock{x} < \styleparam{\styleparam{\param_2}}$\\$\styleact{a}$} (s_3)
		(s_3) edge [above] node[align=center] {$\styleclock{\clockabs} = \styleparam{t'}$\\$\styleact{\$}$\\$\styleclock{x} := 0$} (s_4)
		(s_4) edge [above] node[align=center] {$\styleclock{x} > 0$\\\actionEnd{}} (s_5)
		;
		\end{tikzpicture}}
	\caption{$\TransPattern$ applied to the PTA in \cref{figure:example:PTA}}
	\label{figure:approach:PTA}
	\end{subfigure}
	\begin{subfigure}[b]{0.99\textwidth}
	\centering
	\footnotesize
		\begin{tikzpicture}[shorten >=1pt,node distance=1.5cm,on grid,auto]
		\node[location,initial, accepting] (w0) {$\wloc_0$};
		\node[location,accepting] (w1) [right of=w0] {$\wloc_1$};
		\node[location,accepting] (w2) [right of=w1] {$\wloc_2$};
		\node[location,accepting] (w3) [right of=w2]{$\wloc_3$};
		\node[location,accepting] (w4) [right of=w3]{$\wloc_4$};
		\node[location,accepting] (w5) [right of=w4]{$\wloc_5$};
		\node[location,accepting] (w6) [right of=w5]{$\wloc_6$};
		\node[location,accepting] (w7) [right of=w6]{$\wloc_7$};
		\node[location,accepting] (w8) [right of=w7]{$\wloc_8$};
		\node[location,accepting] (w9) [right of=w8]{$\wloc_9$};

		\path[->]
			(w0) edge [above] node[align=center] {$\styleclock{\clockabs} = 0.5$\\$\styleact{a}$} (w1)
			(w1) edge [above] node[align=center] {$\styleclock{\clockabs} = 0.9$\\$\styleact{a}$} (w2)
			(w2) edge [above] node[align=center] {$\styleclock{\clockabs} = 1.3$\\$\styleact{b}$} (w3)
			(w3) edge [above] node[align=center] {$\styleclock{\clockabs} = 1.7$\\$\styleact{b}$} (w4)
			(w4) edge [above] node[align=center] {$\styleclock{\clockabs} = 2.8$\\$\styleact{a}$} (w5)
			(w5) edge [above] node[align=center] {$\styleclock{\clockabs} = 3.7$\\$\styleact{a}$} (w6)
			(w6) edge [above] node[align=center] {$\styleclock{\clockabs} = 4.9$\\$\styleact{a}$} (w7)
			(w7) edge [above] node[align=center] {$\styleclock{\clockabs} = 5.3$\\$\styleact{a}$} (w8)
			(w8) edge [above] node[align=center] {$\styleclock{\clockabs} = 6.0$\\$\styleact{a}$} (w9)
		;
		\end{tikzpicture}
	\caption{$\TransWord$ applied to the timed word in \cref{figure:example:word}}%
	\label{figure:approach:word}
	\end{subfigure}

	\caption{Our transformations exemplified on \cref{figure:example}}%
	\label{figure:approach}
\end{figure*}

\subsubsection{Making the pattern symbolic}\label{sss:TransPattern}
In this first step, we first add two parameters $t$ and~$t'$, which encode the (symbolic) start and end time where the pattern holds on the input timed word.
This way, we will obtain a result in the form of a finite union of polyhedra in $|\Param| + 2$ dimensions, where the 2 additional dimensions come from the addition of~$t$ and~$t'$.
We also add a clock~$\clockabs$ measuring the absolute time, \ie{} initially~0 and never reset (this clock is shared by the pattern PTA and the subsequent timed word TA).
Then, we modify the pattern PTA as follows:
\begin{enumerate}
	\item we add two fresh locations (say~$\locinit'$ and $\locinit''$) prior to the initial location~$\locinit$;
	\item we add a fresh clock (say~$\styleclock{\clock}$); in practice, as this clock is used only in the initial location, an existing clock of the pattern may be reused;
	\item we add an unguarded self-loop allowing any action of the timed word on $\locinit'$, and resetting~$\styleclock{\clock}$;
	\item we add an unguarded transition from~$\locinit''$ to $\locinit'$ allowing any action of the timed word on $\locinit'$, and resetting~$\styleclock{\clock}$;
	\item we add a transition from~$\locinit'$ to~$\locinit$ guarded by~$\styleclock{\clockabs} = \styleparam{t} \land \styleclock{\clock} > 0$, labeled with a fresh action \actionStart{} and resetting all clocks of the pattern (except~$\styleclock{\clockabs}$);
	\item we add a transition from~$\locinit''$ to~$\locinit$ guarded by~$\styleclock{\clockabs} = \styleparam{t} \land \styleclock{\clockabs} = 0$, labeled with \actionStart{};
	\item we add a guard $\styleclock{\clockabs} = \styleparam{t'}$ and reset $\styleclock{\clock}$ on the accepting transitions labeled with~$\styleact{\$}$;
	\item we add an extra location~$\loc_{F}$ after the former (unique) accepting location, with a transition guarded by $\styleclock{\clock} > 0$, labeled with \actionEnd{};
	\item the initial location of the modified PTA becomes~$\locinit''$;
	\item the (only) final location of the modified PTA becomes~$\loc_{F}$.
\end{enumerate}

Let us give the intuition behind our transformation.
First, the two guards $\styleclock{\clockabs} = \styleparam{t}$ and $\styleclock{\clockabs} = \styleparam{t'}$ allow recording symbolically the value of the starting and ending dates.
Second, the new locations $\locinit''$ and $\locinit'$ allow the pattern to ``start anytime''; that is, it can synchronize with the timed word TA for an arbitrary long time while staying in the initial location~$\locinit''$ (and therefore \emph{not} matching the pattern), and start (using the transition from~$\locinit'$ to~$\locinit$) anytime.
Third, due to the constraint $\styleclock{\clock} > 0$, a non-zero time must elapse between the last action before the pattern start and the actual pattern start.
The distinction between $\locinit''$ and $\locinit'$ is necessary to also allow starting the pattern at $\styleclock{\clockabs} = 0$ if no action occurred before.
Finally, the guard $\styleclock{\clock} > 0$ just before the accepting location ensures the next action of the system (if any) is taken after a non-zero delay, following our definitions of timed word and projection.

Let us formalize this procedure \TransPattern{} below.

\begin{definition}\label{definition:TransPattern}
 Given a PTA $\A$ over the augmented alphabet $\Actions \disjointUnion \{\$\}$, 
 $\TransPattern(\A) = (\Actions \disjointUnion \{\$, \actionStart{}, \actionEnd{}\}, \Loc \disjointUnion \{\locinit', \locinit'', \loc_{F}\}, \loc''_0, \{\loc_{F}\}, \Clock \disjointUnion \{\clockabs, \clock\}, \Param \disjointUnion \{\styleparam{t}, \styleparam{t'}\}, \Edges')$, where $\Edges'$ is the following:
	\begin{align*}
	\Edges' =& \{(\loc, \guard, \action, \resets, \loc') \mid (\loc, \guard, \action, \resets, \loc') \in \Edges, \action \in \Actions \}\\
	\disjointUnion& \{(\loc, \guard \land \clockabs = t', \$, \resets \disjointUnion \{x\}, \loc') \mid (\loc, \guard, \$, \resets, \loc') \in \Edges \}\\
	\disjointUnion& \{(\locinit', \top, a, \{x\}, \locinit') \mid a \in \Actions\}\\ 
	\disjointUnion& \{(\locinit'', \top, a, \{x\}, \locinit') \mid a \in \Actions\}\\
	\disjointUnion& \{(\locinit', \clockabs = t \land x > 0, \actionStart{}, \Clock, \locinit), (\locinit'', \clockabs = 0 \land \clockabs = 0, \actionStart{}, \Clock, \locinit)\} \\
	\disjointUnion& \{(\loc, x > 0, \actionEnd{}, \emptyset, \loc_F) \mid \loc \in \LocFinal\}\text{.}\\
	\end{align*}
\end{definition}

\begin{example}
Consider the pattern PTA~$\A$ in \cref{figure:example:PTA}.
The result of $\TransPattern(\A)$ is given in \cref{figure:approach:PTA}.
Note that we use the same clock $\clock$ for both the extra clock introduced by our construction and the original clock of the pattern automaton from \cref{figure:example:PTA}.
\end{example}

\subsubsection{Converting the timed word into a (P)TA}
\label{sss:word_conversion}
In this second step, we convert the timed word into a (non-parametric) timed automaton.
This is very straightforward, and simply consists in converting a timed word of the form $(\action_1, \tau_1), \dots , (\action_n, \tau_n)$ into a sequence of transitions labeled with~$\action_i$ and guarded with $\clockabs = \tau_i$ (recall that $\clockabs$ measures the absolute time and is shared by the timed word automaton and the pattern automaton).
All locations are made accepting.\footnote{%
	This was not the case in~\cite{AHW18}, but we added this requirement due to the modification of the definition of synchronized product (\cref{def:product}), so as to have a unified framework in \cref{section:PTMC,section:adhoc}.
}%

Let us formalize this procedure \TransWord{} below.

\begin{definition}\label{definition:TransWord}
	Given a timed word $\word = (\action_1, \tau_1), \dots , (\action_n, \tau_n)$, the \emph{transformation of this timed word into a PTA} is the PTA $\TransWord(\word) = (\Actions, \{\locinit,\loc_1,\dots,\loc_n\}, \locinit, \{\locinit,\loc_1,\dots,\loc_n\}, \{\clockabs\}, \emptyset, \Edges)$, where:
	\[\Edges = \bigcup_{i \in \{1,2,\dots,n\}} (\loc_{i-1}, \clockabs = \tau_i, \action_i, \emptyset, \loc_{i})\text{.}\]
\end{definition}

\begin{example}
Consider the timed word~$\word$ in \cref{figure:example:word}.
The result of $\TransWord(\word)$ is given in \cref{figure:approach:word}.
\end{example}

\subsubsection{Synchronized product}
The last part of the method consists in performing the synchronized product of $\TransPattern(\A)$ and $\TransWord(\word)$, and calling \EFsynth{} on the resulting PTA.

\medskip

We summarize our method $\PTPM(\A,\word)$ in \cref{algo:PTPM}.

\begin{algorithm}[tb]
	\Input{A pattern PTA $\A$ with accepting location~$\LocFinal$, a timed word~$\word$}
	\Output{Constraint $\K$ over the parameters}

	\BlankLine

	$\A' \assign \TransPattern(\A)$
	
	$\A_{\word} \assign \TransWord(\word)$
	
	\Return $\EFsynth(\A' \parallel \A_{\word}, \LocFinal)$
	
	\caption{$\PTPM(\A, \word)$}
	\label{algo:PTPM}
\end{algorithm}

\begin{example}
	Consider again the timed word~$\word$ and the PTA pattern~$\A$ in \cref{figure:example}.
	The result of $\PTPM(\A,\word)$ is as follows:

	\begin{align*}
		&
		1.7 < \styleparam{t} < 2.8 - \styleparam{\param_1}
		\land
		4.9 \leq \styleparam{t'} < 5.3
		\land
		\styleparam{\param_2} > 1.2
		\\
		\lor\ \ &
		2.8 < \styleparam{t} < 3.7 - \styleparam{\param_1}
		\land
		5.3 \leq \styleparam{t'} < 6
		\land
		\styleparam{\param_2} > 1.2
		\\
		\lor \ \ &
		3. 7 < \styleparam{t} < 4.9 - \styleparam{\param_1}
		\land
		\styleparam{t'} \geq 6
		\land
		\styleparam{\param_2} > 0.7
	\end{align*}
	Observe that, for the parameter valuation given in the introduction ($\param_1 = \param_2 = 1$), only the pattern corresponding to the last disjunct could be obtained, \ie{} the pattern that matches the last three $\styleact{a}$ of the timed word in \cref{figure:example:word}.
	In contrast, the first disjunct can match the first three $\styleact{a}$ coming after the two $\styleact{b}$s, while the second disjunct allows to match the three $\styleact{a}$s in the middle of the last five $\styleact{a}$s in \cref{figure:example:word}.\LongVersion{
	
	}
	We give various projections of this constraint onto two dimensions in \cref{figure:projections} (the difference between plain red and light red is not significant---light red constraints denote unbounded constraints towards at least one dimension).
\end{example}

\begin{figure*}[t]
	\begin{subfigure}[b]{.3\textwidth}
		\centering
		
		\includegraphics[width=.8\textwidth]{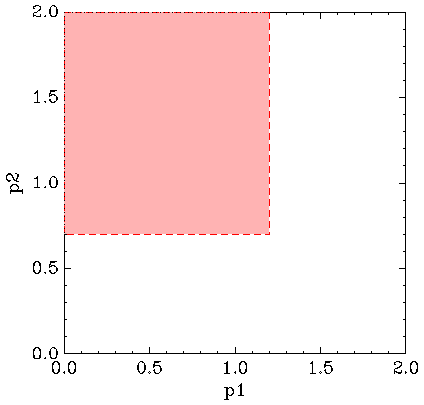}

		\caption{On $\param_1$ and $\param_2$}
	\end{subfigure}
	\hspace{.3cm}
	\begin{subfigure}[b]{.3\textwidth}
		\centering
		
		\includegraphics[width=.8\textwidth]{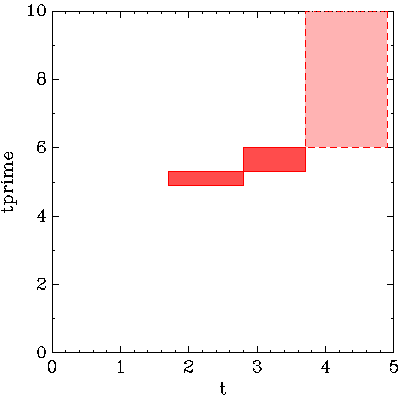}

		\caption{On $t$ and $t'$}
	\end{subfigure}
	\hspace{.3cm}
	\begin{subfigure}[b]{.3\textwidth}
		\centering
		
		\includegraphics[width=.8\textwidth]{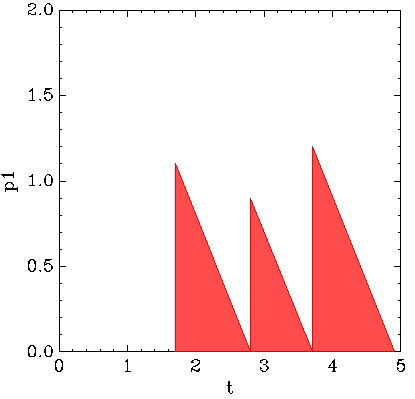}

		\caption{On $t$ and $\param_1$}
	\end{subfigure}
	\caption{Projections of the result of parametric timed pattern matching on \cref{figure:example}}
	\label{figure:projections}
\end{figure*}
\subsection{Termination}

\LongVersion{
	We state below the termination of our procedure.
}

\begin{theorem}[termination]\label{theorem:termination}
	Let $\A$ be a PTA encoding a parametric pattern, and $\word$ be a timed word.
	Then $\PTPM(\A,\word)$ terminates.
\end{theorem}
\begin{proof}
	First, observe that there may be non-determinism in the pattern PTA, \ie{} the timed word can potentially synchronize with two transitions labeled with the same action from a given location.
	Even if there is no syntactic nondeterminism, nondeterminism can appear due to the interleaving of the initial \actionStart{} action:
		in \cref{figure:approach}, the first \styleact{a} of the timed word can either synchronize with the self-loop on~$\locinit'$, or the \actionStart{} action can first occur, and then the first \styleact{a} of the timed word synchronizes with the \styleact{a} labeling the transition from~$\locinit$ to~$\loc_1$ of the pattern PTA.
	Second, the pattern PTA may well have loops (and this is the case in our experiments in \cref{ss:XP:blowup}), which yields an infinite parametric zone graph (for the pattern automaton not synchronized with the word automaton).
	However, let us show that only a finite part of the parametric zone graph is explored by \EFsynth{}:
		indeed, since $\TransWord(\word)$ is only a finite sequence, and thanks to the strong synchronization between the pattern PTA and the timed word PTA and due to \cref{assumption:transitions}, only a finite number of finite discrete paths in the synchronized product will be explored.
	The only interleaving is due to the initial \actionStart{} action (which appears twice in the pattern PTA but can only be taken once at most due to the mutually exclusive guards $\clock = 0$ and $\clock > 0$), and due to the accepting $\styleact{\$}$ action, that only appears on the last transition to the last-but-one accepting location.
	As the pattern PTA is finitely branching, this gives a finite number of finite paths.
	The length of each path is clearly bounded by~$|\word| + 3$.
	Let us now consider the maximal number of such paths: given a location in $\TransWord(\word)$, the choice of the action (say~$\styleact{\action}$) is entirely deterministic.
	However, the pattern PTA may be non-deterministic, and can synchronize with $\BranchingCard$ outgoing transitions labeled with~$\styleact{\action}$, which gives $\BranchingCard^{|\word|}$ combinations.
	In addition, the \actionStart{} action can be inserted exactly once, at any position in the timed word (from before the first action to after the last action of the word---in the case of an empty pattern): this gives therefore $(|\word|+1)\times\BranchingCard^{|\word|}$ different runs.
	(The $\styleact{\$}$ is necessarily the last-but-one action, and does not impact the number of runs, as the (potential) outgoing transitions from the accepting location are not explored.)
	Altogether, a total number of at most $(|\word| + 3) \times (|\word|+1) \times\BranchingCard^{|\word|}$ symbolic states is explored by \EFsynth{} in the worst case.
\end{proof}

\cref{theorem:termination} may not come as a surprise, as the input timed word is finite.
But it is worth noting that it comes in contrast with the fact that the wide majority of decision problems are undecidable for parametric timed automata, including the emptiness of the valuation set for which a given location is reachable both, for integer- and rational-valued parameters~\cite{AHV93,Miller00} (see \cite{Andre19STTT} for a survey).

\subsection{Correctness}\label{ss:correctness}

We show the correctness of our procedure using the following lemma on $\TransPattern$ with the formal definition of $\TransPattern(\A)$.

\begin{lemma}\label{lemma:transpattern}
 For any timed word $\word$, $t < t' \in \Rp$, PTA $\A$, and for any parameter valuation $\pval$ over $\Param$, 
 we have
 $\Lg(\valuate{\TransPattern(\A)}{(\pval, t, t')}) = \TW(\Actions) \circ (\actionStart{}, t) \circ \{\word + t \mid \word \in \Lg(\valuate{\A}{\pval})\} \circ \{(\actionEnd{}, t'') \mid t'' > t'\}$, where $(\pval, t, t')$ is the parameter valuation such that $(\pval, t, t')(\styleparam{\param}) = \pval(\param)$ for any $\styleparam{\param} \in \Param$, $(\pval, t, t')(\styleparam{t}) = t$, and $(\pval, t, t')(\styleparam{t'}) = t'$.
\end{lemma}

\begin{proof}
 Let $\A = (\Actions \disjointUnion \{\$\}, \Loc, \loc_0, \LocFinal, \Clock, \Param, \Edges)$ and
 $\TransPattern(\A) = (\Actions \disjointUnion \{\$, \actionStart{}, \actionEnd{}\}, \Loc \disjointUnion \{\locinit', \locinit'', \loc_{F}\}, \loc''_0, \{\loc_{F}\}, \Clock \disjointUnion \{\clockabs, \clock\}, \Param \disjointUnion \{\styleparam{t}, \styleparam{t'}\}, \Edges')$, where $\Edges'$ is defined according to~\cref{definition:TransPattern}.

Now, let $\A'$ denote the PTA identical to $\TransPattern(\A)$, except that the final location is~$\loc_0$ instead of~$\loc_{F}$. %
That is, $\A' = (\Actions \disjointUnion \{\$, \actionStart{}, \actionEnd{}\}, \Loc \disjointUnion \{\locinit', \locinit'', \loc_{F}\}, \loc''_0, \{\loc_0\}, \Clock \disjointUnion \{\clockabs, \clock\}, \Param \disjointUnion \{\styleparam{t}, \styleparam{t'}\}, \Edges')$.
For $\A'$, we have 
 $\Lg(\valuate{\A'}{(\pval, t, t')}) = \TW(\Actions) \circ (\actionStart{}, t)$.

Now, let $\A''$ denote the PTA identical to $\TransPattern(\A)$, except that the final locations are~$\LocFinal$ (\ie{} the final locations of the original automaton~$\A$).
That is, $\A'' = (\Actions \disjointUnion \{\$, \actionStart{}, \actionEnd{}\}, \Loc \disjointUnion \{\locinit', \locinit'', \loc_{F}\}, \loc''_0, \LocFinal, \Clock \disjointUnion \{\clockabs, \clock\}, \Param \disjointUnion \{\styleparam{t}, \styleparam{t'}\}, \Edges')$.
Since all the clock variables $\Clock$ are reset at all the transitions to $\locinit$ in $\Edges$, 
 for~$\A''$, we have the following.
 \begin{align*}
  \Lg(\valuate{\A''}{(\pval, t, t')}) 
  &= \Lg(\valuate{\A'}{(\pval, t, t')}) \cdot \Lg(\valuate{\A}{\pval}) \\
  &= (\TW(\Actions) \circ (\actionStart{}, t)) \cdot \Lg(\valuate{\A}{\pval})\\
  &= \TW(\Actions) \circ (\actionStart{}, t) \circ \{\word + t \mid \word \in \Lg(\valuate{\A}{\pval})\}
 \end{align*}
 Moreover, since we have $\clockabs = t'$ at all the transitions to $\LocFinal$, we have 
 \begin{align*}
  &\Lg(\valuate{\TransPattern(\A)}{(\pval, t, t')}) \\
  =& \Lg(\valuate{\A''}{(\pval, t, t')}) \cdot \{(\actionEnd{}, t'') \mid t'' > 0\} \\
  =& (\TW(\Actions) \circ (\actionStart{}, t) \circ \{\word + t \mid \word \in \Lg(\valuate{\A}{\pval})\}) \cdot \{(\actionEnd{}, t'') \mid t'' > 0\} \\
  =& \TW(\Actions) \circ (\actionStart{}, t) \circ \{\word + t \mid \word \in \Lg(\valuate{\A}{\pval})\} \circ \{(\actionEnd{}, t'') \mid t'' > t'\}
 \end{align*}
\end{proof}

We state the correctness of the workflow in \cref{subsection:reduction_efsynth} as follows.

\begin{theorem}
 [correctness]
 For any PTA $\A$ and timed word $\word$, we have the following:
 \[
 \mathcal{M}(\word, \A) = \{(t, t', \pval) \mid (\pval, t, t') \in \EFsynth(\TransPattern(\A) \parallel \TransWord(\word), \LocFinal)\}
 \]
 \noindent where $\LocFinal$ denotes the set of final locations of the product automaton $\TransPattern(\A) \parallel \TransWord(\word)$.%
\end{theorem}

\begin{proof}
	Let $\A$ be a PTA, and $\pval$ be a parameter valuation.
	First note that, by definition of $\TransWord(\word)$, 
	for any timed word $\word$, $t < t' \in \Rp$, 
	we have 
	$\Lg(\TransWord(\word)) = \bigl\{\word(0, m) \mid m \in \{0,1,\dots,|\word|\}\bigr\}$.

	Now, thanks to \cref{lemma:transpattern}, we have the following, where $\project{u}{\Actions}$ is the projection of the timed word $u$ to an alphabet $\Actions$, \ie{} the timed word identical to $u$, except that all the actions not in $\Actions$ are removed. We also use the same notation for a set of timed words.
 \begin{align*}
  &(\pval, t, t') \in \EFsynth(\TransPattern(\A) \parallel \TransWord(\word), \LocFinal)\\
  \iff& \exists u \in \TW(\Actions \disjointUnion \{\actionStart{}, \actionEnd{}, \$\}).\, u \in \Lg\bigl(\valuate{\TransPattern(\A) \parallel \TransWord(\word)}{(\pval,t,t')}\bigr) \\
  \iff& \exists u \in \TW(\Actions \disjointUnion \{\actionStart{}, \actionEnd{}, \$\}).\, u \in \Lg\bigl(\valuate{\TransPattern(\A)}{(\pval,t,t')}\bigr) \land  \project{u}{\Actions} \in \Lg(\TransWord(\word)) \\
  \iff& \exists u \in \Lg\bigl(\valuate{\TransPattern(\A)}{(\pval,t,t')}\bigr).\, \project{u}{\Actions} = \word(0,n) \text{, where $\tau_n \leq t' < \tau_{n+1}$}\\
  \iff& \exists u \in \TW(\Actions) \circ (\actionStart{}, t) \circ \{\word' + t \mid \word' \in \Lg(\valuate{\A}{\pval})\} \circ \{(\actionEnd{}, t'') \mid t'' > t'\}.\, \project{u}{\Actions} = \word(0,n),\\ &\qquad\text{where $\tau_n \leq t' < \tau_{n+1}$}\\
  \iff& \word(0,n) \in \project{\Bigl(\TW(\Actions) \circ (\actionStart{}, t) \circ \{\word' + t \mid \word' \in \Lg(\valuate{\A}{\pval})\} \circ \{(\actionEnd{}, t'') \mid t'' > t'\}\Bigr)}{\Actions},\\ &\qquad\text{where $\tau_n \leq t' < \tau_{n+1}$}\\
  \iff& \word|_{(t,t')} \in \Lg(\valuate{\A}{\pval}) \\
  \iff& (t, t', \pval) \in \mathcal{M}(\word, \A)
 \end{align*}
	This concludes the proof.
\end{proof}

\subsection{Pattern matching with optimization}\label{ss:optimization}

We also address the following optimization problem:
given a timed word and a pattern containing parameters, what is the minimum or maximum value of a given parameter such that the pattern is matched by the timed word?
Formally, this gives the following problem, where ``optimal'' denotes minimal or maximal.

\smallskip

\defProblem
	{Optimal parametric timed pattern matching}
	{a timed word~$\str$ over an alphabet $\Actions$, a PTA~$\A$ over the augmented alphabet $\Actions \disjointUnion \{\$\}$, a parameter~$\param$}
	{compute the optimal value $\optvalue$ such that there exist $\pval, t, t'$ such that $\pval(\param) = \optvalue$ and the segment $\word|_{(t,t')}$ is accepted by $\valuate{\A}{\pval}$}

\medskip

That is, we are only interested in the \emph{optimal value~$\optvalue$ of the given parameter~$\param$}, and not in the full list of matches as in~\PTPM{}.

While this problem can be solved using our solution from \cref{ss:general} (by computing the multidimensional constraint, and then eliminating all parameters but the target parameter, using variable elimination techniques), we use here a dedicated approach, with the hope it be more efficient.
Instead of managing all symbolic matches (\ie{} a finite union of polyhedra), we simply manage the current optimum; in addition, we cut branches that cannot improve the optimum, with the hope to reduce the number of states explored.
For example, assume parameter $\param$ is to be minimized; if the current minimum is $\param > 2$, and if a newly computed symbolic state is such that $\param \geq 3$, then the branch starting from this new symbolic state will not improve the minimum, and can safely be discarded.
This branch cutting is directly managed by the underlying algorithm~$\EFsynthOpt$.

\begin{algorithm}[tb]
	\Input{A pattern PTA $\A$ with accepting location~$\LocFinal$, a timed word~$\word$, a parameter~$\param$ to be optimized}
	\Output{Constraint $\K$ over the parameters}

	\BlankLine

	$\A' \assign \TransPattern(\A)$
	
	$\A_{\word} \assign \TransWord(\word)$
	
	\Return $\EFsynthOpt(\A' \parallel \A_{\word}, \LocFinal, \param)$
	
	\caption{$\PTPMopt(\A, \word, \param)$}
	\label{algo:PTPMopt}
\end{algorithm}

We give this procedure \PTPMopt{} in \cref{algo:PTPMopt}.
It is basically a refinement of \cref{algo:PTPM}, which takes as additional argument the parameter~$\param$ to be optimized, and where the call to $\EFsynth$ is replaced with a call to $\EFsynthOpt$.

Similarly to $\EFsynthOpt$ recalled in \cref{newtext:EFsynthopt},
$\PTPMopt$ is a ``generic'' algorithm, that can be ``instantiated'' to either
\begin{ienumeration}%
	\item ``pattern matching with minimization'' (say $\PTPMmin$), by replacing in \cref{algo:PTPMopt} $\EFsynthOpt$ with $\EFsynthMin$; or
	\item ``pattern matching with maximization'' (say $\PTPMmax$), by replacing in \cref{algo:PTPMopt} $\EFsynthOpt$ with $\EFsynthMax$.
\end{ienumeration}%

\begin{example}\label{example:PTPMopt}
	Consider again the timed word~$\word$ and the PTA pattern~$\A$ in \cref{figure:example}.
	Let us first minimize $\styleparam{\param_2}$ so that the pattern matches the timed word for at least one position:
	by calling $\PTPMmin(\A, \word, \styleparam{\param_2})$, we obtain
	$\styleparam{\param_2} > 0.7$.
	That is to say, the smallest valuation of $\param_2$ allowing the pattern to match the timed word for at least one position is infinitesimally larger than~$0.7$, but $0.7$ itself is not an admissible valuation.
	Let us then maximize $\styleparam{\param_1}$ so that the pattern matches the timed word for at least one position:
	by calling $\PTPMmax(\A, \word, \styleparam{\param_1})$, we obtain
	$\styleparam{\param_1} < 1.2$.
\end{example}
\subsection{Experiments}\label{ss:PTMC-experiments}

We evaluated our approach against two standard benchmarks from~\cite{HAF14}, already used in~\cite{WHS17}, as well as a third \emph{ad-hoc} benchmark specifically designed to test the limits of parametric timed pattern matching.
We fixed no bounds for our parameters.

We used \imitator{}~\cite{Andre21} to perform the parameter synthesis (algorithm \EFsynth{}).
\imitator{} relies on the Parma Polyhedra Library (PPL)~\cite{BHZ08} to compute symbolic states.
It was shown in~\cite{BFMU17} that polyhedra may be dozens of times slower than more efficient data structures such as DBMs (difference bound matrices); however, for parametric analyses, DBMs are not suitable, and parameterized extensions (\eg{} in~\cite{HRSV02}) still need polyhedra in their representation.

	We used a slightly modified version of \imitator{} for technical reasons:
\imitator{} handles non-convex constraints (finite unions of polyhedra); while most case studies solved by \imitator{} in the past handle simple constraints (made of a few disjuncts), the experiments in this manuscript may handle up to \emph{dozens of thousands} of such polyhedra.
We therefore had to disable an inclusion test of a newly computed state into the already computed constraint: this test usually has a very interesting gain but, on our complex polyhedra, it had disastrous impact on the performance, due to the inclusion check of a (simple) new convex polyhedron into a disjunction of dozens of thousands of convex polyhedra.
To disable this check, we added a new option\footnote{%
	Option ``\texttt{-no-inclusion-test-in-EF}'' in \imitator{} 2.10.4.
} (not set by default) to the master branch of \imitator{}, and used it in all our experiments.
	
We wrote a simple Python script to implement the \TransWord{} procedure; the patterns (\cref{figure:patterns}) were manually transformed following the \TransPattern{} procedure, and converted into the input language of \imitator{}.

We ran experiments using \imitator{} 2.10.4 ``Butter Jellyfish'' %
on a Dell Precision 3620 i7-7700 3.60\,GHz with 64\,GiB memory running Linux Mint 19 beta 64\,bits.\footnote{Sources, binaries, models, logs can be found at \href{https://www.imitator.fr/static/ICECCS18}{\nolinkurl{imitator.fr/static/ICECCS18}}.}
\begin{figure*}[t]
	\begin{subfigure}[b]{\textwidth}
	\centering
	\footnotesize

  \begin{tikzpicture}[shorten >=1pt,node distance=2.5cm,on grid,auto] 
    \node[location,initial] (s_0)  {}; 
    \node[location,node distance=2.5cm] (s_1) [right=of s_0] {$g_1$}; 
    \node[location,node distance=2.5cm] (s_2) [right=of s_1] {$g_2$};
    \node[location,accepting,node distance=1.5cm] (s_3) [right=of s_2] {\cmark};
    \path[->] 
    (s_0) edge  [above] node[align=center] {$\styleact{g_1}$\\$\styleclock{x} := 0$} (s_1)
    (s_1) edge  [above] node[align=center] {$\styleclock{x} < \styleparam{\param} $\\$\styleact{g_2}$} (s_2) %
    (s_2) edge  [above] node {$\styleact{\$}$} (s_3);
  \end{tikzpicture}
  \caption{\textsc{Gear}}
	\label{figure:patterns:gear}
	\end{subfigure}
	\begin{subfigure}[b]{\textwidth}
	\centering
	\scriptsize

	\begin{tikzpicture}[auto,node distance=2cm]
		\node[location, initial] (s_000) {$?$};

		\node[location] (s_100)[above right=of s_000,yshift=-10] {$g_1$};
		\node[location] (s_001)[below right=of s_000,yshift=30] {$?$};

		\node[location] (s_200)[right=of s_100] {$g_2$};
		\node[location] (s_101)[right=of s_001] {$g_1'$};

		\node[location] (s_300)[right=of s_200] {$g_3$};
		\node[location] (s_201)[right=of s_101] {$g_2'$};

		\node[location] (s_400)[right=of s_300] {$g_4$};
		\node[location] (s_301)[right=of s_201] {$g_3'$};

		\node[location] (s_401)[right=of s_301] {$g_4'$};

		\node[location,accepting] (f)[right=of s_400] {\cmark};

		\path[->]
		(s_000) edge  [above left] node {$\styleact{g_1}, \CTrue$} (s_100)
		(s_100) edge  [above] node {$\styleact{g_2}, \CTrue$} (s_200)
		(s_200) edge  [above] node {$\styleact{g_3}, \CTrue$} (s_300)
		(s_300) edge  [above] node {\begin{tabular}{c}
										$\styleact{g_4}, \styleclock{x} \leq \styleparam{\param}$\\ %
										$\styleclock{x} := 0$
									\end{tabular}} (s_400)

		(s_100) edge  [below left] node {$\styleact{rpmHigh}, \CTrue$} (s_101)
		(s_200) edge  [below left] node {$\styleact{rpmHigh}, \CTrue$} (s_201)
		(s_300) edge  [below left] node {$\styleact{rpmHigh}, \CTrue$} (s_301)
		(s_400) edge  [below left] node {$\styleact{rpmHigh}, \CTrue$} (s_401)

		(s_001) edge  [above] node {$\styleact{g_1}, \CTrue$} (s_101)
		(s_101) edge  [above] node {$\styleact{g_2}, \CTrue$} (s_201)
		(s_201) edge  [above] node {$\styleact{g_3}, \CTrue$} (s_301)
		(s_301) edge  [below] node {\begin{tabular}{c}
										$\styleact{g_4}, \styleclock{x} \leq \styleparam{\param}$\\ %
										$\styleclock{x} := 0$
									\end{tabular}} (s_401)

		(s_000) edge  [below left] node {$\styleact{rpmHigh}, \CTrue$} (s_001)

		(s_401) [bend right=0] edge [right] node {$\styleact{\$}, \styleclock{x} > 1$} (f);
	\end{tikzpicture}

	\caption{\textsc{Accel}}
	\label{figure:patterns:accel}
	\end{subfigure}
	\begin{subfigure}[b]{\textwidth}
	\centering
	\footnotesize

	\begin{tikzpicture}[shorten >=1pt,node distance=2.5cm,on grid,auto] 
		\node[location,initial] (s_0)  {$\loc_1$}; 
		\node[location] (s_1) [right=of s_0] {$\loc_2$}; 
		\node[location] (s_2) [right=of s_1] {$\loc_3$};
		\node[location,accepting] (s_3) [right=of s_2] {$\loc_4$};
		\path[->] 
			(s_0) edge [above] node[align=center] {$\styleact{a}$\\$\styleclock{y} := 0$} (s_1)
			(s_1) edge[bend right] node[above] {$\styleclock{x} < \styleparam{\param_1}$} node[below] {$\styleact{b}$} (s_2)
			(s_2) edge node[above,align=center] {$\styleclock{x} = \styleparam{\param_1}$\\$\styleact{\$}$} (s_3)
			(s_2) edge[bend right] node[above,align=center] {$\styleparam{\param_3} \leq \styleclock{y} < \styleparam{\styleparam{\param_2}}$\\$\styleact{a}$\\$\styleclock{y} := 0$} (s_1)
		;
	\end{tikzpicture}
  \caption{\textsc{Blowup}}
	\label{figure:patterns:blowup}
	\end{subfigure}

	\caption{Experiments: patterns}
	\label{figure:patterns}
\end{figure*}

\subsubsection{\textsc{Gear}}\label{ss:XP:gear}

Benchmark \textsc{Gear} is \LongVersion{inspired by the scenario of }monitoring the gear change of an automatic transmission system.
We conducted simulation of the model of an automatic transmission system~\cite{HAF14}.
We used \breach{}~\cite{Donze10} to generate an input sequence of gear change.
A gear is chosen from $\{g_1,g_2,g_3,g_4\}$.
The generated gear change is recorded in a timed word.
\LongVersion{%
	The set $W$ consists of 10 timed words;
	the length of each word is 1,467 to 14,657.
}

The pattern PTA $\A$, shown in \cref{figure:patterns:gear}, detects the violation of the following condition:
If the gear is changed to~1, it should not be changed to 2 within $\param$ seconds.
This condition is related to the requirement $\phi^{\mathit{AT}}_5$ proposed in~\cite{HAF14} (the nominal value for~$\param$ in~\cite{HAF14} is~2).

We tabulate our experiments in \cref{table:gear}.
We give from left to right the length of the timed word in terms of actions and time, then the data for \PTPM{} (the number of symbolic states explored, the number of (symbolic) matches found, the parsing time and the computation time excluding parsing) and for \PTPMopt{} (number of symbolic states explored and computation time) using \imitator{}.
The parsing time for \PTPMopt{} is almost identical to~\PTPM{} and is therefore omitted.

\begin{table*}[tb]
	\centering
	\scriptsize
	
	\begin{tabular}{| r | r | r | r | r | r | r | r |}
		\hline
		\startMultiCellHeader{2}{Model} & \multiCellHeader{4}{\PTPM{}} & \multiCellHeader{2}{\PTPMopt}\\
		\cellHeader{Length} & \cellHeader{Time frame} & \cellHeader{States} & \cellHeader{Matches} & \cellHeader{Parsing (s)} & \cellHeader{Comp.\ (s)} & \cellHeader{States} & \cellHeader{Comp.\ (s)}
		\\
		\hline
		1,467 & 1,000 & 4,453 & 379 & 0.02 & 1.60 & 3,322 & 0.94
		\\
		\hline
		2,837 & 2,000 & 8,633 & 739 & 0.33 & 2.14 & 6,422 & 1.70
		\\
		\hline
		4,595 & 3,000 & 14,181 & 1,247 & 0.77 & 3.63 & 10,448 & 2.85
		\\
		\hline
		5,839 & 4,000 & 17,865 & 1,546 & 1.23 & 4.68 & 13,233 & 3.74
		\\
		\hline
		7,301 & 5,000 & 22,501 & 1,974 & 1.94 & 5.88 & 16,585 & 4.79
		\\
		\hline
		8,995 & 6,000 & 27,609 & 2,404 & 2.96 & 7.28 & 20,413 & 5.76
		\\
		\hline
		10,316 & 7,000 & 31,753 & 2,780 & 4.00 & 8.38 & 23,419 & 6.86
		\\
		\hline
		11,831 & 8,000 & 36,301 & 3,159 & 5.39 & 9.75 & 26,832 & 7.87
		\\
		\hline
		13,183 & 9,000 & 40,025 & 3,414 & 6.86 & 10.89 & 29,791 & 8.61
		\\
		\hline
		14,657 & 10,000 & 44,581 & 3,816 & 8.70 & 12.15 & 33,141 & 9.89
		\\
		\hline
	\end{tabular}

	\caption{Experiments: \textsc{Gear}}
	\label{table:gear}
\end{table*}

The corresponding chart is given in \cref{figure:experiments:charts:gear} (\PTPM{} is given in plain black, and \PTPMopt{} in red dashed).
\PTPMopt{} brings a gain in terms of memory (symbolic states) of about 25\,\%, while the gain in time is about~20\,\%.

\begin{figure*}[t]
	\scriptsize

	\begin{subfigure}[b]{.3\textwidth}

	\begin{tikzpicture}[scale=0.3, xscale=.8]
		\draw[->] (0.0, 0.0) --++ (18.0, 0.0) node[right]{$|\word|$};
		\draw[->] (0.0, 0.0) --++ (0.0, 13.0) node[right]{$t$ $(s)$};

		\foreach \x in {0, 2, ..., 16} %
			\draw (\x, 0) -- (\x, -.2)node [below] {\tiny{$\x$}};
		\foreach \x in {0, 1, ..., 12} %
			\draw (0, \x) -- (-.2, \x) node [left] {\tiny{$\x$}};

		\draw[-]
			(1.467, 1.60)
			--
			(2.837, 2.14)
			--
			(4.595, 3.63)
			--
			(5.839, 4.68)
			--
			(7.301, 5.88)
			--
			(8.995, 7.22)
			--
			(10.316, 8.38)
			--
			(11.831, 9.75)
			--
			(13.183, 10.89)
			--
			(14.657, 12.15)
		;

		\draw[PTPMOPT]
			(1.467, 0.94)
			--
			(2.837, 1.70)
			--
			(4.595, 2.85)
			--
			(5.839, 3.74)
			--
			(7.301, 4.79)
			--
			(8.995, 5.76)
			--
			(10.316, 6.86)
			--
			(11.831, 7.87)
			--
			(13.183, 8.61)
			--
			(14.657, 9.89)
		;
	\end{tikzpicture}
	
	\caption{\textsc{Gear}}
	\label{figure:experiments:charts:gear}
	\end{subfigure}
	\hspace{1cm}
	\begin{subfigure}[b]{.3\textwidth}

	\begin{tikzpicture}[scale=0.2, xscale=.8]
		\draw[->] (0.0, 0.0) --++ (30.0, 0.0) node[right]{$|\word|$};
		\draw[->] (0.0, 0.0) --++ (0.0, 21.0) node[right]{$t$ $(s)$};

		\foreach \x in {0, 5, ..., 25} %
			\draw (\x, 0) -- (\x, -.2)node [below] {\tiny{$\x$}};
		\foreach \x in {0, 2, ..., 20} %
			\draw (0, \x) -- (-.2, \x) node [left] {\tiny{$\x$}};

		\draw[-]
			(2.559, 1.60)
			--
			(4.894, 3.04)
			--
			(7.799, 4.98)
			--
			(10.045, 6.51)
			--
			(12.531, 8.19)
			--
			(15.375, 10.14)
			--
			(17.688, 11.61)
			--
			(20.299, 13.52)
			--
			(22.691, 15.33)
			--
			(25.137, 16.90)
		;

		\draw[PTPMOPT]
			(2.559, 1.85)
			--
			(4.894, 3.57)
			--
			(7.799, 6.06)
			--
			(10.045, 7.55)
			--
			(12.531, 9.91)
			--
			(15.375, 12.39)
			--
			(17.688, 14.06)
			--
			(20.299, 16.23)
			--
			(22.691, 18.21)
			--
			(25.137, 20.61)

		;
	\end{tikzpicture}
	
	\caption{\textsc{Accel}}
	\label{figure:experiments:charts:accel}
	\end{subfigure}
	\hspace{1cm}
	\begin{subfigure}[b]{.3\textwidth}

	\begin{tikzpicture}[scale=0.05, xscale=5]
		\draw[->] (0.0, 0.0) --++ (11.0, 0.0) node[right]{$|\word|$};
		\draw[->] (0.0, 0.0) --++ (0.0, 100.0) node[right]{$t$ $(s)$};

		\foreach \x in {0, 2, ..., 8} %
			\draw (\x, 0) -- (\x, -1)node [below] {\tiny{$0.\x$}};
		\draw (10, 0) -- (10, -1)node [below] {\tiny{$1$}};
		\foreach \x in {0, 10, ..., 90} %
			\draw (0, \x) -- (-.2, \x) node [left] {\tiny{$\x{}0$}};

		\draw[-]
			(2, 1.53)
			--
			(4, 8.22)
			--
			(6, 23.68)
			--
			(8, 51.46)
			--
			(10, 94.07)
		;

		\draw[PTPMOPT]
			(2, 0.024)
			--
			(4, 0.049)
			--
			(6, 0.071)
			--
			(8, 0.105)
			--
			(10, 0.124)
		;
	\end{tikzpicture}
	
	\caption{\textsc{Blowup}}
	\label{figure:experiments:charts:blowup}
	\end{subfigure}

	\caption{Experiments: charts ($x$-scale $\times 1,000$)}
	\label{figure:experiments:charts}
\end{figure*}

\subsubsection{\textsc{Accel}}\label{ss:XP:accel}
The timed words $W$ of benchmark $\textsc{Accel}$ is also constructed from the
Simulink model of the automatic transmission
system~\cite{HAF14}.  For this benchmark, the
(discretized) value of three state variables are recorded in $W$:
engine RPM (discretized to ``high'' and ``low'' with a certain
threshold), velocity (discretized to ``high'' and ``low'' with a
certain threshold), and 4 gear positions.
We used \breach{}~\cite{Donze10} to generate an input sequence.
\LongVersion{%
Our set $W$ consists of 10 timed words;
the length of each word is 2,559 to 25,137.
}

The pattern PTA $\A$ of this benchmark is shown in \cref{figure:patterns:accel}.
This pattern matches a part of a
timed word that violates the following condition: If a gear changes
from 1 to 2, 3, and~4 in this order in $\param$~seconds and engine RPM
becomes large during this gear change, then the velocity of the car
must be sufficiently large in one second.
This condition models the requirement $\phi^{\mathit{AT}}_8$ proposed in~\cite{HAF14} (the nominal value for~$\param$ in~\cite{HAF14} is~10).

\begin{table*}[tb]
	\centering
	\scriptsize
	
	\begin{tabular}{| r | r | r | r | r | r | r | r |}
		\hline
		\startMultiCellHeader{2}{Model} & \multiCellHeader{4}{\PTPM{}} & \multiCellHeader{2}{\PTPMopt}\\
		\cellHeader{Length} & \cellHeader{Time frame} & \cellHeader{States} & \cellHeader{Matches} & \cellHeader{Parsing (s)} & \cellHeader{Comp.\ (s)} & \cellHeader{States} & \cellHeader{Comp.\ (s)}
		\\
		\hline
		2,559 & 1,000 & 6,504 & 2 & 0.27 & 1.60 & 6,502 & 1.85
		\\
		\hline
		4,894 & 2,000 & 12,429 & 2 & 0.86 & 3.04 & 12,426 & 3.57
		\\
		\hline
		7,799 & 3,000 & 19,922 & 7 & 2.21 & 4.98 & 19,908 & 6.06
		\\
		\hline
		10,045 & 4,000 & 25,520 & 3 & 3.74 & 6.51 & 25,514 & 7.55
		\\
		\hline
		12,531 & 5,000 & 31,951 & 9 & 6.01 & 8.19 & 31,926 & 9.91
		\\
		\hline
		15,375 & 6,000 & 39,152 & 7 & 9.68 & 10.14 & 39,129 & 12.39
		\\
		\hline
		17,688 & 7,000 & 45,065 & 9 & 13.40 & 11.61 & 45,039 & 14.06
		\\
		\hline
		20,299 & 8,000 & 51,660 & 10 & 18.45 & 13.52 & 51,629 & 16.23
		\\
		\hline
		22,691 & 9,000 & 57,534 & 11 & 24.33 & 15.33 & 57,506 & 18.21
		\\
		\hline
		25,137 & 10,000 & 63,773 & 13 & 31.35 & 16.90 & 63,739 & 20.61 %
		\\
		\hline
	\end{tabular}

	\caption{Experiments: \textsc{Accel}}
	\label{table:accel}
\end{table*}

Experiments are tabulated in \cref{table:accel}.
The corresponding chart is given in \cref{figure:experiments:charts:accel}.
This time, \PTPMopt{} brings almost no gain in terms of states, and a loss of speed of about 15 to~20\,\%, which may come from the additional polyhedra inclusion checks to test whether a branch is less good than the current optimum.

\subsubsection{\textsc{Blowup}}\label{ss:XP:blowup}

As a third experiment, we considered an original (toy) benchmark that acts as a worst case situation for parametric timed pattern matching.
Consider the PTA pattern in \cref{figure:patterns:blowup}, and assume a timed word consisting in an alternating sequence of ``$\styleact{a}$'' and ``$\styleact{b}$''.
Observe that the time from the pattern beginning (that resets~$\styleclock{x}$) to the end is exactly $\styleparam{\param_1}$ time units.
Also observe that the duration of the loop through $\loc_2$ and~$\loc_3$ has a duration in $[\styleparam{\param_3} , \styleparam{\param_2})$; therefore, for values sufficiently small of $\styleparam{\param_2},\styleparam{\param_3}$, one can always match a larger number of loops.
That is, for a timed word of length $2 n$ alternating between ``$\styleact{a}$'' and ``$\styleact{b}$'', there will be $n$ possible matches from position~0 (with $n$ different parameter constraints), $n-1$ from position 1, and so on, giving a total number of $\frac{n(n+1)}{2}$ matches with different constraints in 5~dimensions.

Note that this worst case situation is not specific to our approach, but would appear independently of the approach chosen for parametric timed pattern matching.

We generated random timed words of various sizes, all alternating exactly between ``$\styleact{a}$'' and ``$\styleact{b}$''.
\LongVersion{Our set $W$ consists of 5 timed words of length from 200 to~1,000.}

\begin{table*}[tb]
	\centering
	\scriptsize
	
	\begin{tabular}{| r | r | r | r | r | r | r | r |}
		\hline
		\startMultiCellHeader{2}{Model} & \multiCellHeader{4}{\PTPM{}} & \multiCellHeader{2}{\PTPMopt}\\
		\cellHeader{Length} & \cellHeader{Time frame} & \cellHeader{States} & \cellHeader{Matches} & \cellHeader{Parsing (s)} & \cellHeader{Comp.\ (s)} & \cellHeader{States} & \cellHeader{Comp.\ (s)}
		\\
		\hline
		200 & 101 & 20,602 & 5,050 & 0.01 & 15.31 & 515 & 0.24
		\\
		\hline
		400 & 202 & 81,202 & 20,100 & 0.02 & 82.19 & 1,015 & 0.49
		\\
		\hline
		600 & 301 & 181,802 & 45,150 & 0.03 & 236.80 & 1,515 & 0.71
		\\
		\hline
		800 & 405 & 322,402 & 80,200 & 0.05 & 514.57 & 2,015 & 1.05
		\\
		\hline
		1,000 & 503 & 503,002 & 125,250 & 0.06 & 940.74 & 2,515 & 1.24 %
		\\
		\hline
	\end{tabular}

	\caption{Experiments: \textsc{Blowup}}
	\label{table:blowup}
\end{table*}

Experiments are tabulated in \cref{table:blowup}.
The corresponding chart is given in \cref{figure:experiments:charts:blowup}.
\PTPM{} becomes clearly non-linear as expected.
This time, \PTPMopt{} brings a dramatic gain in both memory and time; even more interesting, \PTPMopt{} remains perfectly linear.

\subsection{Discussion}

A first positive outcome is that our method is effectively able to perform parametric pattern matching on words of length up to several dozens of thousands, and is able to output results in the form of several dozens of thousands of symbolic matches in several dimensions, in just a few seconds.
Another positive outcome is that \PTPM{} is perfectly linear in the size of the input word for \textsc{Gear} and \textsc{Accel}: this was expected as these examples are linear, in the sense that the number of states explored by \PTPM{} is linear as these patterns feature no loops.

Note that the parsing time is not linear, but it could be highly improved: due to the relatively small size of the models usually treated by \imitator{}, this part was never properly optimized, and it contains several quadratic syntax checking functions that could easily be avoided.

The performances do not completely allow yet for an \emph{online} usage in the current version of our algorithm and implementation (in~\cite{WHS17}, we pushed the \textsc{Accel} case study for timed words of length up to~17,280,002).
A possible direction is to perform \LongVersion{an on-the-fly computation of the parametric zone graph, more precisely to do }an on-the-fly parsing of the timed word automaton; this will allow \imitator{} to keep in memory a single location at a time (instead of up to 25,137 in our experiments).

Finally, although this is not our original motivation, we believe that, if we are only interested in \emph{robust} pattern matching, \ie{} non-parametric pattern matching but with an allowed deviation (``guard enlargement'') of the pattern automaton, then using the efficient 1-dimensional parameterized DBMs of~\cite{Sankur15} would probably be an interesting alternative: indeed, in contrast to classical parameterized DBMs~\cite{HRSV02} (that are made of a matrix and a parametric polyhedron), the structure of~\cite{Sankur15} only needs an $\ClockCard\times\ClockCard$ matrix with a single parameter, and seems particularly efficient.

\begin{figure*}[t]
	\begin{subfigure}[b]{.49\textwidth}
		\centering
		
		\includegraphics[width=.55\textwidth]{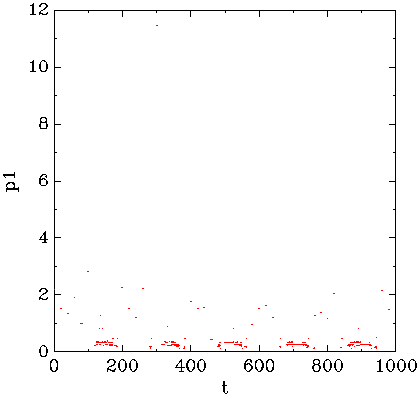}

		\caption{Projection onto $t$ and~$\param$}
		\label{figure:projections:gear:t-param}
	\end{subfigure}
	\begin{subfigure}[b]{.49\textwidth}
		\centering
		
		\includegraphics[width=.55\textwidth]{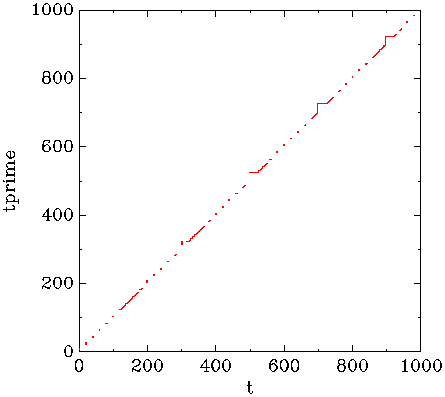}

		\caption{Projection onto $t$ and~$t'$}
	\end{subfigure}

	\caption{Visualizing a large number of matches for \textsc{Gear} ($|\word = 1467|$)}
	\label{figure:projections:gear}
\end{figure*}

\begin{remark}
	In the conference version of this work~\cite{AHW18}, we describe this approach as an \emph{offline} algorithm.
	In fact, it is essentially \emph{online} in the sense that it can potentially run with only a portion of the log: it relies on parallel composition of a specification automaton and a log automaton, and this parallel composition can be achieved on-the-fly.
	However, as mentioned in~\cite{BDDFMNS18}, ``a good online monitoring algorithm must:
	\begin{ienumeration}
		\item be able to generate intermediate estimates of property satisfaction based on partial signals,
		\item use minimal amount of data storage, and
		\item\label{item3} be able to run fast enough in a real-time setting.''
	\end{ienumeration}
	So, at least for point~\ref{item3}, our algorithm may not really run in a real-time setting.
\end{remark}

In contrast, we will present in the next section a contribution fast enough to run in a real-time setting, with runs of dozens of thousands of events being analyzable in less than a second.

\section{Dedicated method}\label{section:adhoc}

In this section, we present a \emph{dedicated} online algorithm for parametric timed pattern matching (\cref{ss:online-algorithm}).
We will then
enhance it with skipping in \cref{section:skipping}, and 
evaluate both versions with and without skipping in \cref{ss:adhoc-experiments}.

\subsection{An online algorithm}\label{ss:online-algorithm}

Similarly to the online algorithm for timed pattern matching in~\cite{WAH16},
our algorithm finds all the matching triples $(t,t',\pval) \in \mathcal{M}(\word,\A)$ by a breadth-first search.
Our algorithm is online in the following sense: after reading the $i$-th element $(a_i, \tau_i)$ of the timed word $\word=(\overline{a},\overline{\tau})$, it immediately outputs all the matching triples $(t,t',\pval)$ over the available prefix $(a_1, \tau_1), (a_2, \tau_2),\dots,(a_i, \tau_i)$ of $\word$.

Firstly, we define the auxiliary functions for our online algorithm for parametric timed pattern matching.
We introduce an additional variable $t$ representing the absolute time of the beginning of the matching.
We use a function $\rho\colon\Clock\to(\Rp\disjointUnion\{t\})$ to represent the latest reset time of each clock variable $\clock\in\Clock$.
Intuitively, $\rho(\clock)=\tau\in\Rp$ means the latest reset of $\clock$ is at $\tau$, and $\rho(\clock)=t$ means $\clock$ is not reset after the beginning of the matching.

\begin{definition}
 [$\eval(\rho,\tau)$]
 Let $\Clock$ be the set of clock variables and $t$ be the variable for the beginning of a matching.
 For a function $\rho\colon \Clock\to(\Rp\disjointUnion\{t\})$ and the current time $\tau\in\Rp$, 
\ShortVersion{\emph{$\eval(\rho,\tau)$} is the constraint 
 \begin{math}
  \eval(\rho,\tau) = \bigwedge_{\clock\in\Clock} \bigl(\clock = \tau - \rho(\clock)\bigr)
 \end{math}
on $\Clock\disjointUnion\{t\}$.}
\LongVersion{\emph{$\eval(\rho,\tau)$} is the following constraint on $\Clock\disjointUnion\{t\}$.
 \begin{displaymath}
  \eval(\rho,\tau) = \bigwedge_{\clock\in\Clock} \bigl(\clock = \tau - \rho(\clock)\bigr)
 \end{displaymath}}
\end{definition}

\begin{definition}
 [$\mathrm{reset}(\rho,\resets,\tau)$]
 For a function $\rho\colon \Clock\to(\Rp\disjointUnion\{t\})$, the set $\resets\subseteq\Clock$ of clocks to be reset, and the current time $\tau\in\Rp$, 
 $\mathrm{reset}(\rho,\resets,\tau)\colon\Clock\to(\Rp\disjointUnion\{t\})$ is the following function.
 \begin{displaymath}
  \mathrm{reset}(\rho,\resets,\tau) (\clock) =
  \begin{cases}
   \tau & \text{if $\clock\in\resets$}\\
   \rho(\clock)& \text{if $\clock\not\in\resets$}
  \end{cases}
 \end{displaymath}
\end{definition}

\begin{definition}
 [$\rhoEmpty$]
 By $\rhoEmpty\colon \Clock\to(\Rp\disjointUnion\{t\})$, we denote the function mapping each $\clock\in\Clock$ to $t$.
\end{definition}

Intuitively, $\eval(\rho,\tau)$ is the constraint corresponding to the clock valuation, $\mathrm{reset}(\rho,\resets,\tau)$ is the operation to reset the clock variables $\clock\in\resets$ at $\tau$, and $\rhoEmpty$ is the initial clock valuation.

\begin{algorithm}
 \caption{Online parametric timed pattern matching without skipping}%
 \label{alg:online_no_skip}
 \DontPrintSemicolon{}
 \newcommand{\myCommentFont}[1]{\texttt{\footnotesize{#1}}}
 \SetCommentSty{myCommentFont}
 \KwIn{A timed word $\word = (\overline{a},\overline{\tau})$, and a PTA $\A = (\Actions, \Loc, \locinit, \LocFinal, \Clock, \Param, %
 \Edges)$.}
 \KwOut{$\bigvee Z$ is the match set $\mathcal{M} (w,\mathcal{A})$}
 
 $\CurrConf \gets \emptyset;\; Z\gets\emptyset$\;
\For{$i \gets 1\, \mathbf{to}\, |w|$} {%
\label{alg_line:no_skip:big_for_beginning}
  \KwPush{} $(\locinit,\rhoEmpty,(\tau_{i-1} < t \leq \tau_i))$ \KwTo{} $\CurrConf$\label{alg_line:no_skip:insert_beginning}\tcp*[f]{start a matching trial from $(\tau_{i-1}, \tau_{i}]$}

  \For(\tcp*[f]{\crefrange{online_alg:insertTermBegin}{online_alg:insertTermEnd} try to insert \$ in $[\tau_{i-1},\tau_{i})$}){$(\loc, \rho,\mathcal{C}) \in \CurrConf$} {\label{online_alg:insertTermBegin}
    \For{$\loc_f \in \LocFinal$} {\label{naive_alg_online:start_of_accepting_state}
      \For{$(\loc, \guard,\$,\resets,\loc_f) \in \Edges$} {\label{naive_alg_online:loop_over_e_2}
        \KwPush{} $\project{ \big(\mathcal{C} \land (\tau_{i-1} \leq t' < \tau_{i}) \land \guard \land \eval(\rho,t') \big)}{\Param \disjointUnion\{t,t'\}}$ \KwTo{} $Z$\;%
        \label{naive_alg_online:add_interval}
      }
    }
  }\label{online_alg:insertTermEnd}
  $(\PrevConf , \CurrConf) \gets (\CurrConf , \emptyset)$\label{alg_line:no_skip:after_intert_term}\;
  \For(\tcp*[f]{\crefrange{online_alg:readLoopBegin}{online_alg:readLoopEnd} try to go forward using $(a_{i},\tau_{i})$}){$(\loc, \rho,\mathcal{C}) \in \PrevConf$} {\label{online_alg:readLoopBegin}
    \For{$(\loc, \guard,a_i,\resets,\loc') \in \Edges$} {
      $\mathcal{C}' \gets \project{\big (\mathcal{C} \land \guard \land \eval(\rho,\tau_i) \big) }{\Param \disjointUnion\{t\}}$\;\label{alg_line:no_skip:examine_guard}
      \If{$\mathcal{C}' \neq \KFalse$} {\label{alg_line:no_skip:branch_examine_guard}
        \KwPush{} $(\loc', \mathrm{reset}(\rho,\resets,\tau),\mathcal{C}')$ \KwTo $\CurrConf$\;\label{alg_line:no_skip:push_next_conf}
      }
    }
 }
}\label{online_alg:readLoopEnd}\label{alg_line:no_skip:big_for_end}
\KwPush{} $(\locinit,\rhoEmpty,\{\tau_{|w|} < t < \infty\})$ \KwTo{} $\CurrConf$
\tcp*[f]{for the trimming after the final event}

 \For(\tcp*[f]{\crefrange{online_alg:insertTermRemainBegin}{online_alg:insertTermRemainEnd} try to insert \$ in $[\tau_{|w|},\infty)$}){$(\loc, \rho,\mathcal{C}) \in \CurrConf$}{\label{online_alg:insertTermRemainBegin}
   \For{$\loc_f \in \LocFinal$}{
     \For{$(\loc, \guard,\$,\resets,\loc_f) \in \Edges$}{
       \KwPush{} $\project{\big(\mathcal{C} \land (\tau_{|w|} \leq t' < \infty) \land \guard \land \eval(\rho,t') \big)}{\Param \disjointUnion\{t,t'\}}$ \KwTo{} $Z$\;
     }
   }
}%
\label{online_alg:insertTermRemainEnd}
\end{algorithm}

\cref{alg:online_no_skip} shows our online algorithm for parametric timed pattern matching.
In the pseudocode, we used $\CurrConf$, $\PrevConf$, and $Z$: $\CurrConf$ and $\PrevConf$ are finite sets of triples $(\loc, \rho,\mathcal{C})$ of a location $\loc\in\Loc$, a mapping $\rho\colon\Clock\to(\Rp\disjointUnion\{t\})$ denoting the latest reset of each clock, and a constraint $\mathcal{C}$ over $\Param \disjointUnion\{t\}$; and $Z$ is a finite set of constraints over $\Param \disjointUnion\{t,t'\}$.
As a running example, we use the PTA and the timed word in \cref{figure:running-example} page~\pageref{figure:running-example}.

At first, the counter $i$ is $1$ (\cref{alg_line:no_skip:big_for_beginning}), and we start the matching trial from $t\in(\tau_0,\tau_1]$.
At \cref{alg_line:no_skip:insert_beginning}, we add the new configuration $(\locinit, \rhoEmpty, (\tau_0 < t \leq \tau_1))$ to $\CurrConf$, which means we are at the initial location $\locinit$, we have no reset of the clock variables yet, and we can potentially start the matching from any $t\in (\tau_0,\tau_1]$.
In \crefrange{online_alg:insertTermBegin}{online_alg:insertTermEnd}, we try to insert \$ (\ie{} the end of the matching) in $[\tau_{0},\tau_{1})$;
in our running example in \cref{figure:running-example}, since there is no edge from $\locinit$ to the final state, we immediately jump to \cref{alg_line:no_skip:after_intert_term}.
Then, in \crefrange{online_alg:readLoopBegin}{online_alg:readLoopEnd}, we consume $(a_1,\tau_1)= (a,0.7)$ and try to transit from~$\locinit$ to~$\loc_1$.
The guard ${\clock} > 1$ at the edge from $\locinit$ to $\loc_1$ is examined at \cref{alg_line:no_skip:examine_guard}.
We take the conjunction of the current constraint $\mathcal{C}$, the guard $\guard$, and the constraints $\eval(\rho,\tau_{i})$ on the clock valuations.
We take the projection to $\Param\disjointUnion\{t\}$ because the constraint on the clock variables changes after time passing.
Since no clock variable is reset so far, the constraint on the clock valuation is $\clock = \tau_1 - t$.
The constraint $\mathcal{C} \land \guard \land \eval(\rho,\tau_1)=(0 < t \leq 0.7)\land(x > 1)\land(x=0.7-t)$ is unsatisfiable, and we go back to \cref{alg_line:no_skip:insert_beginning}.

At \cref{alg_line:no_skip:insert_beginning}, we add the new configuration $(\locinit, \rhoEmpty, (\tau_1 < t \leq \tau_2))$ to $\CurrConf$.
Similarly, we immediately jump to \cref{alg_line:no_skip:after_intert_term}, and we try the edge from $\locinit$ to $\loc_1$ in \crefrange{online_alg:readLoopBegin}{online_alg:readLoopEnd}.
This time, the constraint $\mathcal{C} \land \guard \land \eval(\rho,\tau_2)=(0.7 < t \leq 2.0)\land(x > 1)\land(x=2.0-t)$ is satisfiable at \cref{alg_line:no_skip:branch_examine_guard}, and we push the next configuration $(\loc_1,\rhoEmpty,\mathcal{C}')$ to $\CurrConf$ at \cref{alg_line:no_skip:push_next_conf}.

Similarly, we keep adding and updating configurations until the end of the input timed word $\word$.
Finally, in \crefrange{online_alg:insertTermRemainBegin}{online_alg:insertTermRemainEnd}, we try to insert \$ in $[\tau_{3},\infty)=[4.1,\infty)$.
We can use the edge from $\loc_2$ to the final state, and we add the constraint at the right of \cref{figure:running-example} to $Z$.

\cref{alg:online_no_skip} terminates because the size of $\CurrConf$ is always finite.
\cref{alg:online_no_skip} is correct because it symbolically keeps track of all the runs of $\valuate{\A}{\pval}$ over $w|_{(t,t')}$ for any $\pval\in\PVal$ and $(t,t') \subseteq \Rnn$.

\subsection{Parametric timed pattern matching enhanced with skipping}\label{section:skipping}
We now enhance \cref{alg:online_no_skip} with automata-based \emph{skipping} that has been used for (non-parametric) timed pattern matching~\cite{WAH16,WHS17}.
Skipping is an optimization for pattern matching by reducing the number of the matching trials.
The high-level idea of skipping is to utilize the specification automaton and the information from the latest matching trial to decide if the current matching trial is necessary, and prevent the matching trial if it is unnecessary.
More precisely, in the beginning of each matching trial (at \cref{alg_line:no_skip:big_for_beginning} of \cref{alg:online_no_skip}),
instead of increasing the counter~$i$ by one,
we compute the smallest $\Delta \in \Zp$ such that there may be a matching from $t \in (\tau_{i+\Delta-1}, \tau_{i+\Delta}]$, and update the counter~$i$ to $i+\Delta$.
We call such $\Delta$ a \emph{skip value}.
We compute the skip value $\Delta$ using a function that we call a \emph{skip value function}.
When the skip value $\Delta$ is greater than one, some matching trials are prevented.
Although the computation of such a skip value $\Delta$ requires some additional computational cost, a large part of the computation can be done beforehand thanks to the finiteness of the automata structure.
Therefore, we can gain the runtime performance, which we confirm by our experiments in \cref{ss:adhoc-experiments}.

 Following~\cite{WHS17}, we employ \emph{FJS-style}\footnote{``FJS'' is the acronym of the family names of the authors of~\cite{FJS07}.} skipping~\cite{FJS07}.
 In FJS-style skipping, the skip value $\Delta$ is computed by combining the result of two skip value functions: the KMP-style\footnote{``KMP'' is the acronym of the family names of the authors of~\cite{KMP77}.} skip values function $\KMPSkipFunc$~\cite{KMP77} and the Quick Search-style skip value function $\QSSkipFunc$~\cite{Sunday90}.
Their combination in our FJS-style algorithm is presented later in \cref{subsec:skipping_algorith}.
We remark that this optimization does not change the result thanks to the soundness of the skip value functions (\cref{theorem:KMP,theorem:QS}).%
\label{paragraph:explain_FJS}

The following are auxiliaries for the skip values.
For a PTA $\A$ and a parameter valuation $\pval$, the language without the last element is denoted by $\Lg_{-\$}(\valuate{\A}{\pval}) = \{\word(1,|\word|-1)\mid\word\in\Lg(\valuate{\A}{\pval})\}$.
For a PTA $\A = (\Actions, \Loc, \locinit, \LocFinal, \Clock, \Param, %
 \Edges)$ and $\loc\in\Loc$,
$\A_{\loc}$ denotes the PTA
$\A_\loc=(\Actions, \Loc, \locinit, \{\loc\}, \Clock, \Param, %
 \Edges)$.

\subsubsection{KMP-style skip values}
Given a location $\loc\in\Loc$ and a set $\pvals\subseteq\PVal$ of parameter valuations, the \emph{KMP-style skip value function} $\KMPSkipFunc$ returns the skip value $\KMPSkipFunc(\loc, \pvals)\in\Zp$.
The location $\loc$ and the parameter valuations $\pvals$ represent %
a set of the configurations in the latest matching trial.
We utilize the pair $(\loc, \pvals)$ to overapproximate the subsequence $\word(i,j) - \tau_{i}$ of the timed word $\word$ examined in the latest matching trial.
The following illustrates the idea of the KMP-style skipping.
\begin{example}%
 \label{example:KMP}
 \begin{figure}[tbp]
  \begin{tikzpicture}[scale=1.8,xscale=2.0]
   \draw [thick, -stealth](-0.5,0)--(1.5,0);

   \def\diff{0.15}
   \def\dotsY{0.2}
   \coordinate (a00) at (-0.1,0.1);
   \coordinate (a0) at (0.0,0.1);
   \coordinate (a1) at (0.65,0.1);
   \coordinate (a2) at (1.0,0.1);

   \node at (0-\diff, \dotsY) {$\cdots$};
   \node at (1+\diff, \dotsY) {$\cdots$};

   \draw[bend left=60]
         let \p0 = (a0) in 
         let \p1 = (a1) in
               (\x0, 0) edge [above] node{$1 < $} (\x1, 0);
   \draw[bend left=60]
         let \p1 = (a1) in
         let \p2 = (a2) in
           (\x1, 0) edge [above] node{$< 1$} (\x2, 0);
   \draw[bend left=80,distance=0.7cm]
         let \p0 = (a0) in
         let \p2 = (a2) in
           (\x0, 0) edge [above] node{$< \valuate{\styleparam{\param}}{\pval} = 2$} (\x2, 0);

   \draw let \p0 = (a00) in
           (a00) -- (\x0, -\y0) node[anchor=north]{$\tau_{i}$};
   \draw[dashed,color=gray] let \p0 = (a0) in
           (a0) -- (\x0, -\y0);
   \foreach \i in {1, 2}
     \draw let \p{\i} = (a\i) in
       (a\i) node[anchor=south]{$\styleact{a}$} -- (\x\i,-\y\i) node[anchor=north]{$\tau_{i + \i}$};

   \def\shiftY{-0.5}
   \draw [thick, -stealth](0.5,0+\shiftY)--(1.7,+\shiftY);

   \def\diff{0.15}
   \def\dotsY{0.2}
   \coordinate (b0) at (0.65,0.1+\shiftY);
   \coordinate (b1) at (0.75,0.1+\shiftY);
   \coordinate (b2) at (1.0,0.1+\shiftY);
   \coordinate (b3) at (1.3,0.1+\shiftY);

   \draw[bend right=60,distance=0.1cm]
         let \p0 = (b1) in 
         let \p1 = (b2) in
               (\x0, \shiftY) edge [below,pos=0.6] node{$1 < $} (\x1, \shiftY);
   \draw[bend right=80,distance=0.5cm]
         let \p0 = (b1) in
         let \p2 = (b3) in
           (\x0, \shiftY) edge [below] node{$< \valuate{\styleparam{\param}}{\pval'}$} (\x2, \shiftY);

   \foreach \i in {0, 2, 3}
     \draw let \p{\i} = (b\i) in
       (b\i) -- (\x\i,-0.1+\shiftY);
   \draw[dashed,color=gray] let \p0 = (b1) in
       (b1) -- (\x0,-0.1+\shiftY);

   \def\shiftY{-0.65}
   \draw [thick, -stealth](0.5,0+\shiftY+\shiftY)--(1.9,+\shiftY+\shiftY);

   \def\diff{0.15}
   \def\dotsY{0.2}
   \coordinate (c0) at (1.0,0.1+\shiftY+\shiftY);
   \coordinate (c2) at (1.1,0.1+\shiftY+\shiftY);
   \coordinate (c3) at (1.3,0.1+\shiftY+\shiftY);
   \coordinate (c4) at (1.65,0.1+\shiftY+\shiftY);

   \draw[bend right=60,distance=0.1cm]
         let \p0 = (c2) in 
         let \p1 = (c3) in
               (\x0, \shiftY+\shiftY) edge [below,pos=0.7] node{$1 < $} (\x1, \shiftY+\shiftY);
   \draw[bend right=80,distance=0.5cm]
         let \p0 = (c2) in
         let \p2 = (c4) in
           (\x0, \shiftY+\shiftY) edge [below] node{$< \valuate{\styleparam{\param}}{\pval'}$} (\x2, \shiftY+\shiftY);

   \foreach \i in {0, 3, 4}
     \draw let \p{\i} = (c\i) in
       (c\i) -- (\x\i,-0.1+\shiftY+\shiftY);
   \draw[dashed,color=gray] let \p0 = (c2) in
       (c2) -- (\x0,-0.1+\shiftY+\shiftY);

  \end{tikzpicture}
  \caption{Illustration of the timing constraints in \cref{example:KMP}: constraints from the reached location $\loc_2$ and the parameter valuation $\pval$ (above), necessary timing constraints to have a matching after a shift by 1 (middle), and necessary timing constraints to have a matching after a shift by 2 (below)}%
  \label{figure:illustrates_KMP}
 \end{figure}

 Let $\A$ be the PTA in \cref{figure:running-example:PTA}.
 Suppose our latest matching trial from $t \in (\tau_i, \tau_{i+1}]$
 finished at $\loc_2$ with the parameter valuation $\pval$ satisfying $\pval(\styleparam{\param}) = 2$.
 By the guards of $\A$, there is $t \in (\tau_i,\tau_{i+1}]$ satisfying 
 $1 < \tau_{i+1} - t \land \tau_{i+2} - t < \valuate{\styleparam{\param}}{\pval} = 2 \land \tau_{i+2} - \tau_{i+1} < 1$.
 The above of \cref{figure:illustrates_KMP} illustrates this timing constraint.

If we have a matching from $t \in (\tau_{i+1},\tau_{i+2}]$, because of the guards in $\A$, 
there are $t \in (\tau_{i+1},\tau_{i+2}]$ and a parameter valuation $\pval'$ satisfying $1 < \tau_{i+2} - t \land \tau_{i+3} - t < \valuate{\styleparam{\param}}{\pval'}$.
The middle of \cref{figure:illustrates_KMP} illustrates this timing constraint.
 Since we have both $\tau_{i+2} - \tau_{i+1} < 1$ and $1 < \tau_{i+2} - t < \tau_{i+2} - \tau_{i+1}$, the timing constraints in the above and the middle of \cref{figure:illustrates_KMP} does not hold at the same time.
 Therefore, we conclude that there is no matching from any point between $\tau_{i+1}$ and $\tau_{i+2}$ and we can skip this matching trial.

 In contrast, if we have a matching from $t \in (\tau_{i+2},\tau_{i+3}]$, 
there are $t \in (\tau_{i+2},\tau_{i+3}]$ and a parameter valuation $\pval'$ satisfying
$1 < \tau_{i+3} - t \land \tau_{i+4} - t < \valuate{\styleparam{\param}}{\pval'}$.
The below of \cref{figure:illustrates_KMP} illustrates this timing constraint.
 Since the timing constraints in the above and the below of \cref{figure:illustrates_KMP} are satisfiable at the same time, we may have a matching from $t \in (\tau_{i+2}, \tau_{i+3}]$, and we cannot skip this matching trial.
Overall, we can increase the counter~$i$ in \cref{alg:online_no_skip} at most by 2 without changing the result.
\end{example}

We formalize this idea of the KMP-style skip value function with the languages of PTAs.
Let $i \in \{1,2,\dots,|\word|\}$ and
let $\loc$ be a location we reached in the end of the matching trial from $i$ with the parameter valuation $\pval$.
Since we reached $\loc$ with the parameter valuation $\pval$ by reading a prefix of $\word(i,|\word|)$, 
there is an index $j \geq i$ satisfying $\word(i,j) - \tau_{i} \in \Lg(\valuate{\A_\loc}{\pval})$.\footnote{%
	More precisely, we require $\loc \not\in \LocFinal$ because the transition to the final locations are labeled with the special terminal character \$, which is not in $\word$.}
By appending an arbitrary timed word, we overapproximate the subsequence $\word(i,|\word|) - \tau_{i}$ by the language $\Lg(\valuate{\A_\loc}{\pval})\cdot\TW(\Actions)$.

If there is a matching from the $i$-th event $(\action_i, \tau_i)$ of $\word$ with a parameter valuation $\pval'$, \ie{}
if there is $t \in (\tau_{i}, \tau_{i+1}]$ and $t' > t$ satisfying $\word|_{(t,t')} \in \Lg(\valuate{\A}{\pval})$,
there is an index $j \geq i$ and $t' \in [\tau_{j}, \tau_{j+1})$ satisfying $\bigl(\word(i,j) \circ (\$,t')\bigr) - t \in \Lg(\valuate{\A}{\pval})$.
By removing the terminal character \$, appending an arbitrary timed word, and shifting appropriately,
we have $\word(i,|\word|) \in \{\word''+t\mid\word''\in\Lg_{-\$}(\pval'(\A)),t \geq 0\}\cdot\TW(\Actions)$, and thus,
for each $n\in\Zp$, the matching after $n$-shift with a parameter valuation $\pval'$ is overapproximated by
$\TW^n(\Actions)\cdot\{\word''+t\mid\word''\in\Lg_{-\$}(\pval'(\A)),t \geq 0\}\cdot\TW(\Actions)$.
Therefore, we can skip the matching trial from $i + n$ if for any parameter valuation $\pval'$, and we have the following.
\[
 \Lg(\valuate{\A_\loc}{\pval})\cdot\TW(\Actions) \cap \TW^n(\Actions)\cdot\{\word''+t\mid\word''\in\Lg_{-\$}(\pval'(\A)),t \geq 0\}\cdot\TW(\Actions) = \emptyset
\]
Based on the above idea, the KMP-style skip value function $\KMPSkipFunc$ is defined as follows.

\begin{definition}
 [$\KMPSkipFunc$]%
 \label{def:KMP}
 Let $\A$ be a PTA $\A = (\Actions, \Loc, \locinit, \LocFinal, \Clock, \Param, %
 \Edges)$.
 For a location $\loc\in\Loc$ and $n\in\Zp$, let $\pvalsi{\loc,n}$ be the set of parameter valuations $\pval$ such that there is a parameter valuation $\pval'\in\PVal$ satisfying
 $\Lg(\valuate{\A_\loc}{\pval})\cdot\TW(\Actions)\cap\TW^n(\Actions)\cdot\{\word''+t\mid\word''\in\Lg_{-\$}(\pval'(\A)),t \geq 0\}\cdot\TW(\Actions)\ne\emptyset$.
 The \emph{KMP-style skip value function} $\KMPSkipFunc\colon\Loc\times\powerset{\PVal}\to\Zp$ is 
 $\KMPSkipFunc(\loc, \pvals) = \min\{n\in \Zp\mid \pvals \subseteq \pvalsi{\loc,n} \}$.
\end{definition}
\begin{example}
 We continue \cref{example:KMP}.
 For a parameter valuation $\pval$, we have the following.
 \[
 \Lg(\valuate{\A_{\loc_2}}{\pval}) = \Lg_{-\$}(\valuate{\A}{\pval}) = \{(\styleact{a}, \tau_1),(\styleact{a}, \tau_2) \mid 1 < \tau_1 < \tau_2 < \valuate{\styleparam{\param}}{\pval} \}
 \]
 Therefore, for any parameter valuations $\pval, \pval'$, we have 
 $\Lg(\valuate{\A_{\loc_2}}{\pval})\cdot\TW(\Actions)\cap\TW^n(\Actions)\cdot\{\word''+t\mid\word''\in\Lg_{-\$}(\pval'(\A)),t \geq 0\}\cdot\TW(\Actions)\ne\emptyset$ if and only if the following is satisfiable.
 \begin{equation}
 \exists t \in (\tau_{n}, \tau_{n+1}].\,
  1 < \tau_1 < \tau_2 < \valuate{\styleparam{\param}}{\pval} \land 
  1 < \tau_{n+1} - t < \tau_{n+2} - t < \valuate{\styleparam{\param}}{\pval'}  
  \label{eq:KMP_example}
 \end{equation}
When $\valuate{\styleparam{\param}}{\pval} = 2$ and $n = 1$,
we have $1 < \tau_2 - \tau_1 < 1$, and \cref{eq:KMP_example} is unsatisfiable.
In contrast, when $\valuate{\styleparam{\param}}{\pval} = 2$ and $n = 2$, \cref{eq:KMP_example} is satisfiable, and we have $\KMPSkipFunc(\loc_2, \{\pval\}) = 2$, which coincides with the idea in \cref{example:KMP}.
\end{example}
We note that if we reached both $\loc$ and $\loc'$, 
we have both $\word(i,|\word|) \in \Lg(\valuate{\A_\loc}{\pval})\cdot\TW(\Actions)$ and $\word(i,|\word|) \in \Lg(\valuate{\A_{\loc'}}{\pval})\cdot\TW(\Actions)$, and 
the intersection 
$(\Lg(\valuate{\A_\loc}{\pval})\cdot\TW(\Actions)) \cap (\Lg(\valuate{\A_{\loc'}}{\pval})\cdot\TW(\Actions))$ overapproximates $\word(i,|\word|)$.
Therefore, 
we take the maximum of $\KMPSkipFunc(\loc,\pvals)$ over the reached configurations.

Since $\pvalsi{\loc,n}$ is independent of the timed word $\word$, we can compute it before the matching trials by reachability synthesis of PTAs. 
See \cref{appendix:automatic_construction} for the construction of the PTAs.%
During the matching trials, only the inclusion checking $\pvals\subseteq\pvalsi{\loc,n}$ is necessary.
This test can be achieved thanks to convex polyhedra inclusion.
%

\paragraph{Soundness}
For the KMP-style skip value function $\KMPSkipFunc$, we have the following soundness result, and the use of $\KMPSkipFunc$ does not change the result of parametric timed pattern matching.

\begin{theorem}
 [soundness of $\KMPSkipFunc$]%
 \label{theorem:KMP}
 Let $\A = (\Actions, \Loc, \locinit, \LocFinal, \Clock, \Param, \Edges)$ be a PTA and
 let $\word\in\TW(\Actions)$.
 For any subsequence $\word(i,j)$ of $\word$ and for any
 $(\loc, \pval)\in\Loc\times\PVal$, if there exists $t\in\Rnn$ satisfying $\word(i,j)-t\in\Lg(\valuate{\A_\loc}{\pval})$,
 for any $n\in\bigl\{1,2,\dots,\KMPSkipFunc(\loc, \{\pval\})-1\bigr\}$,
 we have $\bigl((\tau_{i+n-1},\tau_{i+n}]\times\Rp\times\PVal\bigr)\cap\mathcal{M}(w,\mathcal{A})=\emptyset$.
\end{theorem}

First, we prove the following lemma.

\begin{lemma}
 \label{lemma:KMP}
 Let $\A$ be a PTA $\A = (\Actions, \Loc, \locinit, \LocFinal, \Clock, \Param, %
 \Edges)$ and
 let $\word\in\TW(\Actions)$.
 For each subsequence $\word(i,j)$ of $\word$, let $C_{i,j}\subseteq \Loc\times\PVal$ be $C_{i,j}=\{(\loc, \pval)\mid \exists t\in\Rnn.\, \word(i,j)-t \in \Lg(\valuate{\A_\loc}{\pval})\}$.
 For any subsequence $\word(i,j)$ of $\word$ and $n\in\{1,2,\dots,|\word| - i\}$,
 if there exists $(\loc, \pval)\in C_{i,j}$ satisfying
 $\pval\not\in \pvalsi{\loc,n}$, we have $\bigl((\tau_{i+n-1},\tau_{i+n}]\times\Rp\times\PVal\bigr)\cap\mathcal{M}(w,\mathcal{A})=\emptyset$.
\end{lemma}

\begin{proof}
 For any $(\loc, \pval)\in C_{i,j}$ satisfying $\pval\not\in\pvalsi{\loc,n}$, by definition of $C_{i,j}$ and $\pvalsi{\loc,n}$, we have the following.
\begin{itemize}
 \item $\exists t\in\Rnn.\, \word(i,j) - t \in \Lg(\valuate{\A_{\loc}}{\pval})$
 \item $\forall\pval'\in\PVal.\, \Lg(\valuate{\A_\loc}{\pval})\cdot\TW(\Actions)\cap \TW^n(\Actions)\cdot \{\word'' + t \mid \word''\in\Lg_{-\$}(\valuate{\A}{\pval'}),t>0\}\cdot\TW(\Actions)=\emptyset$
\end{itemize}

 Therefore, we have the following.
 \small
 \begin{align*}
  &\exists(\loc, \pval)\in C_{i,j}.\,\pval\not\in\pvalsi{\loc,n}\\
  \Rightarrow&\exists t\in\Rp.\forall \pval'\in\PVal.\, \\
  &\,(\word(i,j)-t)\cdot\TW(\Actions)\cap\TW^n(\Actions)\cdot\{\word''+t'\mid\word''\in\Lg_{-\$}(\valuate{\A}{\pval'}),t'>0\}\cdot\TW(\Actions)=\emptyset\\
  \Rightarrow&\forall \pval'\in\PVal.\,
  \word(i,j)\cdot\TW(\Actions)\cap\TW^n(\Actions)\cdot\{\word''+t\mid\word''\in\Lg_{-\$}(\valuate{\A}{\pval'}),t>0\}\cdot\TW(\Actions)=\emptyset\\
  \Rightarrow&\forall \pval'\in\PVal.\,
  \word(i,|\word|)\cdot\TW(\Actions)\cap\TW^n(\Actions)\cdot\{\word''+t\mid\word''\in\Lg_{-\$}(\valuate{\A}{\pval'}),t>0\}\cdot\TW(\Actions)=\emptyset\\
  \Rightarrow&\forall \pval'\in\PVal.\, \word(i,|\word|)\not\in\TW^n(\Actions)\cdot\{\word''+t\mid\word''\in\Lg_{-\$}(\valuate{\A}{\pval'}),t>0\}\cdot\TW(\Actions)\\
  \Rightarrow&\forall \pval'\in\PVal.\, \word(i,|\word|)\not\in\TW^n(\Actions)\circ\{\word''+t\mid\word''\in\Lg_{-\$}(\valuate{\A}{\pval'}),t\in(\tau_{i+n-1},\tau_{i+n}]\}\cdot\TW(\Actions)\\
  \Rightarrow&\forall \pval'\in\PVal.\, \word(i+n,|\word|)\not\in\{\word''+t\mid\word''\in\Lg_{-\$}(\valuate{\A}{\pval'}),t\in(\tau_{i+n-1},\tau_{i+n}]\}\cdot\TW(\Actions)\\
  \iff&\forall t\in(\tau_{i+n-1},\tau_{i+n}],\pval'\in\PVal.\, \word(i+n,|\word|)-t\not\in\{\word''\mid\word''\in\Lg_{-\$}(\valuate{\A}{\pval'})\}\cdot\TW(\Actions)\\
  \iff &\forall t\in(\tau_{i+n-1},\tau_{i+n}],t'>t,\pval'\in\PVal.\, \word|_{(t,t')}\not\in\Lg(\valuate{\A}{\pval'})\\
  \iff &((\tau_{i+n-1},\tau_{i+n}]\times\Rp\times\PVal)\cap\mathcal{M}(w,\A) = \emptyset
 \end{align*}
\end{proof}

Then, the proof of \cref{theorem:KMP} is as follows.

\begin{proof}
 By definition of $C_{i,j}$ in \cref{lemma:KMP}, for any subsequence $\word(i,j)$ of $\word$ and $(\loc, \pval)\in\Loc\times\PVal$,
 $\exists t\in\Rnn.\,\word(i,j)-t\in\Lg(\valuate{\A_{\loc}}{\pval})$ implies $(\loc, \pval)\in C_{i,j}$.
 By definition of $\KMPSkipFunc$, for any $(\loc, \pval)\in C_{i,j}$ and $n\in\{1,2,\ldots,\KMPSkipFunc(\loc, \{\pval\})-1\}$, we have $\pval\not\in\pvalsi{i,n}$. Therefore, \cref{theorem:KMP} holds because of \cref{lemma:KMP}.
\end{proof}

\paragraph{Discussion}
Although $V_{\loc,n}$ can be computed before the matching trials, \LongVersion{the KMP-style skip value function }$\KMPSkipFunc$ requires checking $V \subseteq V_{\loc,n}$ after each matching trial, which means a polyhedral inclusion test in $|\Param|+2$ dimensions.
To reduce this runtime overhead, we define the \emph{non-parametric} KMP-style skip value function $\KMPSkipFunc'(\loc) = \min_{\pval\in\PVal} \KMPSkipFunc\bigl(\loc, \{\pval\}\bigr)$.
For comparison, we refer to $\KMPSkipFunc$ as the \emph{parametric} KMP-style skip value function.

\subsubsection{Quick Search-style skip values}

In Quick Search-style skipping, we ignore the timing constraints and only utilize the actions.
Given an action $\action\in\Actions$, the \emph{Quick Search-style skip value function} $\QSSkipFunc$ returns the skip value $\QSSkipFunc(\action)\in\Zp$.
Before the matching trial from the $i$-th element $(\action_i,\tau_i)$, we look ahead the action $\action_{i+N-1}$, where $N$ is the length of the shortest matching.
If we observe that there is no potential matching, we also look ahead the action $\action_{i+N}$ and skip by $\QSSkipFunc(\action_{i+N})$.

\begin{example}%
 \label{example:QS}
 Let $\A$ be the PTA in \cref{figure:patterns:accel}.
 The length of the shortest matching $N$ is five.
 For any parameter valuation $\pval$ and timed word $(\action'_1,\tau_1),(\action'_2,\tau'_2),\dots,(\action'_n,\tau'_n) \in \Lg(\valuate{\A}{\pval})$,
 we have $\action'_{1} \in \{\styleact{g_1}, \styleact{\text{rpmHigh}}\}$,
 $\action'_{2} \in \{\styleact{g_1}, \styleact{g_2}, \styleact{\text{rpmHigh}}\}$,
 $\action'_{3} \in \{\styleact{g_2}, \styleact{g_3}, \styleact{\text{rpmHigh}}\}$,
 $\action'_{4} \in \{\styleact{g_3}, \styleact{g_4}, \styleact{\text{rpmHigh}}\}$, and
 $\action'_{5} \in \{\styleact{g_4}, \styleact{\text{rpmHigh}}\}$.

 Suppose we are about to start the matching trial from a point between $\tau_{0}$ and $\tau_{1}$ (\ie{} $i = 1$) and the actions in the timed word $\word$ against pattern matching are
 $a_1 = \styleact{g_1}$, $a_2 = \styleact{g_2}$, $a_3 = \styleact{g_1}$,  $a_4 = \styleact{g_2}$, $a_5 = \styleact{g_1}$, and  $a_6 = \styleact{g_2}$.
 First, we look ahead the action $a_{i + N - 1} = a_{5} = \styleact{g_1}$ and check if $\styleact{g_1}$ can be the 5th action in a matching, which is not the case because the 5th action in a matching must be in $\{\styleact{g_4}, \styleact{\text{rpmHigh}}\}$.
 Then, we move to the position where we can have a matching according to $a_{1 + N} = a_{6}$.
 Since $\action_{6} = \styleact{g_2}$ cannot be the 4th nor 5th actions in a matching but can be the 3rd action in a matching, we move to the position where $\action_{6}$ is the 3rd action in the matching trial, \ie{} $i = 4$.
 Therefore, the Quick Search-style skip value is $\QSSkipFunc(\styleact{g_2}) = 4 - 1 = 3$.
\end{example}

We formalize this idea with the \emph{untimed} language of PTAs.
For $i \in \{1,2,\dots,|\word|\}$, \LongVersion{the subsequence }$\word(i,|\word|)$ is overapproximated by
$\Actions^N \action_{i+N} \Actions^*$.
For $n\in\Zp$, the matching from $i+n$ is overapproximated by 
$\bigcup_{\pval \in \PVal} \Actions^n \untimed{ \Lg_{-\$}(\pval(\A))}$.
Thus, we have no matching from $i + n$ if for any $\pval \in \PVal$, we have 
$\Actions^N \action_{i+N} \Actions^* \cap \Actions^n \untimed{ \Lg_{-\$}(\pval(\A))} = \emptyset$.
The definition of the Quick Search-style skip value function $\QSSkipFunc$ is as follows.

\begin{definition}
 [$\QSSkipFunc$]
 For a PTA $\A = (\Actions, \Loc, \locinit, \LocFinal, \Clock, \Param, %
 \Edges)$,
 the \emph{Quick-Search-style skip value function} $\QSSkipFunc\colon\Actions\to\Zp$ is as follows, where $N\in\Zp$ is $N=\min\{|\word| \mid \word\in\bigcup_{\pval\in\PVal}\Lg_{-\$}(\valuate{\A}{\pval})\}$.
 \[
  \QSSkipFunc(\action) = \min\bigl\{n\in \Zp \,\bigm|\, \exists \pval \in \PVal.\,  \Actions^N \action \Actions^*\cap\Actions^n \untimed{\Lg_{-\$}(\valuate{\A}{\pval})} \ne \emptyset\bigr\}
 \]
\end{definition}

The construction of~\LongVersion{the Quick Search-style skip value function} $\QSSkipFunc$ is by reachability emptiness of PTAs, \ie{} the emptiness of the valuation set reaching a given location.

\paragraph{Soundness}
For the Quick Search-style skip value function $\QSSkipFunc$, we have the following soundness result, and the use of $\QSSkipFunc$ does not change the result of parametric timed pattern matching.

\begin{theorem}
 [soundness of $\QSSkipFunc$]%
 \label{theorem:QS}
 Let $\A$ be a PTA $\A = (\Actions, \Loc, \locinit, \LocFinal, \Clock, \Param, %
 \Edges)$,
 let $\word=(\overline{a}, \overline{\tau})\in\TW(\Actions)$, and
 let $N=\min\{|\word|\mid \word\in\bigcup_{\pval\in\PVal}\Lg_{-\$}(\valuate{\A}{\pval})\}$.
 For any index $i\in\{1,2,\dots,|\word|\}$ of $\word$ and for any $m\in\{1,2,\dots,\QSSkipFunc(\action_{i+N})-1\}$,
 we have $((\tau_{i+m-1},\tau_{i+m}]\times\Rp\times\PVal)\cap\mathcal{M}(w,\mathcal{A})=\emptyset$.
\end{theorem}

First, we prove the following lemma.

\begin{lemma}
 \label{lemma:QS}
 Let $\A$ be a PTA $\A = (\Actions, \Loc, \locinit, \LocFinal, \Clock, \Param, %
 \Edges)$,
 let $\word = (\overline{a}, \overline{\tau})\in\TW(\Actions)$, and
 let $N=\min\{|\word|\mid \word\in\bigcup_{\pval\in\PVal}\Lg_{-\$}(\valuate{\A}{\pval})\}$.
 For any index $i\in\{1,2,\dots,|\word|\}$ of $\word$ and for any $m\in\{1,2,\dots,N\}$,
 if $a_{i+N} \ne a'_{N-m+1}$ holds for any $(\overline{a'},\overline{\tau'})\in\bigcup_{\pval\in\PVal}\Lg(\valuate{\A}{\pval})$, we have $((\tau_{i+m-1},\tau_{i+m}]\times\Rp\times\PVal)\cap\mathcal{M}(\word,\A)=\emptyset$.
\end{lemma}

\begin{proof}
 If $a_{i+N} \ne a'_{N-m+1}$ holds for any $(\overline{a'},\overline{\tau'})\in\bigcup_{\pval\in\PVal}\Lg(\valuate{\A}{\pval})$, we have the following.
 \begin{displaymath}
  \untimed{\{\word(i+m,|\word|)\}} \not\subseteq \bigcup_{\pval\in\PVal} \untimed{\Lg_{-\$}(\valuate{\A}{\pval})} \Actions^*
 \end{displaymath}
 \Cref{lemma:QS} is proved by the following.

 \begin{align*}
 & \untimed{\{\word(i+m,|\word|)\}} \not\subseteq \bigcup_{\pval\in\PVal} \untimed{\Lg_{-\$}(\valuate{\A}{\pval})} \Actions^*\\
  \Rightarrow& \forall t\in(\tau_{i+m-1},\tau_{i+m}], \pval\in\PVal.\,
  \word(i+m,|\word|)-t \not\in \bigcup_{\pval\in\PVal} \Lg_{-\$}(\valuate{\A}{\pval}) \cdot \TW(\Actions)\\
  \Rightarrow&\forall t\in(\tau_{i+m-1},\tau_{i+m}], t'>t,\pval\in\PVal.\,
  \word|_{(t,t')} \not\in \bigcup_{\pval\in\PVal} \Lg(\valuate{\A}{\pval})\\
  \iff&((\tau_{i+m-1},\tau_{i+m}]\times\Rp\times\PVal)\cap\mathcal{M}(\word,\A)=\emptyset
 \end{align*}
\end{proof}

Then, the proof of \cref{theorem:QS} is as follows.

\begin{proof}
 Since $\Actions^N \action \Actions^*\cap\Actions^n \untimed{\Lg_{-\$}(\valuate{\A}{\pval})} \ne \emptyset$ holds for any $n\geq N+1$, we have $\QSSkipFunc(a_{i+n}) -1 \leq N$.
 By definition of $\QSSkipFunc$, $m<\QSSkipFunc(a_{i+N})$ implies the following.
 \begin{displaymath}
  \forall\pval\in\PVal.\, \Actions^N a_{i+N} \Actions^*\cap \Actions^m \untimed{\Lg_{-\$}(\valuate{\A}{\pval})} = \emptyset
 \end{displaymath}
 Therefore, \cref{lemma:QS} implies \cref{theorem:QS}.
\end{proof}

\subsubsection{Skipping enhanced algorithm}
\label{subsec:skipping_algorith}

\cref{alg:online_with_skip} shows an improvement of \cref{alg:online_no_skip} enhanced by skipping.
The loop in \crefrange{alg_line:no_skip:big_for_beginning}{alg_line:no_skip:big_for_end} of \cref{alg:online_no_skip} is used
 in the matching trial, \ie{} \crefrange{line:leftToRightMatchingInTimedFJS}{alg_line:online_with_skip:end_trial} of \cref{alg:online_with_skip}.
Before each matching trial, we check if the $i+N$-th action $a_{i+N}$ of the timed word $\word=(\overline{a},\overline{\tau})$ can be the $N$-th action of a matching.
In our implementation, for the efficiency, we compute the set of the actions that can be the $N$-th action of a matching beforehand.
If there is no potential matching from the $i$-th element of $\word$, we increase the counter~$i$ using the Quick Search-style skip value function $\QSSkipFunc$.
After reading the $i$-th element $(a_i, \tau_i)$ of $\word$, \cref{alg:online_with_skip} does not immediately output the matching over the available prefix $(a_1, \tau_1), (a_2, \tau_2),\dots,(a_i, \tau_i)$ of $\word$, but it still outputs the matching before obtaining the entire timed word with some delay.
At \cref{alg_line:online_with_skip:kmp_style_skipping}, we increase the counter~$i$ using the parametric KMP-style skip value $\KMPSkipFunc(\loc, \pvals)$.
We can employ the non-parametric KMP-style skip value by replacing $\KMPSkipFunc(\loc, \pvals)$ with $\KMPSkipFunc'(\loc)$.

\begin{algorithm}[t]
 \caption{\small Parametric timed pattern matching with parametric skipping}
 \label{alg:online_with_skip}
 \DontPrintSemicolon
 \newcommand{\myCommentFont}[1]{\texttt{\footnotesize{#1}}}
 \SetCommentSty{myCommentFont}
 \KwIn{A timed word $\word$ and a PTA $\A = (\Actions, \Loc, \locinit, \LocFinal, \Clock, \Param, %
 \Edges)$}
  \KwOut{$Z$ is the match set $\mathcal{M} (\word,\A)$}
  $i \gets 1$
  \tcp*{$i$ is the  position in $\word$ of the beginning of the current matching trial}
 $N=\min\{|\word|\mid \word\in\bigcup_{\pval\in\PVal}\Lg_{-\$}(\valuate{\A}{\pval})\}$ \;
  \While{$i \leq |\word|-N+1$}{
  \While{$\forall\, \pval\in\PVal,(\overline{a'},\overline{\tau'}) \in \Lg(\valuate{\A}{\pval}).\, a_{i+N-1}\neq a'_{N}$}{
 \tcp*{Try matching the $N$-th action of $\Lg(\valuate{\A}{\pval})$}
 $i \gets i + \QSSkipFunc(a_{i + N})$
  \tcp*{Quick Search-style skipping}
   \If{$i > |\word| - N + 1$}{
   \Return
  }
 }
 $Z\gets Z\cup \{(t,t',\pval) \in (\tau_{i-1},\tau_{i}] \times [\tau_{i-1},\infty) \times\PVal\mid \word|_{(t,t')} \in \Lg(\valuate{\A}{\pval})\}$
  \label{line:leftToRightMatchingInTimedFJS}
 \tcp*{Try matching}
 $j \gets \max\{j \in \{i,i+1,\dots,|\word|\} \mid \exists \loc \in \Loc,\pval\in\PVal,t\in\Rp.\, \word(i,j)-t \in \Lg(\valuate{\A_\loc}{\pval})\}$\;
 $C \gets \{(\loc, \pvals) \in \Loc\times\powerset{\PVal} \mid \forall\pval\in\pvals,\exists t\in\Rp.\, \word(i,j)-t \in \Lg(\valuate{\A_\loc}{\pval})\}$\;
  \label{alg_line:online_with_skip:end_trial}
 $i \gets i + \max_{(\loc, \pvals) \in C}{\KMPSkipFunc(\loc, \pvals)}$
  \label{alg_line:online_with_skip:kmp_style_skipping}
  \tcp*{Parametric KMP-style skipping}
 }
 $Z\gets Z\cup \{(t,t',\pval) \in (\tau_{|\word|},\infty) \times [\tau_{|\word|},\infty) \times\PVal\mid \word|_{(t,t')} \in \Lg(\valuate{\A}{\pval})\}$\;
\end{algorithm}
\subsection{Experiments}\label{ss:adhoc-experiments}

We implemented our dedicated algorithms for parametric timed pattern matching in a tool \pmonaa.
We implemented the following three algorithms: the online algorithm without skipping (\cref{alg:online_no_skip}, referred as ``no skip''); the online algorithm with parametric skipping (\cref{alg:online_with_skip}, referred as ``parametric skip''); and the online algorithm with non-parametric skipping (\cref{alg:online_with_skip} where $\KMPSkipFunc(\loc, \pvals)$ at~\cref{alg_line:online_with_skip:kmp_style_skipping} is replaced with $\KMPSkipFunc'(\loc)$, referred as ``non-parametric skip''). 
We note that in \pmonaa{}, the definition of timed word segments is slightly different from \cref{subsection:timed_word}. Namely, we use the timing constraints $\tau_{i-1} \leq t < \tau_{i}$ and $\tau_{j} < t \leq \tau_{j+1}$ instead of $\tau_{i-1} < t \leq \tau_{i}$ and $\tau_{j} \leq t < \tau_{j+1}$. This difference is minor and does not affect the performance evaluation.

In the skip value computation, we use reachability synthesis for PTAs.
Since reachability synthesis is intractable in general (the emptiness problem, \ie{} the (non-)existence of a valuation reaching a given location, is undecidable~\cite{AHV93}), we use the following overapproximation: after investigating 100~configurations, we speculate that all the inconclusive parameter valuations are reachable parameter valuations.
We remark that this overapproximation does not affect the correctness of parametric timed pattern matching, as it potentially \emph{decreases} the skip value.
We conducted experiments to answer the following research questions.

\begin{description}
 \item[RQ1] Which is the fastest algorithm for parametric timed pattern matching?
 \item[RQ2] Why is parametric timed pattern matching slower than non-parametric timed pattern matching? Namely, is it purely because of the difficulty of the problem itself or is it mainly because of the general implementation and data structure required by the general problem setting?
 \item[RQ3] How large is the overhead of the skip value computation? Namely, is it small and acceptable?
\end{description}

We implemented \pmonaa in C++ and we compiled them using GCC 7.3.0.
We conducted the experiments on an Amazon EC2 c4.large instance (2.9\,GHz Intel Xeon E5-2666 v3, 2 vCPUs, and 3.75\,GiB RAM) that runs Ubuntu 18.04 LTS (64\,bit).\footnote{Experiment data can be found on
\href{https://github.com/MasWag/ParamMONAA/blob/master/doc/NFM2019.md}{\nolinkurl{github.com/MasWag/ParamMONAA/blob/master/doc/NFM2019.md}}.}

\begin{figure}[t]
	\centering
	\footnotesize

\scalebox{0.87}{	\begin{tikzpicture}[shorten >=1pt,node distance=2.5cm,on grid,auto] 
		\node[location,initial] (s_0)  {$\loc_1$}; 
		\node[location] (s_1) [right=of s_0] {$\loc_2$}; 
		\node[location] (s_2) [right=of s_1] {$\loc_3$};
		\node[location] (s_3) [right=of s_2] {$\loc_4$};
		\node[location,accepting] (s_4) [right=of s_3] {$\loc_5$};
		\path[->] 
			(s_0) edge [above] node[align=center] {$\styleact{a}$\\$\styleclock{x} := 0$} (s_1)
			(s_1) edge node[above,align=center] {$\styleact{a}$\\ $\styleclock{x} > 1$} node[below] {$\styleclock{x} := 0$} (s_2)
			(s_2) edge node[above,align=center] {$\styleact{a}$\\$\styleclock{x} < \styleparam{\param}$} (s_3)
			(s_3) edge node[above,align=center] {$\styleact{\$}$} (s_4)
		;
	\end{tikzpicture}}
  \caption{\textsc{OnlyTiming}: the parameter $\styleparam{\param}$ is substituted to $1$ in \textsc{OnlyTiming-np}.}
	\label{figure:patterns:onlytiming}
\end{figure}

\Cref{figure:patterns} shows the pattern PTAs we used in the experiments.
We reuse the benchmarks \textsc{Gear}, \textsc{Accel}, and \textsc{Blowup} from \cref{section:PTMC} as well as the new benchmark \textsc{OnlyTiming}.
The timed words for \textsc{Gear} and \textsc{Accel} are generated by the automatic transmission system model in~\cite{HAF14}.
\textsc{Blowup} and \textsc{OnlyTiming} are toy examples.
\textsc{Blowup} shows the worst case situation for parametric timed pattern matching (see \cref{ss:XP:blowup}).
In \textsc{OnlyTiming}, the parametric skip values are greater than the non-parametric skip values.
In \cref{subsec:parametric_vs_non_parametric,subsec:overhead_skip_values}, we also used the non-parametric variants \textsc{Gear-np}, \textsc{Accel-np}, \textsc{Blowup-np}, and \textsc{OnlyTiming-np} where the parameters are substituted to specific concrete values.

\subsubsection{RQ1: Overall execution time}

\begin{table}[tbp]
{
	\centering
 \begin{minipage}[t]{0.5\textwidth}
  \caption{Execution time for \textsc{Gear} [s]}
 \label{table:gear-result}
  \scalebox{0.87}{.\begin{tabular}{|r|r|r|r|r|}
\hline
\cellHeaderColor{}&\cellHeaderColor{}&\cellHeader{Non-Param.}&\cellHeader{Param.}&\cellHeaderColor{}\\
\multirow{-2}{*}{\cellHeader{$|\word|$}}&\multirow{-2}{*}{\cellHeader{No Skip}}&\cellHeader{Skip}&\cellHeader{Skip}&\multirow{-2}{*}{\cellHeader{\imitator}}\\\hline
1467&0.04&0.05&0.05&1.781
\\\hline
2837&0.0725&0.0805&0.09&3.319
\\\hline
4595&0.124&0.13&0.1405&5.512
\\\hline
5839&0.1585&0.156&0.17&7.132
\\\hline
7301&0.201&0.193&0.2115&8.909
\\\hline
8995&0.241&0.2315&0.2505&10.768
\\\hline
10315&0.2815&0.269&0.2875&12.778
\\\hline
11831&0.322&0.301&0.325&14.724
\\\hline
13183&0.3505&0.3245&0.353&16.453
\\\hline
14657&0.392&0.361&0.395&18.319\\
\hline
\end{tabular}
}
 \end{minipage}
 \AuthorVersion{%

 }%
 \begin{minipage}[t]{0.5\textwidth}
  \caption{Execution time for \textsc{Accel} [s]}
 \label{table:accel-result}
 \scalebox{0.87}{\begin{tabular}{|r|r|r|r|r|}
\hline
\cellHeaderColor{}&\cellHeaderColor{}&\cellHeader{Non-Param.}&\cellHeader{Param.}&\cellHeaderColor{}\\
\multirow{-2}{*}{\cellHeader{$|\word|$}}&\multirow{-2}{*}{\cellHeader{No Skip}}&\cellHeader{Skip}&\cellHeader{Skip}&\multirow{-2}{*}{\cellHeader{\imitator}}\\\hline
2559&0.03&0.0515&0.06&2.332
\\\hline
4894&0.0605&0.0605&0.0705&4.663
\\\hline
7799&0.1005&0.071&0.08&7.532
\\\hline
10045&0.13&0.08&0.09&9.731
\\\hline
12531&0.161&0.09&0.1&12.503
\\\hline
15375&0.1985&0.1005&0.113&15.583
\\\hline
17688&0.2265&0.1095&0.1215&17.754
\\\hline
20299&0.261&0.115&0.1325&21.040
\\\hline
22691&0.288&0.121&0.143&23.044
\\\hline
25137&0.3205&0.1315&0.159&25.815
\\\hline
\end{tabular}
}
 \end{minipage}
 \AuthorVersion{%

 }%
 \begin{minipage}[t]{0.5\textwidth}
  \caption{\small Execution time for \textsc{Blowup} [s]}
 \label{table:blowup-result}
 \scalebox{0.80}{\begin{tabular}{|r|r|r|r|c|}
\hline
\cellHeaderColor{}&\cellHeaderColor{}&\cellHeader{Non-Param.}&\cellHeader{Param.}&\cellHeaderColor{}\\
\multirow{-2}{*}{\cellHeader{$|\word|$}}&\multirow{-2}{*}{\cellHeader{No Skip}}&\cellHeader{Skip}&\cellHeader{Skip}&\multirow{-2}{*}{\cellHeader{\imitator}}\\\hline
2000&66.75&68.0125&67.9735&OutOfMemory
\\\hline
4000&267.795&271.642&269.084&OutOfMemory
\\\hline
6000&601.335&611.782&607.58&OutOfMemory
\\\hline
8000&1081.42&1081.25&1079&OutOfMemory
\\\hline
10000&1678.15&1688.22&1694.53&OutOfMemory
\\\hline
\end{tabular}
}
 \end{minipage}
 \AuthorVersion{%

 }%
 \begin{minipage}[t]{0.53\textwidth}
  \caption{\small Execution time for \textsc{OnlyTiming} [s]}
 \label{table:only-timing-result}
 \scalebox{0.87}{\begin{tabular}{|r|r|r|r|r|}
 \hline
\cellHeaderColor{}&\cellHeaderColor{}&\cellHeader{Non-Param.}&\cellHeader{Param.}&\cellHeaderColor{}\\
\multirow{-2}{*}{\cellHeader{$|\word|$}}&\multirow{-2}{*}{\cellHeader{No Skip}}&\cellHeader{Skip}&\cellHeader{Skip}&\multirow{-2}{*}{\cellHeader{\imitator}}\\\hline
1000&0.0995&0.1305&0.11&1.690
\\\hline
2000&0.191&0.23&0.191&3.518
\\\hline
3000&0.2905&0.3265&0.273&5.499
\\\hline
4000&0.3905&0.426&0.3525&7.396
\\\hline
5000&0.488&0.5225&0.4325&9.123
\\\hline
6000&0.588&0.6235&0.517&11.005
\\\hline
\end{tabular}
}
 \end{minipage}

 }
\end{table}

\begin{figure}[tbp]
 \scalebox{0.5}{\begin{tikzpicture}[gnuplot]
\tikzset{every node/.append style={font={\fontsize{18.0pt}{21.6pt}\selectfont}}}
\path (0.000,0.000) rectangle (12.500,8.750);
\gpcolor{color=gp lt color border}
\gpsetlinetype{gp lt border}
\gpsetdashtype{gp dt solid}
\gpsetlinewidth{1.00}
\draw[gp path] (2.705,1.772)--(2.885,1.772);
\draw[gp path] (11.506,1.772)--(11.326,1.772);
\node[gp node right] at (2.374,1.772) {$0$};
\draw[gp path] (2.705,2.575)--(2.885,2.575);
\draw[gp path] (11.506,2.575)--(11.326,2.575);
\node[gp node right] at (2.374,2.575) {$0.05$};
\draw[gp path] (2.705,3.378)--(2.885,3.378);
\draw[gp path] (11.506,3.378)--(11.326,3.378);
\node[gp node right] at (2.374,3.378) {$0.1$};
\draw[gp path] (2.705,4.181)--(2.885,4.181);
\draw[gp path] (11.506,4.181)--(11.326,4.181);
\node[gp node right] at (2.374,4.181) {$0.15$};
\draw[gp path] (2.705,4.984)--(2.885,4.984);
\draw[gp path] (11.506,4.984)--(11.326,4.984);
\node[gp node right] at (2.374,4.984) {$0.2$};
\draw[gp path] (2.705,5.786)--(2.885,5.786);
\draw[gp path] (11.506,5.786)--(11.326,5.786);
\node[gp node right] at (2.374,5.786) {$0.25$};
\draw[gp path] (2.705,6.589)--(2.885,6.589);
\draw[gp path] (11.506,6.589)--(11.326,6.589);
\node[gp node right] at (2.374,6.589) {$0.3$};
\draw[gp path] (2.705,7.392)--(2.885,7.392);
\draw[gp path] (11.506,7.392)--(11.326,7.392);
\node[gp node right] at (2.374,7.392) {$0.35$};
\draw[gp path] (2.705,8.195)--(2.885,8.195);
\draw[gp path] (11.506,8.195)--(11.326,8.195);
\node[gp node right] at (2.374,8.195) {$0.4$};
\draw[gp path] (2.705,1.772)--(2.705,1.952);
\draw[gp path] (2.705,8.195)--(2.705,8.015);
\node[gp node center] at (2.705,1.218) {$0$};
\draw[gp path] (3.805,1.772)--(3.805,1.952);
\draw[gp path] (3.805,8.195)--(3.805,8.015);
\node[gp node center] at (3.805,1.218) {$20$};
\draw[gp path] (4.905,1.772)--(4.905,1.952);
\draw[gp path] (4.905,8.195)--(4.905,8.015);
\node[gp node center] at (4.905,1.218) {$40$};
\draw[gp path] (6.005,1.772)--(6.005,1.952);
\draw[gp path] (6.005,8.195)--(6.005,8.015);
\node[gp node center] at (6.005,1.218) {$60$};
\draw[gp path] (7.106,1.772)--(7.106,1.952);
\draw[gp path] (7.106,8.195)--(7.106,8.015);
\node[gp node center] at (7.106,1.218) {$80$};
\draw[gp path] (8.206,1.772)--(8.206,1.952);
\draw[gp path] (8.206,8.195)--(8.206,8.015);
\node[gp node center] at (8.206,1.218) {$100$};
\draw[gp path] (9.306,1.772)--(9.306,1.952);
\draw[gp path] (9.306,8.195)--(9.306,8.015);
\node[gp node center] at (9.306,1.218) {$120$};
\draw[gp path] (10.406,1.772)--(10.406,1.952);
\draw[gp path] (10.406,8.195)--(10.406,8.015);
\node[gp node center] at (10.406,1.218) {$140$};
\draw[gp path] (11.506,1.772)--(11.506,1.952);
\draw[gp path] (11.506,8.195)--(11.506,8.015);
\node[gp node center] at (11.506,1.218) {$160$};
\draw[gp path] (2.705,8.195)--(2.705,1.772)--(11.506,1.772)--(11.506,8.195)--cycle;
\node[gp node center,rotate=-270] at (0.496,4.983) {Execution Time [s]};
\node[gp node center] at (7.105,0.387) {Number of Events [$\times 100$]};
\node[gp node right] at (9.009,7.738) {No Skip};
\gpcolor{rgb color={0.580,0.000,0.827}}
\gpsetlinewidth{3.00}
\draw[gp path] (9.340,7.738)--(10.844,7.738);
\draw[gp path] (3.512,2.414)--(4.266,2.936)--(5.233,3.763)--(5.917,4.317)--(6.721,5.000)%
  --(7.653,5.642)--(8.379,6.292)--(9.213,6.943)--(9.956,7.400)--(10.767,8.067);
\gpsetpointsize{6.00}
\gppoint{gp mark 7}{(3.512,2.414)}
\gppoint{gp mark 7}{(4.266,2.936)}
\gppoint{gp mark 7}{(5.233,3.763)}
\gppoint{gp mark 7}{(5.917,4.317)}
\gppoint{gp mark 7}{(6.721,5.000)}
\gppoint{gp mark 7}{(7.653,5.642)}
\gppoint{gp mark 7}{(8.379,6.292)}
\gppoint{gp mark 7}{(9.213,6.943)}
\gppoint{gp mark 7}{(9.956,7.400)}
\gppoint{gp mark 7}{(10.767,8.067)}
\gppoint{gp mark 7}{(10.092,7.738)}
\gpcolor{color=gp lt color border}
\node[gp node right] at (9.009,7.184) {Parametric Skip};
\gpcolor{rgb color={0.000,0.620,0.451}}
\draw[gp path] (9.340,7.184)--(10.844,7.184);
\draw[gp path] (3.512,2.575)--(4.266,3.217)--(5.233,4.028)--(5.917,4.502)--(6.721,5.168)%
  --(7.653,5.794)--(8.379,6.389)--(9.213,6.991)--(9.956,7.440)--(10.767,8.115);
\gppoint{gp mark 8}{(3.512,2.575)}
\gppoint{gp mark 8}{(4.266,3.217)}
\gppoint{gp mark 8}{(5.233,4.028)}
\gppoint{gp mark 8}{(5.917,4.502)}
\gppoint{gp mark 8}{(6.721,5.168)}
\gppoint{gp mark 8}{(7.653,5.794)}
\gppoint{gp mark 8}{(8.379,6.389)}
\gppoint{gp mark 8}{(9.213,6.991)}
\gppoint{gp mark 8}{(9.956,7.440)}
\gppoint{gp mark 8}{(10.767,8.115)}
\gppoint{gp mark 8}{(10.092,7.184)}
\gpcolor{color=gp lt color border}
\node[gp node right] at (9.009,6.630) {Non-Parametric Skip};
\gpcolor{rgb color={0.337,0.706,0.914}}
\draw[gp path] (9.340,6.630)--(10.844,6.630);
\draw[gp path] (3.512,2.575)--(4.266,3.065)--(5.233,3.859)--(5.917,4.277)--(6.721,4.871)%
  --(7.653,5.489)--(8.379,6.091)--(9.213,6.605)--(9.956,6.983)--(10.767,7.569);
\gppoint{gp mark 4}{(3.512,2.575)}
\gppoint{gp mark 4}{(4.266,3.065)}
\gppoint{gp mark 4}{(5.233,3.859)}
\gppoint{gp mark 4}{(5.917,4.277)}
\gppoint{gp mark 4}{(6.721,4.871)}
\gppoint{gp mark 4}{(7.653,5.489)}
\gppoint{gp mark 4}{(8.379,6.091)}
\gppoint{gp mark 4}{(9.213,6.605)}
\gppoint{gp mark 4}{(9.956,6.983)}
\gppoint{gp mark 4}{(10.767,7.569)}
\gppoint{gp mark 4}{(10.092,6.630)}
\gpcolor{color=gp lt color border}
\gpsetlinewidth{1.00}
\draw[gp path] (2.705,8.195)--(2.705,1.772)--(11.506,1.772)--(11.506,8.195)--cycle;
\gpdefrectangularnode{gp plot 1}{\pgfpoint{2.705cm}{1.772cm}}{\pgfpoint{11.506cm}{8.195cm}}
\end{tikzpicture}}
 \scalebox{0.5}{\begin{tikzpicture}[gnuplot]
\tikzset{every node/.append style={font={\fontsize{18.0pt}{21.6pt}\selectfont}}}
\path (0.000,0.000) rectangle (12.500,8.750);
\gpcolor{color=gp lt color border}
\gpsetlinetype{gp lt border}
\gpsetdashtype{gp dt solid}
\gpsetlinewidth{1.00}
\draw[gp path] (2.705,1.772)--(2.885,1.772);
\draw[gp path] (11.506,1.772)--(11.326,1.772);
\node[gp node right] at (2.374,1.772) {$0$};
\draw[gp path] (2.705,2.690)--(2.885,2.690);
\draw[gp path] (11.506,2.690)--(11.326,2.690);
\node[gp node right] at (2.374,2.690) {$0.05$};
\draw[gp path] (2.705,3.607)--(2.885,3.607);
\draw[gp path] (11.506,3.607)--(11.326,3.607);
\node[gp node right] at (2.374,3.607) {$0.1$};
\draw[gp path] (2.705,4.525)--(2.885,4.525);
\draw[gp path] (11.506,4.525)--(11.326,4.525);
\node[gp node right] at (2.374,4.525) {$0.15$};
\draw[gp path] (2.705,5.442)--(2.885,5.442);
\draw[gp path] (11.506,5.442)--(11.326,5.442);
\node[gp node right] at (2.374,5.442) {$0.2$};
\draw[gp path] (2.705,6.360)--(2.885,6.360);
\draw[gp path] (11.506,6.360)--(11.326,6.360);
\node[gp node right] at (2.374,6.360) {$0.25$};
\draw[gp path] (2.705,7.277)--(2.885,7.277);
\draw[gp path] (11.506,7.277)--(11.326,7.277);
\node[gp node right] at (2.374,7.277) {$0.3$};
\draw[gp path] (2.705,8.195)--(2.885,8.195);
\draw[gp path] (11.506,8.195)--(11.326,8.195);
\node[gp node right] at (2.374,8.195) {$0.35$};
\draw[gp path] (2.705,1.772)--(2.705,1.952);
\draw[gp path] (2.705,8.195)--(2.705,8.015);
\node[gp node center] at (2.705,1.218) {$0$};
\draw[gp path] (4.172,1.772)--(4.172,1.952);
\draw[gp path] (4.172,8.195)--(4.172,8.015);
\node[gp node center] at (4.172,1.218) {$50$};
\draw[gp path] (5.639,1.772)--(5.639,1.952);
\draw[gp path] (5.639,8.195)--(5.639,8.015);
\node[gp node center] at (5.639,1.218) {$100$};
\draw[gp path] (7.106,1.772)--(7.106,1.952);
\draw[gp path] (7.106,8.195)--(7.106,8.015);
\node[gp node center] at (7.106,1.218) {$150$};
\draw[gp path] (8.572,1.772)--(8.572,1.952);
\draw[gp path] (8.572,8.195)--(8.572,8.015);
\node[gp node center] at (8.572,1.218) {$200$};
\draw[gp path] (10.039,1.772)--(10.039,1.952);
\draw[gp path] (10.039,8.195)--(10.039,8.015);
\node[gp node center] at (10.039,1.218) {$250$};
\draw[gp path] (11.506,1.772)--(11.506,1.952);
\draw[gp path] (11.506,8.195)--(11.506,8.015);
\node[gp node center] at (11.506,1.218) {$300$};
\draw[gp path] (2.705,8.195)--(2.705,1.772)--(11.506,1.772)--(11.506,8.195)--cycle;
\node[gp node center,rotate=-270] at (0.496,4.983) {Execution Time [s]};
\node[gp node center] at (7.105,0.387) {Number of Events [$\times 100$]};
\node[gp node right] at (9.009,7.738) {No Skip};
\gpcolor{rgb color={0.580,0.000,0.827}}
\gpsetlinewidth{3.00}
\draw[gp path] (9.340,7.738)--(10.844,7.738);
\draw[gp path] (3.456,2.323)--(4.141,2.882)--(4.993,3.616)--(5.652,4.158)--(6.381,4.727)%
  --(7.216,5.415)--(7.894,5.929)--(8.660,6.562)--(9.362,7.057)--(10.079,7.654);
\gpsetpointsize{6.00}
\gppoint{gp mark 7}{(3.456,2.323)}
\gppoint{gp mark 7}{(4.141,2.882)}
\gppoint{gp mark 7}{(4.993,3.616)}
\gppoint{gp mark 7}{(5.652,4.158)}
\gppoint{gp mark 7}{(6.381,4.727)}
\gppoint{gp mark 7}{(7.216,5.415)}
\gppoint{gp mark 7}{(7.894,5.929)}
\gppoint{gp mark 7}{(8.660,6.562)}
\gppoint{gp mark 7}{(9.362,7.057)}
\gppoint{gp mark 7}{(10.079,7.654)}
\gppoint{gp mark 7}{(10.092,7.738)}
\gpcolor{color=gp lt color border}
\node[gp node right] at (9.009,7.184) {Parametric Skip};
\gpcolor{rgb color={0.000,0.620,0.451}}
\draw[gp path] (9.340,7.184)--(10.844,7.184);
\draw[gp path] (3.456,2.873)--(4.141,3.066)--(4.993,3.240)--(5.652,3.424)--(6.381,3.607)%
  --(7.216,3.846)--(7.894,4.002)--(8.660,4.204)--(9.362,4.396)--(10.079,4.690);
\gppoint{gp mark 8}{(3.456,2.873)}
\gppoint{gp mark 8}{(4.141,3.066)}
\gppoint{gp mark 8}{(4.993,3.240)}
\gppoint{gp mark 8}{(5.652,3.424)}
\gppoint{gp mark 8}{(6.381,3.607)}
\gppoint{gp mark 8}{(7.216,3.846)}
\gppoint{gp mark 8}{(7.894,4.002)}
\gppoint{gp mark 8}{(8.660,4.204)}
\gppoint{gp mark 8}{(9.362,4.396)}
\gppoint{gp mark 8}{(10.079,4.690)}
\gppoint{gp mark 8}{(10.092,7.184)}
\gpcolor{color=gp lt color border}
\node[gp node right] at (9.009,6.630) {Non-Parametric Skip};
\gpcolor{rgb color={0.337,0.706,0.914}}
\draw[gp path] (9.340,6.630)--(10.844,6.630);
\draw[gp path] (3.456,2.717)--(4.141,2.882)--(4.993,3.075)--(5.652,3.240)--(6.381,3.424)%
  --(7.216,3.616)--(7.894,3.781)--(8.660,3.882)--(9.362,3.993)--(10.079,4.185);
\gppoint{gp mark 4}{(3.456,2.717)}
\gppoint{gp mark 4}{(4.141,2.882)}
\gppoint{gp mark 4}{(4.993,3.075)}
\gppoint{gp mark 4}{(5.652,3.240)}
\gppoint{gp mark 4}{(6.381,3.424)}
\gppoint{gp mark 4}{(7.216,3.616)}
\gppoint{gp mark 4}{(7.894,3.781)}
\gppoint{gp mark 4}{(8.660,3.882)}
\gppoint{gp mark 4}{(9.362,3.993)}
\gppoint{gp mark 4}{(10.079,4.185)}
\gppoint{gp mark 4}{(10.092,6.630)}
\gpcolor{color=gp lt color border}
\gpsetlinewidth{1.00}
\draw[gp path] (2.705,8.195)--(2.705,1.772)--(11.506,1.772)--(11.506,8.195)--cycle;
\gpdefrectangularnode{gp plot 1}{\pgfpoint{2.705cm}{1.772cm}}{\pgfpoint{11.506cm}{8.195cm}}
\end{tikzpicture}}
 \scalebox{0.5}{\begin{tikzpicture}[gnuplot]
\tikzset{every node/.append style={font={\fontsize{18.0pt}{21.6pt}\selectfont}}}
\path (0.000,0.000) rectangle (12.500,8.750);
\gpcolor{color=gp lt color border}
\gpsetlinetype{gp lt border}
\gpsetdashtype{gp dt solid}
\gpsetlinewidth{1.00}
\draw[gp path] (2.705,1.772)--(2.885,1.772);
\draw[gp path] (11.506,1.772)--(11.326,1.772);
\node[gp node right] at (2.374,1.772) {$0$};
\draw[gp path] (2.705,2.486)--(2.885,2.486);
\draw[gp path] (11.506,2.486)--(11.326,2.486);
\node[gp node right] at (2.374,2.486) {$200$};
\draw[gp path] (2.705,3.199)--(2.885,3.199);
\draw[gp path] (11.506,3.199)--(11.326,3.199);
\node[gp node right] at (2.374,3.199) {$400$};
\draw[gp path] (2.705,3.913)--(2.885,3.913);
\draw[gp path] (11.506,3.913)--(11.326,3.913);
\node[gp node right] at (2.374,3.913) {$600$};
\draw[gp path] (2.705,4.627)--(2.885,4.627);
\draw[gp path] (11.506,4.627)--(11.326,4.627);
\node[gp node right] at (2.374,4.627) {$800$};
\draw[gp path] (2.705,5.340)--(2.885,5.340);
\draw[gp path] (11.506,5.340)--(11.326,5.340);
\node[gp node right] at (2.374,5.340) {$1000$};
\draw[gp path] (2.705,6.054)--(2.885,6.054);
\draw[gp path] (11.506,6.054)--(11.326,6.054);
\node[gp node right] at (2.374,6.054) {$1200$};
\draw[gp path] (2.705,6.768)--(2.885,6.768);
\draw[gp path] (11.506,6.768)--(11.326,6.768);
\node[gp node right] at (2.374,6.768) {$1400$};
\draw[gp path] (2.705,7.481)--(2.885,7.481);
\draw[gp path] (11.506,7.481)--(11.326,7.481);
\node[gp node right] at (2.374,7.481) {$1600$};
\draw[gp path] (2.705,8.195)--(2.885,8.195);
\draw[gp path] (11.506,8.195)--(11.326,8.195);
\node[gp node right] at (2.374,8.195) {$1800$};
\draw[gp path] (2.705,1.772)--(2.705,1.952);
\draw[gp path] (2.705,8.195)--(2.705,8.015);
\node[gp node center] at (2.705,1.218) {$20$};
\draw[gp path] (3.805,1.772)--(3.805,1.952);
\draw[gp path] (3.805,8.195)--(3.805,8.015);
\node[gp node center] at (3.805,1.218) {$30$};
\draw[gp path] (4.905,1.772)--(4.905,1.952);
\draw[gp path] (4.905,8.195)--(4.905,8.015);
\node[gp node center] at (4.905,1.218) {$40$};
\draw[gp path] (6.005,1.772)--(6.005,1.952);
\draw[gp path] (6.005,8.195)--(6.005,8.015);
\node[gp node center] at (6.005,1.218) {$50$};
\draw[gp path] (7.106,1.772)--(7.106,1.952);
\draw[gp path] (7.106,8.195)--(7.106,8.015);
\node[gp node center] at (7.106,1.218) {$60$};
\draw[gp path] (8.206,1.772)--(8.206,1.952);
\draw[gp path] (8.206,8.195)--(8.206,8.015);
\node[gp node center] at (8.206,1.218) {$70$};
\draw[gp path] (9.306,1.772)--(9.306,1.952);
\draw[gp path] (9.306,8.195)--(9.306,8.015);
\node[gp node center] at (9.306,1.218) {$80$};
\draw[gp path] (10.406,1.772)--(10.406,1.952);
\draw[gp path] (10.406,8.195)--(10.406,8.015);
\node[gp node center] at (10.406,1.218) {$90$};
\draw[gp path] (11.506,1.772)--(11.506,1.952);
\draw[gp path] (11.506,8.195)--(11.506,8.015);
\node[gp node center] at (11.506,1.218) {$100$};
\draw[gp path] (2.705,8.195)--(2.705,1.772)--(11.506,1.772)--(11.506,8.195)--cycle;
\node[gp node center,rotate=-270] at (0.496,4.983) {Execution Time [s]};
\node[gp node center] at (7.105,0.387) {Number of Events [$\times 100$]};
\node[gp node right] at (9.009,7.738) {No Skip};
\gpcolor{rgb color={0.580,0.000,0.827}}
\gpsetlinewidth{3.00}
\draw[gp path] (9.340,7.738)--(10.844,7.738);
\draw[gp path] (2.705,2.010)--(4.905,2.728)--(7.106,3.918)--(9.306,5.631)--(11.506,7.760);
\gpsetpointsize{6.00}
\gppoint{gp mark 7}{(2.705,2.010)}
\gppoint{gp mark 7}{(4.905,2.728)}
\gppoint{gp mark 7}{(7.106,3.918)}
\gppoint{gp mark 7}{(9.306,5.631)}
\gppoint{gp mark 7}{(11.506,7.760)}
\gppoint{gp mark 7}{(10.092,7.738)}
\gpcolor{color=gp lt color border}
\node[gp node right] at (9.009,7.184) {Parametric Skip};
\gpcolor{rgb color={0.000,0.620,0.451}}
\draw[gp path] (9.340,7.184)--(10.844,7.184);
\draw[gp path] (2.705,2.015)--(4.905,2.732)--(7.106,3.940)--(9.306,5.622)--(11.506,7.819);
\gppoint{gp mark 8}{(2.705,2.015)}
\gppoint{gp mark 8}{(4.905,2.732)}
\gppoint{gp mark 8}{(7.106,3.940)}
\gppoint{gp mark 8}{(9.306,5.622)}
\gppoint{gp mark 8}{(11.506,7.819)}
\gppoint{gp mark 8}{(10.092,7.184)}
\gpcolor{color=gp lt color border}
\node[gp node right] at (9.009,6.630) {Non-Parametric Skip};
\gpcolor{rgb color={0.337,0.706,0.914}}
\draw[gp path] (9.340,6.630)--(10.844,6.630);
\draw[gp path] (2.705,2.015)--(4.905,2.741)--(7.106,3.955)--(9.306,5.630)--(11.506,7.796);
\gppoint{gp mark 4}{(2.705,2.015)}
\gppoint{gp mark 4}{(4.905,2.741)}
\gppoint{gp mark 4}{(7.106,3.955)}
\gppoint{gp mark 4}{(9.306,5.630)}
\gppoint{gp mark 4}{(11.506,7.796)}
\gppoint{gp mark 4}{(10.092,6.630)}
\gpcolor{color=gp lt color border}
\gpsetlinewidth{1.00}
\draw[gp path] (2.705,8.195)--(2.705,1.772)--(11.506,1.772)--(11.506,8.195)--cycle;
\gpdefrectangularnode{gp plot 1}{\pgfpoint{2.705cm}{1.772cm}}{\pgfpoint{11.506cm}{8.195cm}}
\end{tikzpicture}}
 \scalebox{0.5}{\begin{tikzpicture}[gnuplot]
\tikzset{every node/.append style={font={\fontsize{18.0pt}{21.6pt}\selectfont}}}
\path (0.000,0.000) rectangle (12.500,8.750);
\gpcolor{color=gp lt color border}
\gpsetlinetype{gp lt border}
\gpsetdashtype{gp dt solid}
\gpsetlinewidth{1.00}
\draw[gp path] (2.374,1.772)--(2.554,1.772);
\draw[gp path] (11.506,1.772)--(11.326,1.772);
\node[gp node right] at (2.043,1.772) {$0$};
\draw[gp path] (2.374,2.690)--(2.554,2.690);
\draw[gp path] (11.506,2.690)--(11.326,2.690);
\node[gp node right] at (2.043,2.690) {$0.1$};
\draw[gp path] (2.374,3.607)--(2.554,3.607);
\draw[gp path] (11.506,3.607)--(11.326,3.607);
\node[gp node right] at (2.043,3.607) {$0.2$};
\draw[gp path] (2.374,4.525)--(2.554,4.525);
\draw[gp path] (11.506,4.525)--(11.326,4.525);
\node[gp node right] at (2.043,4.525) {$0.3$};
\draw[gp path] (2.374,5.442)--(2.554,5.442);
\draw[gp path] (11.506,5.442)--(11.326,5.442);
\node[gp node right] at (2.043,5.442) {$0.4$};
\draw[gp path] (2.374,6.360)--(2.554,6.360);
\draw[gp path] (11.506,6.360)--(11.326,6.360);
\node[gp node right] at (2.043,6.360) {$0.5$};
\draw[gp path] (2.374,7.277)--(2.554,7.277);
\draw[gp path] (11.506,7.277)--(11.326,7.277);
\node[gp node right] at (2.043,7.277) {$0.6$};
\draw[gp path] (2.374,8.195)--(2.554,8.195);
\draw[gp path] (11.506,8.195)--(11.326,8.195);
\node[gp node right] at (2.043,8.195) {$0.7$};
\draw[gp path] (2.374,1.772)--(2.374,1.952);
\draw[gp path] (2.374,8.195)--(2.374,8.015);
\node[gp node center] at (2.374,1.218) {$10$};
\draw[gp path] (4.200,1.772)--(4.200,1.952);
\draw[gp path] (4.200,8.195)--(4.200,8.015);
\node[gp node center] at (4.200,1.218) {$20$};
\draw[gp path] (6.027,1.772)--(6.027,1.952);
\draw[gp path] (6.027,8.195)--(6.027,8.015);
\node[gp node center] at (6.027,1.218) {$30$};
\draw[gp path] (7.853,1.772)--(7.853,1.952);
\draw[gp path] (7.853,8.195)--(7.853,8.015);
\node[gp node center] at (7.853,1.218) {$40$};
\draw[gp path] (9.680,1.772)--(9.680,1.952);
\draw[gp path] (9.680,8.195)--(9.680,8.015);
\node[gp node center] at (9.680,1.218) {$50$};
\draw[gp path] (11.506,1.772)--(11.506,1.952);
\draw[gp path] (11.506,8.195)--(11.506,8.015);
\node[gp node center] at (11.506,1.218) {$60$};
\draw[gp path] (2.374,8.195)--(2.374,1.772)--(11.506,1.772)--(11.506,8.195)--cycle;
\node[gp node center,rotate=-270] at (0.496,4.983) {Execution Time [s]};
\node[gp node center] at (6.940,0.387) {Number of Events [$\times 100$]};
\node[gp node right] at (9.009,7.738) {No Skip};
\gpcolor{rgb color={0.580,0.000,0.827}}
\gpsetlinewidth{3.00}
\draw[gp path] (9.340,7.738)--(10.844,7.738);
\draw[gp path] (2.374,2.685)--(4.200,3.525)--(6.027,4.438)--(7.853,5.355)--(9.680,6.250)%
  --(11.506,7.167);
\gpsetpointsize{6.00}
\gppoint{gp mark 7}{(2.374,2.685)}
\gppoint{gp mark 7}{(4.200,3.525)}
\gppoint{gp mark 7}{(6.027,4.438)}
\gppoint{gp mark 7}{(7.853,5.355)}
\gppoint{gp mark 7}{(9.680,6.250)}
\gppoint{gp mark 7}{(11.506,7.167)}
\gppoint{gp mark 7}{(10.092,7.738)}
\gpcolor{color=gp lt color border}
\node[gp node right] at (9.009,7.184) {Parametric Skip};
\gpcolor{rgb color={0.000,0.620,0.451}}
\draw[gp path] (9.340,7.184)--(10.844,7.184);
\draw[gp path] (2.374,2.781)--(4.200,3.525)--(6.027,4.277)--(7.853,5.006)--(9.680,5.740)%
  --(11.506,6.516);
\gppoint{gp mark 8}{(2.374,2.781)}
\gppoint{gp mark 8}{(4.200,3.525)}
\gppoint{gp mark 8}{(6.027,4.277)}
\gppoint{gp mark 8}{(7.853,5.006)}
\gppoint{gp mark 8}{(9.680,5.740)}
\gppoint{gp mark 8}{(11.506,6.516)}
\gppoint{gp mark 8}{(10.092,7.184)}
\gpcolor{color=gp lt color border}
\node[gp node right] at (9.009,6.630) {Non-Parametric Skip};
\gpcolor{rgb color={0.337,0.706,0.914}}
\draw[gp path] (9.340,6.630)--(10.844,6.630);
\draw[gp path] (2.374,2.969)--(4.200,3.882)--(6.027,4.768)--(7.853,5.681)--(9.680,6.566)%
  --(11.506,7.493);
\gppoint{gp mark 4}{(2.374,2.969)}
\gppoint{gp mark 4}{(4.200,3.882)}
\gppoint{gp mark 4}{(6.027,4.768)}
\gppoint{gp mark 4}{(7.853,5.681)}
\gppoint{gp mark 4}{(9.680,6.566)}
\gppoint{gp mark 4}{(11.506,7.493)}
\gppoint{gp mark 4}{(10.092,6.630)}
\gpcolor{color=gp lt color border}
\gpsetlinewidth{1.00}
\draw[gp path] (2.374,8.195)--(2.374,1.772)--(11.506,1.772)--(11.506,8.195)--cycle;
\gpdefrectangularnode{gp plot 1}{\pgfpoint{2.374cm}{1.772cm}}{\pgfpoint{11.506cm}{8.195cm}}
\end{tikzpicture}}
 \caption{Execution time for the benchmarks with parameters which \monaa cannot handle: \textsc{Gear} (above left), \textsc{Accel} (above right), \textsc{Blowup} (below left), and \textsc{OnlyTiming} (below right)}
 \label{figure:normal-result}
\end{figure}

To answer RQ1, we compared the total execution time of \pmonaa using \textsc{Gear}, \textsc{Accel}, \textsc{Blowup}, and \textsc{OnlyTiming}. As a baseline, we used our previous implementation of parametric timed pattern matching based on \imitator{} (``Butter Jellyfish'', version 2.10.4).
\crefrange{table:gear-result}{table:only-timing-result} and \cref{figure:normal-result} show the execution time of our online algorithms compared with the \imitator{}-based implementation.

In \crefrange{table:gear-result}{table:only-timing-result}, we observe that our algorithms are faster than the \imitator{}-based implementation by orders of magnitude.
Moreover, for \textsc{Blowup}, the \imitator{}-based implementation aborted due to out of memory. 
This is mainly because \pmonaa is specific to parametric timed pattern matching while \imitator{} is a general tool for parametric verification. %
This shows the much better efficiency of our dedicated approach compared to~\cref{section:PTMC}.

In \cref{figure:normal-result}, we observe that 
the curve of ``no skip'' has the steepest slope and
the curves of either ``parametric skip'' or ``non-parametric skip'' have the gentlest slope
except for \textsc{Blowup}.
\textsc{Blowup} is a benchmark designed on purpose to observe exponential blowup of the execution time, and it requires much time for all of the implementations.

For \textsc{Gear} and \textsc{Accel}, the execution time of ``non-parametric skip'' increases the most gently.
This is because the parametric KMP-style skip value $\KMPSkipFunc(\loc, \pvals)$ and the non-parametric KMP-style skip value $\KMPSkipFunc'(\loc)$ are equal for these benchmarks, and ``parametric skip'' is slower due to the inclusion checking $\pvals\subseteq\pvalsi{\loc,n}$.

For \textsc{OnlyTiming}, we observe that the execution time of ``parametric skip'' increases the most gently because the parametric KMP-style skip value $\KMPSkipFunc(\loc, \pvals)$ is larger than the non-parametric KMP-style skip value $\KMPSkipFunc'(\loc)$.

We conclude that skipping usually makes parametric timed pattern matching efficient.
The preference between two skipping methods depends on the pattern PTA and it is a future work to investigate the tendency.
Since the computation of the skip values does not take much time, 
the following work flow is reasonable:
\begin{ienumeration}
	\item compute the skip values for both of them; and 
	\item use ``parametric skip'' only if its skip values are strictly larger than that of ``non-parametric skip''.
\end{ienumeration}

\subsubsection{RQ2: Parametric vs.\ non-parametric timed pattern matching}
\label{subsec:parametric_vs_non_parametric}

\begin{figure}[tbp]
 \scalebox{0.5}{\begin{tikzpicture}[gnuplot]
\tikzset{every node/.append style={font={\fontsize{18.0pt}{21.6pt}\selectfont}}}
\path (0.000,0.000) rectangle (12.500,8.750);
\gpcolor{color=gp lt color border}
\gpsetlinetype{gp lt border}
\gpsetdashtype{gp dt solid}
\gpsetlinewidth{1.00}
\draw[gp path] (2.705,1.772)--(2.885,1.772);
\draw[gp path] (11.506,1.772)--(11.326,1.772);
\node[gp node right] at (2.374,1.772) {$0$};
\draw[gp path] (2.705,2.575)--(2.885,2.575);
\draw[gp path] (11.506,2.575)--(11.326,2.575);
\node[gp node right] at (2.374,2.575) {$0.05$};
\draw[gp path] (2.705,3.378)--(2.885,3.378);
\draw[gp path] (11.506,3.378)--(11.326,3.378);
\node[gp node right] at (2.374,3.378) {$0.1$};
\draw[gp path] (2.705,4.181)--(2.885,4.181);
\draw[gp path] (11.506,4.181)--(11.326,4.181);
\node[gp node right] at (2.374,4.181) {$0.15$};
\draw[gp path] (2.705,4.984)--(2.885,4.984);
\draw[gp path] (11.506,4.984)--(11.326,4.984);
\node[gp node right] at (2.374,4.984) {$0.2$};
\draw[gp path] (2.705,5.786)--(2.885,5.786);
\draw[gp path] (11.506,5.786)--(11.326,5.786);
\node[gp node right] at (2.374,5.786) {$0.25$};
\draw[gp path] (2.705,6.589)--(2.885,6.589);
\draw[gp path] (11.506,6.589)--(11.326,6.589);
\node[gp node right] at (2.374,6.589) {$0.3$};
\draw[gp path] (2.705,7.392)--(2.885,7.392);
\draw[gp path] (11.506,7.392)--(11.326,7.392);
\node[gp node right] at (2.374,7.392) {$0.35$};
\draw[gp path] (2.705,8.195)--(2.885,8.195);
\draw[gp path] (11.506,8.195)--(11.326,8.195);
\node[gp node right] at (2.374,8.195) {$0.4$};
\draw[gp path] (2.705,1.772)--(2.705,1.952);
\draw[gp path] (2.705,8.195)--(2.705,8.015);
\node[gp node center] at (2.705,1.218) {$0$};
\draw[gp path] (3.805,1.772)--(3.805,1.952);
\draw[gp path] (3.805,8.195)--(3.805,8.015);
\node[gp node center] at (3.805,1.218) {$20$};
\draw[gp path] (4.905,1.772)--(4.905,1.952);
\draw[gp path] (4.905,8.195)--(4.905,8.015);
\node[gp node center] at (4.905,1.218) {$40$};
\draw[gp path] (6.005,1.772)--(6.005,1.952);
\draw[gp path] (6.005,8.195)--(6.005,8.015);
\node[gp node center] at (6.005,1.218) {$60$};
\draw[gp path] (7.106,1.772)--(7.106,1.952);
\draw[gp path] (7.106,8.195)--(7.106,8.015);
\node[gp node center] at (7.106,1.218) {$80$};
\draw[gp path] (8.206,1.772)--(8.206,1.952);
\draw[gp path] (8.206,8.195)--(8.206,8.015);
\node[gp node center] at (8.206,1.218) {$100$};
\draw[gp path] (9.306,1.772)--(9.306,1.952);
\draw[gp path] (9.306,8.195)--(9.306,8.015);
\node[gp node center] at (9.306,1.218) {$120$};
\draw[gp path] (10.406,1.772)--(10.406,1.952);
\draw[gp path] (10.406,8.195)--(10.406,8.015);
\node[gp node center] at (10.406,1.218) {$140$};
\draw[gp path] (11.506,1.772)--(11.506,1.952);
\draw[gp path] (11.506,8.195)--(11.506,8.015);
\node[gp node center] at (11.506,1.218) {$160$};
\draw[gp path] (2.705,8.195)--(2.705,1.772)--(11.506,1.772)--(11.506,8.195)--cycle;
\node[gp node center,rotate=-270] at (0.496,4.983) {Execution Time [s]};
\node[gp node center] at (7.105,0.387) {Number of Events [$\times 100$]};
\node[gp node right] at (9.009,7.738) {No Skip};
\gpcolor{rgb color={0.580,0.000,0.827}}
\gpsetlinewidth{3.00}
\draw[gp path] (9.340,7.738)--(10.844,7.738);
\draw[gp path] (3.512,2.286)--(4.266,2.896)--(5.233,3.586)--(5.917,4.060)--(6.721,4.694)%
  --(7.653,5.337)--(8.379,5.947)--(9.213,6.509)--(9.956,6.918)--(10.767,7.569);
\gpsetpointsize{8.00}
\gppoint{gp mark 7}{(3.512,2.286)}
\gppoint{gp mark 7}{(4.266,2.896)}
\gppoint{gp mark 7}{(5.233,3.586)}
\gppoint{gp mark 7}{(5.917,4.060)}
\gppoint{gp mark 7}{(6.721,4.694)}
\gppoint{gp mark 7}{(7.653,5.337)}
\gppoint{gp mark 7}{(8.379,5.947)}
\gppoint{gp mark 7}{(9.213,6.509)}
\gppoint{gp mark 7}{(9.956,6.918)}
\gppoint{gp mark 7}{(10.767,7.569)}
\gppoint{gp mark 7}{(10.092,7.738)}
\gpcolor{color=gp lt color border}
\node[gp node right] at (9.009,7.184) {Parametric Skip};
\gpcolor{rgb color={0.000,0.620,0.451}}
\draw[gp path] (9.340,7.184)--(10.844,7.184);
\draw[gp path] (3.512,2.575)--(4.266,3.065)--(5.233,3.715)--(5.917,4.197)--(6.721,4.831)%
  --(7.653,5.385)--(8.379,5.971)--(9.213,6.485)--(9.956,6.918)--(10.767,7.448);
\gppoint{gp mark 8}{(3.512,2.575)}
\gppoint{gp mark 8}{(4.266,3.065)}
\gppoint{gp mark 8}{(5.233,3.715)}
\gppoint{gp mark 8}{(5.917,4.197)}
\gppoint{gp mark 8}{(6.721,4.831)}
\gppoint{gp mark 8}{(7.653,5.385)}
\gppoint{gp mark 8}{(8.379,5.971)}
\gppoint{gp mark 8}{(9.213,6.485)}
\gppoint{gp mark 8}{(9.956,6.918)}
\gppoint{gp mark 8}{(10.767,7.448)}
\gppoint{gp mark 8}{(10.092,7.184)}
\gpcolor{color=gp lt color border}
\node[gp node right] at (9.009,6.630) {Non-Parametric Skip};
\gpcolor{rgb color={0.337,0.706,0.914}}
\draw[gp path] (9.340,6.630)--(10.844,6.630);
\draw[gp path] (3.512,2.414)--(4.266,2.912)--(5.233,3.570)--(5.917,4.028)--(6.721,4.590)%
  --(7.653,5.184)--(8.379,5.650)--(9.213,6.124)--(9.956,6.517)--(10.767,7.095);
\gppoint{gp mark 4}{(3.512,2.414)}
\gppoint{gp mark 4}{(4.266,2.912)}
\gppoint{gp mark 4}{(5.233,3.570)}
\gppoint{gp mark 4}{(5.917,4.028)}
\gppoint{gp mark 4}{(6.721,4.590)}
\gppoint{gp mark 4}{(7.653,5.184)}
\gppoint{gp mark 4}{(8.379,5.650)}
\gppoint{gp mark 4}{(9.213,6.124)}
\gppoint{gp mark 4}{(9.956,6.517)}
\gppoint{gp mark 4}{(10.767,7.095)}
\gppoint{gp mark 4}{(10.092,6.630)}
\gpcolor{color=gp lt color border}
\node[gp node right] at (9.009,6.076) {\monaa};
\gpcolor{rgb color={0.902,0.624,0.000}}
\draw[gp path] (9.340,6.076)--(10.844,6.076);
\draw[gp path] (3.512,1.772)--(4.266,1.772)--(5.233,1.933)--(5.917,1.933)--(6.721,1.933)%
  --(7.653,1.933)--(8.379,1.941)--(9.213,1.973)--(9.956,2.093)--(10.767,2.093);
\gppoint{gp mark 3}{(3.512,1.772)}
\gppoint{gp mark 3}{(4.266,1.772)}
\gppoint{gp mark 3}{(5.233,1.933)}
\gppoint{gp mark 3}{(5.917,1.933)}
\gppoint{gp mark 3}{(6.721,1.933)}
\gppoint{gp mark 3}{(7.653,1.933)}
\gppoint{gp mark 3}{(8.379,1.941)}
\gppoint{gp mark 3}{(9.213,1.973)}
\gppoint{gp mark 3}{(9.956,2.093)}
\gppoint{gp mark 3}{(10.767,2.093)}
\gppoint{gp mark 3}{(10.092,6.076)}
\gpcolor{color=gp lt color border}
\gpsetlinewidth{1.00}
\draw[gp path] (2.705,8.195)--(2.705,1.772)--(11.506,1.772)--(11.506,8.195)--cycle;
\gpdefrectangularnode{gp plot 1}{\pgfpoint{2.705cm}{1.772cm}}{\pgfpoint{11.506cm}{8.195cm}}
\end{tikzpicture}}
 \scalebox{0.5}{\begin{tikzpicture}[gnuplot]
\tikzset{every node/.append style={font={\fontsize{18.0pt}{21.6pt}\selectfont}}}
\path (0.000,0.000) rectangle (12.500,8.750);
\gpcolor{color=gp lt color border}
\gpsetlinetype{gp lt border}
\gpsetdashtype{gp dt solid}
\gpsetlinewidth{1.00}
\draw[gp path] (2.705,1.772)--(2.885,1.772);
\draw[gp path] (11.506,1.772)--(11.326,1.772);
\node[gp node right] at (2.374,1.772) {$0$};
\draw[gp path] (2.705,2.690)--(2.885,2.690);
\draw[gp path] (11.506,2.690)--(11.326,2.690);
\node[gp node right] at (2.374,2.690) {$0.05$};
\draw[gp path] (2.705,3.607)--(2.885,3.607);
\draw[gp path] (11.506,3.607)--(11.326,3.607);
\node[gp node right] at (2.374,3.607) {$0.1$};
\draw[gp path] (2.705,4.525)--(2.885,4.525);
\draw[gp path] (11.506,4.525)--(11.326,4.525);
\node[gp node right] at (2.374,4.525) {$0.15$};
\draw[gp path] (2.705,5.442)--(2.885,5.442);
\draw[gp path] (11.506,5.442)--(11.326,5.442);
\node[gp node right] at (2.374,5.442) {$0.2$};
\draw[gp path] (2.705,6.360)--(2.885,6.360);
\draw[gp path] (11.506,6.360)--(11.326,6.360);
\node[gp node right] at (2.374,6.360) {$0.25$};
\draw[gp path] (2.705,7.277)--(2.885,7.277);
\draw[gp path] (11.506,7.277)--(11.326,7.277);
\node[gp node right] at (2.374,7.277) {$0.3$};
\draw[gp path] (2.705,8.195)--(2.885,8.195);
\draw[gp path] (11.506,8.195)--(11.326,8.195);
\node[gp node right] at (2.374,8.195) {$0.35$};
\draw[gp path] (2.705,1.772)--(2.705,1.952);
\draw[gp path] (2.705,8.195)--(2.705,8.015);
\node[gp node center] at (2.705,1.218) {$0$};
\draw[gp path] (4.172,1.772)--(4.172,1.952);
\draw[gp path] (4.172,8.195)--(4.172,8.015);
\node[gp node center] at (4.172,1.218) {$50$};
\draw[gp path] (5.639,1.772)--(5.639,1.952);
\draw[gp path] (5.639,8.195)--(5.639,8.015);
\node[gp node center] at (5.639,1.218) {$100$};
\draw[gp path] (7.106,1.772)--(7.106,1.952);
\draw[gp path] (7.106,8.195)--(7.106,8.015);
\node[gp node center] at (7.106,1.218) {$150$};
\draw[gp path] (8.572,1.772)--(8.572,1.952);
\draw[gp path] (8.572,8.195)--(8.572,8.015);
\node[gp node center] at (8.572,1.218) {$200$};
\draw[gp path] (10.039,1.772)--(10.039,1.952);
\draw[gp path] (10.039,8.195)--(10.039,8.015);
\node[gp node center] at (10.039,1.218) {$250$};
\draw[gp path] (11.506,1.772)--(11.506,1.952);
\draw[gp path] (11.506,8.195)--(11.506,8.015);
\node[gp node center] at (11.506,1.218) {$300$};
\draw[gp path] (2.705,8.195)--(2.705,1.772)--(11.506,1.772)--(11.506,8.195)--cycle;
\node[gp node center,rotate=-270] at (0.496,4.983) {Execution Time [s]};
\node[gp node center] at (7.105,0.387) {Number of Events [$\times 100$]};
\node[gp node right] at (9.009,7.738) {No Skip};
\gpcolor{rgb color={0.580,0.000,0.827}}
\gpsetlinewidth{3.00}
\draw[gp path] (9.340,7.738)--(10.844,7.738);
\draw[gp path] (3.456,2.323)--(4.141,2.873)--(4.993,3.497)--(5.652,3.974)--(6.381,4.571)%
  --(7.216,5.195)--(7.894,5.699)--(8.660,6.250)--(9.362,6.764)--(10.079,7.369);
\gpsetpointsize{8.00}
\gppoint{gp mark 7}{(3.456,2.323)}
\gppoint{gp mark 7}{(4.141,2.873)}
\gppoint{gp mark 7}{(4.993,3.497)}
\gppoint{gp mark 7}{(5.652,3.974)}
\gppoint{gp mark 7}{(6.381,4.571)}
\gppoint{gp mark 7}{(7.216,5.195)}
\gppoint{gp mark 7}{(7.894,5.699)}
\gppoint{gp mark 7}{(8.660,6.250)}
\gppoint{gp mark 7}{(9.362,6.764)}
\gppoint{gp mark 7}{(10.079,7.369)}
\gppoint{gp mark 7}{(10.092,7.738)}
\gpcolor{color=gp lt color border}
\node[gp node right] at (9.009,7.184) {Parametric Skip};
\gpcolor{rgb color={0.000,0.620,0.451}}
\draw[gp path] (9.340,7.184)--(10.844,7.184);
\draw[gp path] (3.456,2.690)--(4.141,2.873)--(4.993,3.075)--(5.652,3.240)--(6.381,3.433)%
  --(7.216,3.625)--(7.894,3.791)--(8.660,3.983)--(9.362,4.167)--(10.079,4.387);
\gppoint{gp mark 8}{(3.456,2.690)}
\gppoint{gp mark 8}{(4.141,2.873)}
\gppoint{gp mark 8}{(4.993,3.075)}
\gppoint{gp mark 8}{(5.652,3.240)}
\gppoint{gp mark 8}{(6.381,3.433)}
\gppoint{gp mark 8}{(7.216,3.625)}
\gppoint{gp mark 8}{(7.894,3.791)}
\gppoint{gp mark 8}{(8.660,3.983)}
\gppoint{gp mark 8}{(9.362,4.167)}
\gppoint{gp mark 8}{(10.079,4.387)}
\gppoint{gp mark 8}{(10.092,7.184)}
\gpcolor{color=gp lt color border}
\node[gp node right] at (9.009,6.630) {Non-Parametric Skip};
\gpcolor{rgb color={0.337,0.706,0.914}}
\draw[gp path] (9.340,6.630)--(10.844,6.630);
\draw[gp path] (3.456,2.680)--(4.141,2.735)--(4.993,3.002)--(5.652,3.066)--(6.381,3.240)%
  --(7.216,3.424)--(7.894,3.598)--(8.660,3.726)--(9.362,3.809)--(10.079,4.011);
\gppoint{gp mark 4}{(3.456,2.680)}
\gppoint{gp mark 4}{(4.141,2.735)}
\gppoint{gp mark 4}{(4.993,3.002)}
\gppoint{gp mark 4}{(5.652,3.066)}
\gppoint{gp mark 4}{(6.381,3.240)}
\gppoint{gp mark 4}{(7.216,3.424)}
\gppoint{gp mark 4}{(7.894,3.598)}
\gppoint{gp mark 4}{(8.660,3.726)}
\gppoint{gp mark 4}{(9.362,3.809)}
\gppoint{gp mark 4}{(10.079,4.011)}
\gppoint{gp mark 4}{(10.092,6.630)}
\gpcolor{color=gp lt color border}
\node[gp node right] at (9.009,6.076) {\monaa};
\gpcolor{rgb color={0.902,0.624,0.000}}
\draw[gp path] (9.340,6.076)--(10.844,6.076);
\draw[gp path] (3.456,2.506)--(4.141,2.506)--(4.993,2.506)--(5.652,2.506)--(6.381,2.506)%
  --(7.216,2.506)--(7.894,2.506)--(8.660,2.506)--(9.362,2.515)--(10.079,2.534);
\gppoint{gp mark 3}{(3.456,2.506)}
\gppoint{gp mark 3}{(4.141,2.506)}
\gppoint{gp mark 3}{(4.993,2.506)}
\gppoint{gp mark 3}{(5.652,2.506)}
\gppoint{gp mark 3}{(6.381,2.506)}
\gppoint{gp mark 3}{(7.216,2.506)}
\gppoint{gp mark 3}{(7.894,2.506)}
\gppoint{gp mark 3}{(8.660,2.506)}
\gppoint{gp mark 3}{(9.362,2.515)}
\gppoint{gp mark 3}{(10.079,2.534)}
\gppoint{gp mark 3}{(10.092,6.076)}
\gpcolor{color=gp lt color border}
\gpsetlinewidth{1.00}
\draw[gp path] (2.705,8.195)--(2.705,1.772)--(11.506,1.772)--(11.506,8.195)--cycle;
\gpdefrectangularnode{gp plot 1}{\pgfpoint{2.705cm}{1.772cm}}{\pgfpoint{11.506cm}{8.195cm}}
\end{tikzpicture}}
 \scalebox{0.5}{\begin{tikzpicture}[gnuplot]
\tikzset{every node/.append style={font={\fontsize{18.0pt}{21.6pt}\selectfont}}}
\path (0.000,0.000) rectangle (12.500,8.750);
\gpcolor{color=gp lt color border}
\gpsetlinetype{gp lt border}
\gpsetdashtype{gp dt solid}
\gpsetlinewidth{1.00}
\draw[gp path] (2.374,1.772)--(2.554,1.772);
\draw[gp path] (11.506,1.772)--(11.326,1.772);
\node[gp node right] at (2.043,1.772) {$0.5$};
\draw[gp path] (2.374,2.414)--(2.554,2.414);
\draw[gp path] (11.506,2.414)--(11.326,2.414);
\node[gp node right] at (2.043,2.414) {$1$};
\draw[gp path] (2.374,3.057)--(2.554,3.057);
\draw[gp path] (11.506,3.057)--(11.326,3.057);
\node[gp node right] at (2.043,3.057) {$1.5$};
\draw[gp path] (2.374,3.699)--(2.554,3.699);
\draw[gp path] (11.506,3.699)--(11.326,3.699);
\node[gp node right] at (2.043,3.699) {$2$};
\draw[gp path] (2.374,4.341)--(2.554,4.341);
\draw[gp path] (11.506,4.341)--(11.326,4.341);
\node[gp node right] at (2.043,4.341) {$2.5$};
\draw[gp path] (2.374,4.984)--(2.554,4.984);
\draw[gp path] (11.506,4.984)--(11.326,4.984);
\node[gp node right] at (2.043,4.984) {$3$};
\draw[gp path] (2.374,5.626)--(2.554,5.626);
\draw[gp path] (11.506,5.626)--(11.326,5.626);
\node[gp node right] at (2.043,5.626) {$3.5$};
\draw[gp path] (2.374,6.268)--(2.554,6.268);
\draw[gp path] (11.506,6.268)--(11.326,6.268);
\node[gp node right] at (2.043,6.268) {$4$};
\draw[gp path] (2.374,6.910)--(2.554,6.910);
\draw[gp path] (11.506,6.910)--(11.326,6.910);
\node[gp node right] at (2.043,6.910) {$4.5$};
\draw[gp path] (2.374,7.553)--(2.554,7.553);
\draw[gp path] (11.506,7.553)--(11.326,7.553);
\node[gp node right] at (2.043,7.553) {$5$};
\draw[gp path] (2.374,8.195)--(2.554,8.195);
\draw[gp path] (11.506,8.195)--(11.326,8.195);
\node[gp node right] at (2.043,8.195) {$5.5$};
\draw[gp path] (2.374,1.772)--(2.374,1.952);
\draw[gp path] (2.374,8.195)--(2.374,8.015);
\node[gp node center] at (2.374,1.218) {$20$};
\draw[gp path] (3.516,1.772)--(3.516,1.952);
\draw[gp path] (3.516,8.195)--(3.516,8.015);
\node[gp node center] at (3.516,1.218) {$30$};
\draw[gp path] (4.657,1.772)--(4.657,1.952);
\draw[gp path] (4.657,8.195)--(4.657,8.015);
\node[gp node center] at (4.657,1.218) {$40$};
\draw[gp path] (5.799,1.772)--(5.799,1.952);
\draw[gp path] (5.799,8.195)--(5.799,8.015);
\node[gp node center] at (5.799,1.218) {$50$};
\draw[gp path] (6.940,1.772)--(6.940,1.952);
\draw[gp path] (6.940,8.195)--(6.940,8.015);
\node[gp node center] at (6.940,1.218) {$60$};
\draw[gp path] (8.082,1.772)--(8.082,1.952);
\draw[gp path] (8.082,8.195)--(8.082,8.015);
\node[gp node center] at (8.082,1.218) {$70$};
\draw[gp path] (9.223,1.772)--(9.223,1.952);
\draw[gp path] (9.223,8.195)--(9.223,8.015);
\node[gp node center] at (9.223,1.218) {$80$};
\draw[gp path] (10.365,1.772)--(10.365,1.952);
\draw[gp path] (10.365,8.195)--(10.365,8.015);
\node[gp node center] at (10.365,1.218) {$90$};
\draw[gp path] (11.506,1.772)--(11.506,1.952);
\draw[gp path] (11.506,8.195)--(11.506,8.015);
\node[gp node center] at (11.506,1.218) {$100$};
\draw[gp path] (2.374,8.195)--(2.374,1.772)--(11.506,1.772)--(11.506,8.195)--cycle;
\node[gp node center,rotate=-270] at (0.496,4.983) {Execution Time [s]};
\node[gp node center] at (6.940,0.387) {Number of Events [$\times 100$]};
\node[gp node right] at (9.009,7.738) {No Skip};
\gpcolor{rgb color={0.580,0.000,0.827}}
\gpsetlinewidth{3.00}
\draw[gp path] (9.340,7.738)--(10.844,7.738);
\draw[gp path] (2.374,2.138)--(4.657,3.150)--(6.940,4.176)--(9.223,5.174)--(11.506,6.205);
\gpsetpointsize{8.00}
\gppoint{gp mark 7}{(2.374,2.138)}
\gppoint{gp mark 7}{(4.657,3.150)}
\gppoint{gp mark 7}{(6.940,4.176)}
\gppoint{gp mark 7}{(9.223,5.174)}
\gppoint{gp mark 7}{(11.506,6.205)}
\gppoint{gp mark 7}{(10.092,7.738)}
\gpcolor{color=gp lt color border}
\node[gp node right] at (9.009,7.184) {Parametric Skip};
\gpcolor{rgb color={0.000,0.620,0.451}}
\draw[gp path] (9.340,7.184)--(10.844,7.184);
\draw[gp path] (2.374,3.795)--(4.657,4.807)--(6.940,5.835)--(9.223,6.850)--(11.506,7.879);
\gppoint{gp mark 8}{(2.374,3.795)}
\gppoint{gp mark 8}{(4.657,4.807)}
\gppoint{gp mark 8}{(6.940,5.835)}
\gppoint{gp mark 8}{(9.223,6.850)}
\gppoint{gp mark 8}{(11.506,7.879)}
\gppoint{gp mark 8}{(10.092,7.184)}
\gpcolor{color=gp lt color border}
\node[gp node right] at (9.009,6.630) {Non-Parametric Skip};
\gpcolor{rgb color={0.337,0.706,0.914}}
\draw[gp path] (9.340,6.630)--(10.844,6.630);
\draw[gp path] (2.374,3.767)--(4.657,4.790)--(6.940,5.795)--(9.223,6.819)--(11.506,7.853);
\gppoint{gp mark 4}{(2.374,3.767)}
\gppoint{gp mark 4}{(4.657,4.790)}
\gppoint{gp mark 4}{(6.940,5.795)}
\gppoint{gp mark 4}{(9.223,6.819)}
\gppoint{gp mark 4}{(11.506,7.853)}
\gppoint{gp mark 4}{(10.092,6.630)}
\gpcolor{color=gp lt color border}
\node[gp node right] at (9.009,6.076) {\monaa};
\gpcolor{rgb color={0.902,0.624,0.000}}
\draw[gp path] (9.340,6.076)--(10.844,6.076);
\draw[gp path] (2.374,3.136)--(4.657,3.145)--(6.940,3.156)--(9.223,3.166)--(11.506,3.192);
\gppoint{gp mark 3}{(2.374,3.136)}
\gppoint{gp mark 3}{(4.657,3.145)}
\gppoint{gp mark 3}{(6.940,3.156)}
\gppoint{gp mark 3}{(9.223,3.166)}
\gppoint{gp mark 3}{(11.506,3.192)}
\gppoint{gp mark 3}{(10.092,6.076)}
\gpcolor{color=gp lt color border}
\gpsetlinewidth{1.00}
\draw[gp path] (2.374,8.195)--(2.374,1.772)--(11.506,1.772)--(11.506,8.195)--cycle;
\gpdefrectangularnode{gp plot 1}{\pgfpoint{2.374cm}{1.772cm}}{\pgfpoint{11.506cm}{8.195cm}}
\end{tikzpicture}}
 \scalebox{0.5}{\begin{tikzpicture}[gnuplot]
\tikzset{every node/.append style={font={\fontsize{18.0pt}{21.6pt}\selectfont}}}
\path (0.000,0.000) rectangle (12.500,8.750);
\gpcolor{color=gp lt color border}
\gpsetlinetype{gp lt border}
\gpsetdashtype{gp dt solid}
\gpsetlinewidth{1.00}
\draw[gp path] (2.374,1.772)--(2.554,1.772);
\draw[gp path] (11.506,1.772)--(11.326,1.772);
\node[gp node right] at (2.043,1.772) {$0$};
\draw[gp path] (2.374,2.843)--(2.554,2.843);
\draw[gp path] (11.506,2.843)--(11.326,2.843);
\node[gp node right] at (2.043,2.843) {$0.1$};
\draw[gp path] (2.374,3.913)--(2.554,3.913);
\draw[gp path] (11.506,3.913)--(11.326,3.913);
\node[gp node right] at (2.043,3.913) {$0.2$};
\draw[gp path] (2.374,4.984)--(2.554,4.984);
\draw[gp path] (11.506,4.984)--(11.326,4.984);
\node[gp node right] at (2.043,4.984) {$0.3$};
\draw[gp path] (2.374,6.054)--(2.554,6.054);
\draw[gp path] (11.506,6.054)--(11.326,6.054);
\node[gp node right] at (2.043,6.054) {$0.4$};
\draw[gp path] (2.374,7.124)--(2.554,7.124);
\draw[gp path] (11.506,7.124)--(11.326,7.124);
\node[gp node right] at (2.043,7.124) {$0.5$};
\draw[gp path] (2.374,8.195)--(2.554,8.195);
\draw[gp path] (11.506,8.195)--(11.326,8.195);
\node[gp node right] at (2.043,8.195) {$0.6$};
\draw[gp path] (2.374,1.772)--(2.374,1.952);
\draw[gp path] (2.374,8.195)--(2.374,8.015);
\node[gp node center] at (2.374,1.218) {$10$};
\draw[gp path] (4.200,1.772)--(4.200,1.952);
\draw[gp path] (4.200,8.195)--(4.200,8.015);
\node[gp node center] at (4.200,1.218) {$20$};
\draw[gp path] (6.027,1.772)--(6.027,1.952);
\draw[gp path] (6.027,8.195)--(6.027,8.015);
\node[gp node center] at (6.027,1.218) {$30$};
\draw[gp path] (7.853,1.772)--(7.853,1.952);
\draw[gp path] (7.853,8.195)--(7.853,8.015);
\node[gp node center] at (7.853,1.218) {$40$};
\draw[gp path] (9.680,1.772)--(9.680,1.952);
\draw[gp path] (9.680,8.195)--(9.680,8.015);
\node[gp node center] at (9.680,1.218) {$50$};
\draw[gp path] (11.506,1.772)--(11.506,1.952);
\draw[gp path] (11.506,8.195)--(11.506,8.015);
\node[gp node center] at (11.506,1.218) {$60$};
\draw[gp path] (2.374,8.195)--(2.374,1.772)--(11.506,1.772)--(11.506,8.195)--cycle;
\node[gp node center,rotate=-270] at (0.496,4.983) {Execution Time [s]};
\node[gp node center] at (6.940,0.387) {Number of Events [$\times 100$]};
\node[gp node right] at (9.009,7.738) {No Skip};
\gpcolor{rgb color={0.580,0.000,0.827}}
\gpsetlinewidth{3.00}
\draw[gp path] (9.340,7.738)--(10.844,7.738);
\draw[gp path] (2.374,2.639)--(4.200,3.592)--(6.027,4.496)--(7.853,5.417)--(9.680,6.338)%
  --(11.506,7.258);
\gpsetpointsize{8.00}
\gppoint{gp mark 7}{(2.374,2.639)}
\gppoint{gp mark 7}{(4.200,3.592)}
\gppoint{gp mark 7}{(6.027,4.496)}
\gppoint{gp mark 7}{(7.853,5.417)}
\gppoint{gp mark 7}{(9.680,6.338)}
\gppoint{gp mark 7}{(11.506,7.258)}
\gppoint{gp mark 7}{(10.092,7.738)}
\gpcolor{color=gp lt color border}
\node[gp node right] at (9.009,7.184) {Parametric Skip};
\gpcolor{rgb color={0.000,0.620,0.451}}
\draw[gp path] (9.340,7.184)--(10.844,7.184);
\draw[gp path] (2.374,2.735)--(4.200,3.485)--(6.027,4.234)--(7.853,4.989)--(9.680,5.717)%
  --(11.506,6.488);
\gppoint{gp mark 8}{(2.374,2.735)}
\gppoint{gp mark 8}{(4.200,3.485)}
\gppoint{gp mark 8}{(6.027,4.234)}
\gppoint{gp mark 8}{(7.853,4.989)}
\gppoint{gp mark 8}{(9.680,5.717)}
\gppoint{gp mark 8}{(11.506,6.488)}
\gppoint{gp mark 8}{(10.092,7.184)}
\gpcolor{color=gp lt color border}
\node[gp node right] at (9.009,6.630) {Non-Parametric Skip};
\gpcolor{rgb color={0.337,0.706,0.914}}
\draw[gp path] (9.340,6.630)--(10.844,6.630);
\draw[gp path] (2.374,2.735)--(4.200,3.485)--(6.027,4.218)--(7.853,4.919)--(9.680,5.658)%
  --(11.506,6.386);
\gppoint{gp mark 4}{(2.374,2.735)}
\gppoint{gp mark 4}{(4.200,3.485)}
\gppoint{gp mark 4}{(6.027,4.218)}
\gppoint{gp mark 4}{(7.853,4.919)}
\gppoint{gp mark 4}{(9.680,5.658)}
\gppoint{gp mark 4}{(11.506,6.386)}
\gppoint{gp mark 4}{(10.092,6.630)}
\gpcolor{color=gp lt color border}
\node[gp node right] at (9.009,6.076) {\monaa};
\gpcolor{rgb color={0.902,0.624,0.000}}
\draw[gp path] (9.340,6.076)--(10.844,6.076);
\draw[gp path] (2.374,1.772)--(4.200,1.879)--(6.027,1.879)--(7.853,1.879)--(9.680,1.879)%
  --(11.506,1.879);
\gppoint{gp mark 3}{(2.374,1.772)}
\gppoint{gp mark 3}{(4.200,1.879)}
\gppoint{gp mark 3}{(6.027,1.879)}
\gppoint{gp mark 3}{(7.853,1.879)}
\gppoint{gp mark 3}{(9.680,1.879)}
\gppoint{gp mark 3}{(11.506,1.879)}
\gppoint{gp mark 3}{(10.092,6.076)}
\gpcolor{color=gp lt color border}
\gpsetlinewidth{1.00}
\draw[gp path] (2.374,8.195)--(2.374,1.772)--(11.506,1.772)--(11.506,8.195)--cycle;
\gpdefrectangularnode{gp plot 1}{\pgfpoint{2.374cm}{1.772cm}}{\pgfpoint{11.506cm}{8.195cm}}
\end{tikzpicture}}
 \caption{Execution time for the benchmarks without parameters: \textsc{Gear-np} (above left), \textsc{Accel-np} (above right), \textsc{Blowup-np} (below left), and \textsc{OnlyTiming-np} (below right)}
 \label{figure:np-result}
\end{figure}

To answer RQ2, we ran \pmonaa using the non-parametric benchmarks (\textsc{Accel-np}, \textsc{Gear-np}, \textsc{Blowup-np}, and \textsc{OnlyTiming-np}) and compared the execution time with a tool \monaa~\cite{WHS18} for non-parametric timed pattern matching. 

In \cref{figure:np-result}, we observe that our algorithms are slower than \monaa by orders of magnitude even though we solve the same problem (non-parametric timed pattern matching).
This is presumably because our implementations rely on Parma Polyhedra Library (PPL)~\cite{BHZ08} to compute symbolic states, while \monaa{} utilizes DBMs (difference bound matrices)~\cite{Dill89}.
It was shown in~\cite{BFMU17} that polyhedra may be dozens of times slower than DBMs; however, for parametric analyses, DBMs are not suitable, and parameterized extensions (\eg{} in~\cite{HRSV02}) still need polyhedra in their representation.

Moreover, in \cref{figure:normal-result,figure:np-result}, we observe that the execution time of our algorithms are not much different between each parametric benchmark and its non-parametric variant except \textsc{Blowup}. 
This observation shows that at least one additional parameter does not make the problem too difficult.

Therefore, we conclude that the lower efficiency of parametric timed pattern matching is mainly because of its general data structure required by the general problem setting.

\subsubsection{RQ3: Overhead of skip value computation}
\label{subsec:overhead_skip_values}

\begin{table}[tb]
 \centering
 \caption{Execution time [s] for the skip value computation}
 \label{table:overhead-skipping}
\begin{tabular}{|c|c|c|c|}
\hline
\cellHeaderColor{}&\cellHeader{Non-Parametric}&\cellHeader{Parametric}&\cellHeaderColor{}\\
 \cellHeaderColor{}&\cellHeader{Skip}&\cellHeader{Skip}&\multirow{-2}{*}{\cellHeader{\monaa{}}}\\\hline
\textsc{Gear}&0.0115&0.0175&n/a\\\hline
\textsc{Gear-np}&0.01&0.01&$ < 0.01$\\\hline
\textsc{Accel}&0.042&0.0435&n/a\\\hline
\textsc{Accel-np}&0.04&0.04&0.0305\\\hline
\textsc{OnlyTiming}&0.03&0.03&n/a\\\hline
\textsc{OnlyTiming-np}&0.02&0.02&$ < 0.01$\\\hline
\textsc{Blowup}&0.3665&0.381&n/a\\\hline
\textsc{Blowup-np}&1.268&1.2905&1.5455\\\hline
\end{tabular}
\end{table}

To answer RQ3, we compared the execution time of our algorithms for an empty timed word using all the benchmarks.
As a baseline, we also measured the execution time of \monaa.

In \cref{table:overhead-skipping}, we observe that the execution time for the skip values is less than 0.05 second except for \textsc{Blowup} and \textsc{Blowup-np}.
Even for the slowest pattern PTA \textsc{Blowup-np}, the execution time for the skip values is less than 1.5 second and it is faster than that of \monaa. We conclude that the overhead of the skip value computation is small and acceptable in many usage scenarios.

\section{Comparison Between Two Approaches}\label{section:comparison}
In \cref{section:PTMC,section:adhoc}, we presented two approaches for parametric timed pattern matching.
In \cref{section:PTMC}, we reduced parametric timed pattern matching to parametric timed model checking of PTAs (\cref{algo:PTPM}), using an existing parametric timed model checker.
Parametric timed model checking relies on the reachability analysis with symbolic abstraction by convex polyhedra.
In \cref{section:adhoc}, we solved parametric timed pattern matching by following the transitions of the PTA (\cref{alg:online_no_skip}).
Moreover, we optimized \cref{alg:online_no_skip} using \emph{skipping} (\cref{alg:online_with_skip}), which is a technique from string matching.
Although the relationship between parametric timed pattern matching and parametric timed model checking is implicit in \cref{section:adhoc}, \cref{alg:online_no_skip,alg:online_with_skip} utilize convex polyhedra to compute the \emph{reachable} concrete states.
Moreover, both \imitator{} and \pmonaa{} rely on the Parma Polyhedra Library (PPL)~\cite{BHZ08} for convex polyhedra operations.
In this sense, the underlying technique utilized in \cref{section:PTMC} is also used in \cref{section:adhoc}.

For the performance, \cref{ss:adhoc-experiments} shows that \cref{alg:online_no_skip,alg:online_with_skip} are much more efficient than \cref{algo:PTPM}. This efficiency is thanks to the skipping and its direct implementation rather than an indirect approach via parametric timed model checking.

\begin{figure}[tbp]
 \centering
 \begin{tikzpicture}[scale=1.8,xscale=1.5]
   \draw [thick, -stealth](-0.5,0)--(4.6,0) node [anchor=north]{$t$};
   \draw (0,0.1) -- (0,-0.1) node [anchor=north]{$0$};

   \draw (0.35,0.1) node[anchor=south]{$\styleact{BobEnter}$} -- (0.35,-0.1);
   \draw (1.0,0.1) node[anchor=south]{$\styleact{AliceEnter}$} -- (1.0,-0.1);
   \draw (1.75,0.1) node[anchor=south]{$\styleact{BobLeft}$} -- (1.75,-0.1);
   \draw (2.35,0.1) node[anchor=south]{$\styleact{AliceLeft}$} -- (2.35,-0.1);
   \draw (3.0,0.1) node[anchor=south]{$\styleparam{X}\styleact{Enter}$} -- (3.0,-0.1);
   \draw (3.5,0.1) node[anchor=south]{$\styleparam{Y}\styleact{Enter}$} -- (3.5,-0.1);
   \draw (4.35,0.1) node[anchor=south]{$\styleact{AliceLeft}$} -- (4.35,-0.1);
 \end{tikzpicture}
 \caption{A log of entrance and leaving from a building. Timestamps are omitted for simplicity. We usually know who entered or left the building (\eg{} $\styleact{BobEnter}$) but we sometimes do not know who (\eg{} $\styleparam{X}\styleact{Enter}$).}
 \label{fig:example_log_enterance_left}
\end{figure}
\begin{figure}[tbp]
 \centering
 \scalebox{0.95}{
 \begin{tikzpicture}[shorten >=1pt,node distance=2.0cm,on grid,auto]
  \node[location,initial] (w0) {$\wloc_0$};
  \node[location] (w1) [right of=w0] {$\wloc_1$};
  \node[location] (w2) [right of=w1] {$\wloc_2$};
  \node[location] (w3) [right of=w2]{$\wloc_3$};
  \node[location] (w4) [right of=w3]{$\wloc_4$};
  \node[location] (w5-1) [above right of=w4]{$\wloc_{5,1}$};
  \node[location] (w5-2) [below right of=w4]{$\wloc_{5,2}$};
  \node[location] (w6) [below right of=w5-1]{$\wloc_6$};
  \node[location] (w7) [right of=w6]{$\wloc_7$};

  \path[->]
  (w0) edge [above] node[align=center] {$\styleact{BobEnter}$} (w1)
  (w1) edge [above] node[align=center] {$\styleact{AliceEnter}$} (w2)
  (w2) edge [above] node[align=center] {$\styleact{BobLeft}$} (w3)
  (w3) edge [above] node[align=center] {$\styleact{AliceLeft}$} (w4)
  (w4) edge [above left] node[align=center] {$\styleact{BobEnter}$} (w5-1)
  (w4) edge [below left] node[align=center] {$\styleact{AliceEnter}$} (w5-2)
  (w5-1) edge [above right] node[align=center] {$\styleact{AliceEnter}$} (w6)
  (w5-2) edge [below right] node[align=center] {$\styleact{BobEnter}$} (w6)
  (w6) edge [above] node[align=center] {$\styleact{AliceLeft}$} (w7)
  ;
 \end{tikzpicture}}
 \caption{The TA constructed from the log in \cref{fig:example_log_enterance_left} using \TransWord{} in \cref{sss:word_conversion}. We assume $\styleparam{X},\styleparam{Y} \in \{\styleact{Alice},\styleact{Bob}\}$ and $\styleparam{X} \neq \styleparam{Y}$. The timing constraints are omitted for simplicity.}
 \label{fig:example_ta_enterance_left}
\end{figure}
Therefore, someone might consider that the approach in \cref{section:adhoc} is strictly superior to the one in \cref{section:PTMC}.
However, we note that the model checking-based framework in \cref{section:PTMC} is generic, and it is robust to the modification of the problem definition.
For example, consider the ill-shaped log shown in \cref{fig:example_log_enterance_left}: we take a log of the entrance and leaving from a building with the name and the timestamp, but we sometimes fail to identify who passed the gate \eg{} due to a sensor error.
Since a timed word $\word$ is a sequence of pairs $(a_i, \tau_i)$ of an event $a_i$ and a timestamp $\tau_i$, it cannot directly represent the log in \cref{fig:example_log_enterance_left}, and therefore, it is not straightforward to handle such a log using \cref{alg:online_no_skip,alg:online_with_skip}.
However, it is easy to generalize the conversion \TransWord{} from a timed word to a (P)TA in \cref{sss:word_conversion} to a log with branching. An example of the result is shown in \cref{fig:example_ta_enterance_left}. Thus, it is rather straightforward to monitor such an ill-shaped log using \cref{algo:PTPM}.
Moreover, the framework in \cref{section:PTMC} has been used for a different monitoring problem following the conference version of this work~\cite{AHW18}, \ie{} model-bounded monitoring in~\cite{WAH21}.

In summary, \cref{alg:online_no_skip,alg:online_with_skip} are better at solving the parametric timed pattern matching problem, but the construction used in \cref{algo:PTPM} is general, and it is potentially applicable to extensions of the problem considered here.

\section{Conclusion and perspectives}\label{section:conclusion}

\paragraph{Conclusion}
We proposed two approaches to perform timed pattern matching in the presence of an uncertain specification.
This allows us to synthesize parameter valuations and intervals for which the specification holds on an input timed word.
Our implementation using \imitator{} may not completely allow for \emph{online} timed pattern matching yet, but already gives an interesting feedback in terms of parametric monitoring.
Our second algorithm aiming at finding minimal or maximal parameter valuations is less sensitive to state space explosion.
While our algorithms should be further optimized, we believe they pave the way for a more precise monitoring of real-time systems with an output richer than just timed intervals.

In a second part, our dedicated method dramatically outperforms the previous approach using parametric timed model checking.
In addition, we discussed an optimization using skipping, that brings an interesting speedup.

\paragraph{Perspectives}

Exploiting the polarity of parameters, as done in~\cite{ADMN11} or in lower-bound/upper-bound parametric timed automata~\cite{HRSV02,ALime17}, may help to improve the efficiency of \PTPMopt{}.

In addition, natural future works include more expressive specifications than (parametric) timed automata-based specifications, \eg{} using more expressive logics such as~\cite{BKMZ15}.

Another challenge is the interpretation (and the visualization) of the results of parametric timed pattern matching.
While the result of $\PTPMopt$ is natural, the fully symbolic result of~$\PTPM$ remains a challenge to be interpreted; for example, the 125,250 matches for \textsc{Blowup} means the union of 125,250 polyhedra in 5~dimensions.
We give a possible way to visualize such results in \cref{figure:projections:gear} for \textsc{Gear} ($|\word| = 1,467$): in particular, observe in \cref{figure:projections:gear:t-param} that a single point exceeds~3, only a few exceed~2, while the wide majority remain in~$[0,1]$.
This helps to visualize how fast the gear is changed from~1 to~2, and at what timestamps.

\ifdefined\VersionAuthor
	\subsection*{Acknowledgments}
\else
	\begin{acks}
\fi
	This work is partially supported
	by
	JST ERATO HASUO Metamathematics for Systems Design Project (No.\ JPMJER1603),
	by JST ACT-X Grant No.\ JPMJAX200U, 
        by JSPS Grants-in-Aid No.\ 15KT0012 \& 18J22498,
        by JST CEREST Grant No.\ JPMJCR2012,
      by the ANR national research program PACS (ANR-14-CE28-0002)
      and
      by the ANR-NRF French-Singaporean research program ProMiS (ANR-19-CE25-0015).
\ifdefined\VersionAuthor
\else
	\end{acks}
\fi

\FinalVersion{
\newpage
}

	\newcommand{\CCIS}{Communications in Computer and Information Science}
	\newcommand{\ENTCS}{Electronic Notes in Theoretical Computer Science}
	\newcommand{\FAC}{Formal Aspects of Computing}
	\newcommand{\FI}{Fundamenta Informaticae}
	\newcommand{\FMSD}{Formal Methods in System Design}
	\newcommand{\IJFCS}{International Journal of Foundations of Computer Science}
	\newcommand{\IJSSE}{International Journal of Secure Software Engineering}
	\newcommand{\IPL}{Information Processing Letters}
	\newcommand{\JAIR}{Journal of Artificial Intelligence Research}
	\newcommand{\JLAP}{Journal of Logic and Algebraic Programming}
	\newcommand{\JLAMP}{Journal of Logical and Algebraic Methods in Programming} %
	\newcommand{\JLC}{Journal of Logic and Computation}
	\newcommand{\LMCS}{Logical Methods in Computer Science}
	\newcommand{\LNCS}{Lecture Notes in Computer Science}
	\newcommand{\RESS}{Reliability Engineering \& System Safety}
	\newcommand{\STTT}{International Journal on Software Tools for Technology Transfer}
	\newcommand{\TCS}{Theoretical Computer Science}
	\newcommand{\ToPNoC}{Transactions on Petri Nets and Other Models of Concurrency}
	\newcommand{\TSE}{{IEEE} Transactions on Software Engineering}

\ifdefined\VersionAuthor
	\renewcommand*{\bibfont}{\small}
	\printbibliography[title={References}]
\else
	\bibliographystyle{ACM-Reference-Format} %
	\bibliography{PTPM}
\fi

\LongVersion{
\newpage
\appendix

\section{Construction of $\pvalsi{\loc,n}$}
\label{appendix:automatic_construction}

We present a construction of $\pvalsi{\loc,n}$ in~\cref{def:KMP}.
We fix a PTA $\A=(\Actions, \Loc, \locinit, \LocFinal, \Clock, \Param, %
 \Edges)$, a location $\loc\in\Loc$, and $n\in\Zp$.
Since $\pvalsi{\loc,n}$ is the set of parameter valuations $\pval\in\PVal$ such that there is a $\pval'\in\PVal$ satisfying $\Lg(\valuate{\A_\loc}{\pval})\cdot\TW(\Actions)\cap\TW^n(\Actions)\cdot\{\word''+t\mid\word''\in\Lg_{-\$}(\pval'(\A)),t>0\}\cdot\TW(\Actions)\ne\emptyset$, 
we construct PTAs $\A'_{\loc}$ and $\A'_{+n}$ satisfying $\Lg(\valuate{\A'_{\loc}}{\pval}) = \Lg(\valuate{\A_{\loc}}{\pval})\cdot\TW(\Actions)$ and $\Lg(\valuate{\A'_{+n}}{\pval'}) = \TW^n(\Actions)\cdot\{\word + t \mid\word\in\Lg_{-\$}(\valuate{\A}{\pval'}),t>0\}\cdot\TW(\Actions)$.

We define $\A'_\loc$ as
$\A_\loc=(\Actions, \Loc\disjointUnion\{\loc_{\mathrm{fin}}\}, \locinit, \{\loc, \loc_{\mathrm{fin}}\}, \Clock, \Param, %
 \Edges'_\loc)$, where $\Edges'_{\loc}=\Edges\cup\{(\loc, \top,\action,\emptyset,\loc_{\mathrm{fin}})\mid\action\in\Actions\}\cup\{(\loc_{\mathrm{fin}},\top,\action,\emptyset,\loc_{\mathrm{fin}})\mid\action\in\Actions\}$.
For any $\pval\in\PVal$, $\word'\in\Lg(\valuate{\A_{\loc}}{\pval})$, and $\word''\in\TW(\Actions)$, we have
$\locinit\xrightarrow{\word'} \loc \xrightarrow{\word''}\loc_{\mathrm{fin}}$ in $\valuate{\A'_\loc}{\pval}$ and $\word'\cdot\word''\in\Lg(\valuate{\A'_{\loc}}{\pval})$.
For any $\pval\in\PVal$ and $\word\in\Lg(\valuate{\A'_{\loc}}{\pval})$, there exist timed words $\word',\word''\in\TW(\Actions)$ satisfying $\word=\word'\cdot\word''$ and
$\locinit\xrightarrow{\word'} \loc $ in $\valuate{\A'_{\loc}}{\pval}$, which implies $\word'\in\Lg(\valuate{\A_{\loc}}{\pval})$.
Therefore, we have $\Lg(\valuate{\A'_{\loc}}{\pval}) = \Lg(\valuate{\A_{\loc}}{\pval})\cdot\TW(\Actions)$.

We define $\A'_{+n}$ as
$\A_{+n} = (\Actions\disjointUnion\{\varepsilon\}, \Loc', \loc_{n+1}, \LocFinal', \Clock, \Param, %
 \Edges')$, where
\begin{itemize}
 \item $\varepsilon$ is the unobservable character;
 \item $\Loc'=\Loc\disjointUnion\{\loc_{\mathrm{i}}\mid i \in \{1,2,\dots,n+1\}\}\disjointUnion\{\loc_{\mathrm{fin}}\}$;
 \item $\LocFinal' = \{\loc\mid \exists \loc'\in\LocFinal.\, (\loc, \guard,\action,\resets,\loc')\in\Edges\}\disjointUnion\{\loc_{\mathrm{fin}}\}$; and
 \item $\Edges'=\Edges \disjointUnion \{(\loc_{i+1},\top,a,\emptyset,\loc_{i}) \mid a \in \Actions,i\in\{1,2,\dots,n\}\} \disjointUnion \{(l_1,\top,\varepsilon,\Clock,\locinit)\}\disjointUnion\{(\loc,\top,a,\emptyset,\loc_{\mathrm{fin}}) \mid a \in \Actions,\loc\in\LocFinal'\}\disjointUnion\{(\loc_{\mathrm{fin}},\top,a,\emptyset,\loc_{\mathrm{fin}}) \mid a \in \Actions\}$.
\end{itemize}

For any parameter valuation $\pval'\in\PVal$, timed words $\word'\in\TW^n(\Actions)$, $\word''\in\Lg(\valuate{\A}{\pval'})$, $\word'''\in\TW(\Actions)$, and $t\in\Rp$, we have $\loc_{n+1}\xrightarrow{\word'} \loc_1\xrightarrow{(\varepsilon,t)}\locinit\xrightarrow{\word''}\loc_f\xrightarrow{\word'''}\loc_{\mathrm{fin}}$ in $\valuate{\A}{\pval'}$, where $\loc_f\in\LocFinal'$.
For any parameter valuation $\pval'\in\PVal$ and for any timed word $\word\in\Lg(\valuate{\A'_{+n}}{\pval'})$, there exist $\word'\in\TW(\Actions)$, $\word''\in\TW(\Actions)$, $\word'''\in\TW(\Actions)$, and $t\in\Rp$ satisfying $\word=\word'\cdot (\varepsilon,t)\cdot\word''\cdot\word'''$ and $\locinit\xrightarrow{\word''}\loc_f$ in $\valuate{\A'_{+n}}{\pval}$, where $\loc_f\in\LocFinal'$, and therefore, we have $\word''\in\Lg_{-\$}(\valuate{\A}{\pval'})$.
Overall, for any $\pval\in\PVal$, we have
$\Lg(\valuate{\A'_{+n}}{\pval'}) = \TW^n(\Actions)\cdot\{\word + t \mid\word\in\Lg_{-\$}(\valuate{\A}{\pval'}),t>0\}\cdot\TW(\Actions)$.

To take the intersection of $\Lg(\valuate{\A'_{\loc}}{\pval})$ and $\Lg(\valuate{\A'_{+n}}{\pval'})$, we use the synchronous product.

For any $\pval,\pval'\in\PVal$, we have $\Lg(\valuate{\A'_{\loc}\parallel\A'_{+n}}{(\pval\disjointUnion\pval')}) = \Lg(\valuate{\A'_{\loc}}{\pval})\cap\Lg(\valuate{\A'_{+n}}{\pval'})$. Therefore, we have $\pvalsi{\loc,n} = \{\pval\mid\exists\pval'\in\PVal.\,\Lg(\valuate{\A'_{\loc}\parallel\A'_{+n}}{(\pval\disjointUnion\pval')})\ne\emptyset\}$, which can be computed by reachability synthesis of PTAs.
}

\end{document}